%% file: NewDraft7_Sensor.tex
\documentclass[11pt]{article}
\usepackage[toc,page,header]{appendix}

\include{header}
\usepackage[margin=1in]{geometry}
\usepackage{natbib,libertine}
\usepackage{siunitx}

\usepackage[disable]{todonotes}
\usepackage{titlesec}
\usepackage[font=small,labelfont=bf]{caption}
\usepackage{nameref}

\newcommand{\blind}{1}
\usepackage{minitoc}

\begin{document}
\doparttoc 
\faketableofcontents 



\if1\blind
  {
  \title{A semi-parametric model for target localization in distributed systems}
 \author{Rohit K. Patra\thanks{Corresponding author. E-mail: rohitpatra@ufl.edu}\hspace{.2cm}\\
      {University of Florida}\\ Moulinath Banerjee\\
      {University of Michigan}\\
       George Michailidis\\
      {University of Florida}}
    \maketitle
  } \fi
  
  \if0\blind
  {    \title{A semi-parametric model for target localization in distributed systems}
        \maketitle
  } \fi

\begin{abstract}
Distributed systems serve as a key technological infrastructure for monitoring diverse systems across space and time. Examples of their widespread applications include: precision agriculture, surveillance, ecosystem and physical infrastructure monitoring, animal behavior and tracking, disaster response and recovery to name a few. Such systems comprise of a large number of sensor devices at fixed locations, wherein each individual sensor obtains measurements that are subsequently fused and processed at a central processing node. A key problem for such systems is to detect targets and identify their locations, for which a large body of literature has been developed focusing primarily on employing parametric models for signal attenuation from target to device. In this paper, we adopt a nonparametric approach that only assumes that the signal is nonincreasing as function of the distance between the sensor and the target. We propose a simple tuning \textit{parameter free} estimator for the target location, namely, the simple score estimator (SSCE). We show that the SSCE is $\sqrt{n}$ consistent and has a Gaussian limit distribution which can be used to construct asymptotic confidence regions for the location of the target. We study the performance of the SSCE through extensive simulations, and finally demonstrate an application to target detection in a video surveillance data set. 
\end{abstract}

\section{Introduction} 
\label{sec:introduction}

Target detection and localization represents a canonical problem in distributed systems, wherein information is obtained from sensing devices and then appropriately fused to identify the presence and location of target(s). Sensing technologies have evolved over time from phased arrays in radar systems (see \cite{niu2012target} and references therein), to wireless sensor networks involving many inexpensive sensors (see survey paper \cite{akyildiz2002wireless}), to highly sophisticated surveillance/monitoring systems integrating video and other sensor data \cite{joshi2012survey}. 
Examples of this canonical problem based on such diverse technologies abound and include precision agriculture~\cite{cardell2005reactive}, surveillance~\cite{estrin2007reflections}, animal behavior~\cite{mainwaring2002wireless}, drone tracking,  emergent disaster response and recovery~\cite{blatt2006energy}, fire hazards \cite{son2006design}, structural integrity of critical infrastructure \cite{chen2017nb}. 

Such distributed systems comprise a large number of sensors (acoustic, image/video, chemical, environmental) deployed at various (fixed or random) locations, wherein each individual sensor acquires signals from the surrounding area at fixed time intervals. The task of a central location is to integrate or fuse the data recorded by the sensors to locate or track an object or other quantity of interest, such as a crack in a bridge, or a chemical spill in the environment. In what follows, we formally describe the problem for a single fixed time. 

Consider $n$ identical sensors deployed at locations $\{X_i\}_{i=1}^n \in \R^d$ over a $d$-dimensional region $\rchi  \in \R^d$, where $d$ is $1, 2, $ or $3$.  Our object of interest is the location of a target that emits a  signal; e.g., infrared, acoustic, temperature, etc. Let $\theta_0 :=(\theta_{0, 1}, \ldots, \theta_{0, d}) \in \Theta$ denote the position of the target, where $\Theta \subset \R^d$ is called the `monitoring region'.  We assume that the energy or intensity of the signal attenuates with distance (from the target) according to a \textit{nonincreasing} function $\eta_0:\R^+\to\R^+$, where $\R^+$ denotes the positive real line, i.e., the true  energy/intensity of the signal at  sensor located at $X$ is $\eta_0\left(|\theta_0-X|^2\right), $  where $|\cdot|$ denotes the Euclidean norm and $\eta_0(0)$ is the energy of the signal at the target. However, since the sensor measurements are error prone, the observed energy at a sensor at $X$ is
\begin{equation}\label{eq:model_sen}
Y = \eta_0\left(|\theta_0-X|^2\right) + \epsilon,
\end{equation}
where $\epsilon$ is the unobserved measurement error. We will assume that $\E(\epsilon|X)=0$ and $\E(\epsilon^2|X) < \infty$ for almost every $X.$ The goal here is to estimate  the \textit{unknown} function $\eta_0$ and $\theta_0\in \Theta$  based on an i.i.d.~sample $(X_i, Y_i)$ under minimal assumptions on $\eta_0.$

Given the importance and wide applicability of this canonical problem, a large body of work has emerged; for a comprehensive review see the book \cite{varshney2012distributed}. The majority of existing work imposes a parametric functional form on the signal attenuation function $\eta$ leveraging information about the nature of the signal obtained by the sensors. This is a justified approach when dealing with thermal or acoustic signals that exhibit exponential and polynomial rates of decay, respectively~see e.g., \cite{blatt2006energy, clouqueur2001value, li2002detection, sheng2005maximum}. However, in many real-life scenarios, such parametric assumptions fall short. For example, the above assumptions are hard to justify for image/video acquiring sensors, or fail to hold in non-ideal environments, like the presence of dense vegetation or at high altitude, where the signal attenuation deviates from such nice parametric forms~\cite{watanabe1996sound}. To address such issues, some approaches quantize the signal (record a value 0 if the signal is below a certain threshold and 1 otherwise) -- see \cite{katenka2007local, katenka2008robust} and references therein. However, this strategy can lead to significant loss of information thereby negatively impacting target detection and localization capabilities of distributed systems. To that end, we propose to address the problem of estimating $\theta_0$ and $\eta_0$ in~\eqref{eq:model_sen} \textit{without any parametric} assumptions on $\eta_0$. 

Note that the nature of the problem under consideration corresponds to  semiparametric estimation with \textit{bundled parameters}, wherein the parametric and nonparametric components are intertwined (see \cite{huang1997interval}). Such problems have been studied in the literature based on estimators that require tuning parameters, see e.g., ~\cite{Powelletal89, LiDuan89, ICHI93, HardleEtAl93, Hristacheetal01, DelecroixEtal06, MR2529970, cuietal11} and references therein. However, observe that in our setting, the nonparametric component $\eta_0$ is governed by a natural shape constraint: \textit{monotonicity}. It is by now very well known that the use of shape constraints like monotonicity, convexity, log-concavity, etc., lead to elegant tuning parameter free estimates in a wide repertoire of nonparametric problems involving function estimation that, at least in one dimension (i.e., shape constrained functions of one variable), produce minimax optimal rates under minimal smoothness assumptions, see e.g., ~\cite{groeneboom2001estimation, zhang2002risk, GS15, GJ14, MR3576560,
MR2509075, han2017sharp, 2019arXiv190902088K, gao2020estimation} and references therein. Bypassing the tuning parameter selection step provides estimates that are truly \textit{data-driven}: in fact, shape-constrained procedures have an adaptive data-driven bandwidth choice built into the algorithms for their computation, and therefore extraneous stipulations of bandwidth via cross-validation or other techniques are not necessary. 

Hence, we adopt the shape-constrained approach to leverage its advantages for the problem at hand. This is rendered feasible in our version of the bundled parameters problem by the recent developments in~\cite{groeneboom2016current, balabdaoui2016least, balabdaoui2019score, 2019arXiv190902088K} in the related single index model (where $\E(Y|X)= \eta_0(\theta_0^\top X)$) with shape constrained link function $\eta$, which demonstrated how the parameter of interest $\theta_0$ can be estimated at the optimal $\sqrt{n}$ rate while using a tuning parameter free approach for the estimation of the nuisance parameter $\eta$. 

The main difficulty in studying the asymptotics for the estimator of $\theta_0$ in the shape constrained framework  of~\eqref{eq:model_sen} comes from the fact that the monotonically constrained estimator of $\eta_0$ is piecewise constant and thus lies on the ``boundary'' of the space of monotone functions. When the parametric ($\theta_0$) and nonparametric ($\eta_0$) components are not bundled (e.g., Cox proportional hazards model or partial linear regression model), the discontinuity (and boundary problem) of the estimator of $\eta_0$ can be overcome using traditional techniques because the asymptotics of the estimators for $\theta_0$ do not involve $\eta_0'$, see e.g.,~\cite{MR1915446, MR1394975, huang2002note}. This is, however, not true when $\theta_0$ and $\eta_0$ are intertwined, see e.g., ~\cite{2017arXiv170800145K}. To overcome the above difficulties, we adapt the powerful and elegant techniques developed in~\cite{groeneboom2016current} and~\cite{balabdaoui2019score} to develop and study tuning parameter free estimators for~\eqref{eq:model_sen}. As opposed to the single index model studied in the above works, the index in~\eqref{eq:model_sen} ($|\theta_0-X|$) is not a linear function of the parameter. This creates a number of new technical challenges that require careful handling; see e.g., Section~\ref{sec:PropertyofM}. Our work shows that the tools developed in~\cite{groeneboom2016current} can be used for general bundled problems (where the index is not linear in $\theta$ or $X$), provided the index is only locally linear, and therefore expands the scope of these techniques to a broader set of problems/models. Finally, \cite{balabdaoui2019score} assumes that the errors have all moments, while we relax this assumption significantly and establish $\sqrt{n}$-consistency for the estimator of $\theta_0$ for heavier tailed errors; we require errors to have only finite  sixth   moment (conditional on $X)$. This relaxation is important in many applications, due to the nature of the operating environment \cite{liu2009robust} or possible adversarial signal contamination \citep{swami2002some, dai2017sparse}.
A summary of the key technical contributions of this work is provided next:
\begin{enumerate}
  \item We find simple conditions on $\rchi$ the support of $X$ and $\eta_0$ for the  parameter in~\eqref{eq:model_sen} ($\eta_0$ and $\theta_0$) to be identifiable. 
 \item We provide two tuning parameter free estimators for $\theta_0$; namely the simple score estimator and least squares estimator, see Section~\ref{sec:Ident}. Furthermore, each of the estimators is associated with a tuning parameter free estimator for $\eta_0$.
  \item In contrast to most works in the shape constrained literature, we find the rate of convergence of the location estimators under  heavy-tailed errors. We allow the errors to be arbitrarily dependent on $X$. We show that the simple score estimator is $\sqrt{n}$ consistent and asymptotically normal as long as $\E(\epsilon^6|X)< \infty$ almost every $X.$
  \item We study the performance of both the  simple score estimator and least squares estimator through extensive simulations, and analyze a real life dataset in Section~\ref{sec:real_data_analysis_CCTV_footage}. We use the $m$-out-of-$n$ bootstrap to provide tuning  parameter free inference for the simple score estimator. 
\end{enumerate}

\paragraph*{Organization.} The remainder of the paper is organized as follows. In Section~\ref{sec:Ident}, we provide simple conditions for identifiability of the model, followed by two parameter free estimators for the target location $\theta_0$. In Section~\ref{sec:asymptotic_analysis}, we provide asymptotic analyses for the two estimators. Section~\ref{sub:asymptotic_analysis_of_the_lse} finds rate upper bounds for the least squares estimators of $\theta_0$ and $\eta_0$. Section~\ref{sub:asymptotic_analysis_of_the_sse} shows that the simple score estimator for $\theta_0$ is $\sqrt{n}$ consistent and asymptotically normal. In Section~\ref{sec:simulation_study}, we study the finite sample performance of both  estimators through extensive simulations. We also illustrate that the $m$-out-of-$n$ bootstrap can be used for valid inference for the simple score estimator. In Section~\ref{sec:real_data_analysis_CCTV_footage}, we use the proposed estimators in a surveillance application to locate an individual. Specifically,  we use video footage from a wide angle CCTV camera for the entrance lobby of the INRIA Labs at Grenoble, France~\citep{fisher2005caviar}. Section~\ref{sec:concluding_remarks} summarizes the contribution of the paper and provides some concluding remarks, and in particular thoughts about extending this approach to the time-series case [i.e. a series of data on sensor readings that arrive at consecutive time points] which is relevant to tracking a moving target. The proofs of all the results in the Appendix.

\section{Model identifiability and estimation}\label{sec:Ident}

We start by addressing identifiability issues for the model posited in~\eqref{eq:model_sen} based on the following two assumptions.
\begin{enumerate}[label=\bfseries (A\arabic*)]
\setcounter{enumi}{0}
 \item The support $\rchi$ of $X$ is a bounded convex set with at least one interior point.  The covariate $X$ has a bounded density with respect to the Lebesgue measure on $\rchi$. The parameter set $\Theta$ is bounded with non empty interior and $\theta_0$ belongs to the interior of $\Theta$. Further, let  $T\in \R$  be some finite number such that  $\sup_{x\in \rchi} |x| \le T$ and $\sup_{\theta\in \Theta} |\theta| \le T$. \label{a1}
 \item The function $t\mapsto\eta_0(t)$ is nonconstant, continuously differentiable, and nonincreasing  on $\R^+$.  \label{eta_bound}
\end{enumerate}

The boundedness assumptions on $\rchi$ can be replaced by a sub-Gaussianity or heavy-tail moment assumption. In that case, the rate upper bound derived for the estimators of $\eta$ will suffer. On the other hand, as long as the elements of $X$ have enough moments, the score estimator proposed later in the paper will still be $\sqrt{n}$ consistent; also see Remark~3 of~\cite{balabdaoui2019score}.  In the following result (proved in Section~\ref{sec:proof_of_thmIdent}), we establish the identifiability of~\eqref{eq:model_sen}. The boundedness assumption on $\Theta$ is natural as in practice the monitoring regions are well known in advance. The continuity, non-constancy, and monotonicity assumptions on the attenuation function is very natural is justified by the physics of signal attenuation. 
\begin{lemma}\label{thm:Ident}
Suppose assumptions~\ref{a1} and~\ref{eta_bound} hold.  Then, $\theta_0$ and $\eta_0$ are unique. In other words, if there exists a function $g:\R^+ \to \R^+$ and $\beta\in \R^d$ such that 
\begin{equation}\label{eq:non_iden}
 \eta_0(|\theta_0-x|^2) = g(|\beta-x|^2) \qquad \text{ for all } x\in \rchi,
 \end{equation}
  then $\theta_0=\beta$ and $\eta_0= g$ on $\{|\theta_0-x|^2 : x\in \rchi\}.$
\end{lemma}
\todo[inline]{Do you need $g$ to monotone? Ans: No}
In this paper, we suppose that we have $n$ i.i.d.~observations $\{(X_i, Y_i) \in \rchi \times \R, 1\le i\le n\}$ from~\eqref{eq:model_sen}. Before introducing the estimators for the location parameter $\theta_0$, we propose the following simple \textit{profile least squares} estimator for the attenuation function for any location $\theta\in\Theta$
\begin{equation}\label{eq:eta_tilde}
\tilde\eta_\theta :=\argmin_{\eta \in \M}\Q_n(\eta, \theta),
\end{equation}
where 
\begin{equation}\label{eq:Q_def}
\Q_n(\eta, \theta) :=\sum_{i=1}^{n} \left(Y_i - \eta(|\theta-X_i|^2)\right)^2
\end{equation}
and \begin{equation}\label{eq:M_def}
\M := \big\{\eta : [0, 4T^2]\to \R^+ : \eta \text{ is a non-increasing function}\big\}.
\end{equation}
For every fixed $\theta$, the optimization problem in~\eqref{eq:eta_tilde} can be shown to be convex. However, note that $\tilde{\eta}_\theta$ is well defined only at $\{|\theta-X_i|^2\}_{i=1}^n$. In this paper, we consider the canonical extension of $\tilde\eta_\theta$, and define $t\mapsto\tilde\eta_\theta(t)$ to be the unique right continuous piecewise constant function on $[0, 4T^2]$ with potential jumps at $\{|\theta-X_i|^2\}_{i=1}^n$.
Further, it is well known that, when there are no ties in $\{|\theta- X_i|\}_{i=1}^n$, the profiled estimator $\tilde{\eta}_\theta$ is the left derivative of the least concave majorant of the cumulative sum diagram %
\begin{equation}\label{eq:cusum}
\left\{(0, 0), \big(1, Y_{(1, \theta)}\big), \ldots, \Big(k, \sum_{j=1}^k Y_{(j, \theta)} \Big), \ldots, \Big(n, \sum_{j=1}^n Y_{(j, \theta)} \Big)\right\}, 
\end{equation}
where $Y_{(k, \theta)}$ is the measurement corresponding to the sensor that is $k$'th closest to $\theta$; see for example~\cite[Theorem~1.2.1]{RWD88} or~\cite[Theorem~1.1]{BarlowEtAl72}. For ease of presentation, we assume that there are no ties in $\{|\theta- X_i|\}_{i=1}^n$. The case of ties can be easily handled by ``merging'' tied data points and considering a weighted least squares problem, see e.g., ~\cite{balabdaoui2016least, balabdaoui2019score, Patra16}. From~\eqref{eq:cusum}, the profile least squares estimator can be easily with a complexity of $O(n)$  via  the pool adjacent violators algorithm (PAVA); see~\cite[Section 1.2]{RWD88} and  \cite{Grotzinger84}.

For any $\theta\in \Theta$, we define the ``population'' version of $\tilde{\eta}_\theta$ as follows:
\begin{equation}\label{eq:eta_profile_pop}
\eta_{\theta}(t) := \E\big(\eta_0(|\theta_0-X|^2)\big| |\theta-X|^2=t\big)\qquad \text{ for all } t\in\R^+.
\end{equation}
The following lemma states a useful characterization for $\eta_{\theta}(t)$. It shows that for some fixed $\theta$, $\eta_\theta(\cdot)$ can be thought of as the minimizer of the population version squared error loss in~\eqref{eq:Q_def}.
\begin{lemma}\label{lem:eta_theta_def}
Suppose $\eta_0$ is a strictly decreasing function. Then, there exists $\delta_0>0$ such that for every $\theta\in B(\theta_0, \delta_0)$, $\eta_{\theta}$ uniquely minimizes $\eta \mapsto \E(| \eta(|\theta-X|^2)- \eta_0(|\theta_0-X|^2)|)$ over the class of nonincreasing functions $\M$.
\end{lemma}

Observe that $\eta_{\theta_0} = \eta_0$. Thus,  if $\theta_0$ were known, then $\tilde{\eta}_{\theta_0}$ would be the estimator for $\eta_0$. In Theorem~\ref{thm:unif_eta}, we will study the asymptotic properties of $\tilde{\eta}_\theta$ as $\theta$ varies in a small neighborhood of $\theta_0$.  Below, we  use the profile least squares estimator to propose two  \textit{tuning parameter free} estimators for $\theta_0$. 

\paragraph{Least Squares Estimator (LSE).} The LSE for $\theta_0$ defined as
\begin{equation}\label{eq:LSE}
 \check{\theta} :=\argmin_{ \theta\, \in \Theta} \sum_{i=1}^{n} \left(Y_i - \tilde{\eta}_\theta(|\theta-X_i|^2)\right)^2,
\end{equation}
where the profile least squares estimator $\tilde{\eta}_\theta$ is defined via~\eqref{eq:eta_tilde}. Note that the above minimization problem is free of tuning parameters. However, unlike~\eqref{eq:eta_tilde}, the optimization problem in~\eqref{eq:LSE} is typically \textit{non-convex}. Recall that the cumulative sum diagram in~\eqref{eq:cusum} depends only on the ordering of $\{|\theta-X_i|\}_{i=1}^n$. Thus, for every $\theta$ in the interior of the parameter space, and $\beta$ in some small neighborhood of $\theta$, we have that 
\begin{equation}\label{eq:Jumps}
\sum_{i=1}^{n} \left(Y_i - \tilde{\eta}_\theta(|\theta-X_i|^2)\right)^2 =\sum_{i=1}^{n} \left(Y_i - \tilde{\eta}_\beta(|\beta-X_i|^2)\right)^2.
\end{equation}
Further, for every $t\in [0, 4T^2]$ the function $\theta\mapsto \tilde{\eta}_\theta(t)$ is piecewise constant, as changes in $\theta$ may lead to different ordering for $\{|\theta-X_i|\}_{i=1}^n$. As a consequence, $\theta\mapsto\sum_{i=1}^{n} (Y_i - \tilde{\eta}_\theta(|\theta-X_i|^2))^2$ is piecewise constant with multiple global minimizers. The results that follow hold true for any global minimizer $\check{\theta}$. Once we have the LSE for $\theta_0$, we define $\check{\eta}$, the LSE for $\eta_0$, as
\begin{equation}\label{eq:eta_check}
\check{\eta}:= \argmin_{\eta \in \M}\Q_n(\eta, \check\theta).
\end{equation}
We study the asymptotic behavior of LSE $(\check{\eta}, \check{\theta})$ in Theorem~\ref{thm:LSE_joint}. 
\begin{figure}[h!]
\centering
\includegraphics[width=.9\textwidth]{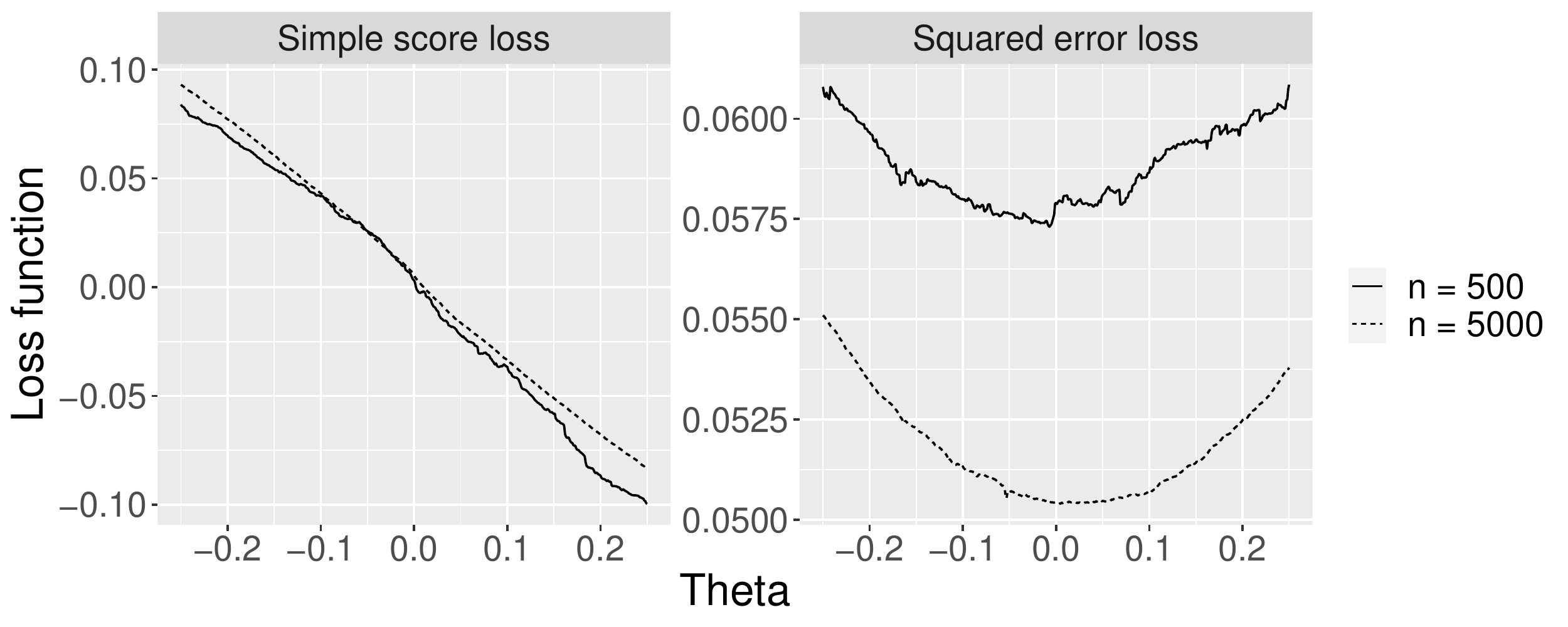}
  \caption[]{Plot of  $\theta \mapsto\frac{1}{n} \sum_{i=1}^{n} (Y_i - \tilde\eta_\theta(|\theta-X_i|^2))X_i$ (left panel) and $\theta \mapsto \frac{1}{n}\sum_{i=1}^{n} (Y_i - \tilde{\eta}_\theta(|\theta-X_i|^2))^2$ (right panel) for sample sizes $500$ (solid) and $5000$ (doted). In both cases, we have taken $\theta =0, $ $X\sim\text{Uniform}[-3, 3]$, $\epsilon\sim t_7/4$, and $Y= 1/(1+ .1 |X-\theta|^2)^2 + \epsilon$.}
  \label{fig:plot_of_loss}
\end{figure}

\paragraph{ Simple Score Estimator (SSCE).} To motivate this estimator, assume for the moment that $\tilde\eta_\theta$ is differentiable;
 then, $\check{\theta}$ can be defined as the solution to the following score equation:
 \begin{equation}\label{eq:fake_score_estim}
\sum_{i=1}^{n} \left(Y_i - \tilde\eta_{\theta}\big(|\theta-X_i|^2\big)\right)\tilde\eta_{\theta}'\big(|\theta-X_i|^2\big)(X_i-\theta)  = \textbf{0}_d,
\end{equation}
where $\textbf{0}_d\in\R^d$ is a vector comprising of zeros. The above score equation is equivalent to a semiparametric efficient score equation, and hence it is reasonable to expect that its solution would give rise to a semiparametrically efficient estimator for $\theta_0$. However, since $\tilde\eta_\theta$ is a piecewise constant function, $\tilde\eta_\theta'$ does not exist. In light of this, we propose a new estimator based on the following simple modification of the above efficient score equation~\cite{balabdaoui2019score, groeneboom2016current}. We define an SSCE by a zero of
\begin{equation}
\theta \mapsto n^{-1}\sum_{i=1}^{n} \left(Y_i - \tilde\eta_\theta\big(|\theta-X_i|^2\big)\right)(X_i-\theta).
\end{equation}
Further, by the property of the isotonic estimator $\tilde{\eta}_\theta$, we have that $\sum_{i=1}^n Y_i = \sum_{i=1}^n \tilde\eta_\theta\big(|\theta-X_i|^2\big)$. Thus, the SSCE is the zero of $\theta\mapsto \mathbb{M}_n(\theta)$, where 
 \begin{equation}\label{eq:score_estim}
\mathbb{M}_n(\theta):=\sum_{i=1}^{n} \left(Y_i - \tilde\eta_\theta\big(|\theta-X_i|^2\big)\right)X_i.
\end{equation}
Note that in the definition of $\mathbb{M}_n(\cdot), $ we have ignored the non-differentiability of $\tilde\eta_\theta$ and replaced $\tilde\eta_\theta'$ in~\eqref{eq:fake_score_estim} with $1$.  The motivation for~\eqref{eq:score_estim} is that in the absence of $\tilde\eta_\theta'$, $\mathbb{M}_n(\cdot)$ can be seen as a ``rough approximation" of the efficient score equation~\cite{balabdaoui2019score, groeneboom2016current}.  Further, 
$\phi(x, y)=  \left(y - \tilde\eta_\theta\big(|\theta-x|^2\big)\right)(x-\theta)$ is a valid influence function in the sense of~\cite[Chapter 1.2]{MR1915446}.  Another motivation for SSCE stems from observing that $\E(X [Y-\eta_{\theta_0} (|\theta_0-X|^2)])$, the ``population'' version of $\MM_n(\theta_0)$, is ${\bf 0}_d$. In Section~\ref{sub:asymptotic_analysis_of_the_sse}, we discuss assumptions (see assumptions~\ref{a4} and~\ref{assum:NonzeroEverywhere}) under which $\theta_0$ is the unique zero of $\theta\mapsto \E(X [Y-\eta_{\theta} (|\theta-X|^2)]).$

After obtaining the SSCE for $\theta_0$, define $\hat{\eta}$ (the SSCE for $\eta_0$) as
\begin{equation}\label{eq:eta_hat}
\hat{\eta}:= \argmin_{\eta \in \M}\Q_n(\eta, \hat\theta).
\end{equation}
Now recall that $\theta\mapsto \tilde{\eta}_\theta(t)$ is piecewise constant with discontinuities (for any $t\in [0, 4T^2]$). This and~\eqref{eq:Jumps} imply that $\theta\mapsto\MM_n(\theta)$ will have discontinuities and exact zeros in~\eqref{eq:score_estim} may not always exist. Instead, we define the SSCE $\hat \theta$ as a ``zero crossing'' of $\theta \mapsto \MM_n(\theta).$ The following definition is from~\cite{groeneboom2016current}.

\begin{defn}[Zero Crossing, \cite{groeneboom2016current}]
We say that $\beta^*$ is a zero crossing of a real-valued function $\beta\mapsto\zeta(\beta)$ on a set $\mathcal{A}$ if each open neighborhood of $\beta^*$ contains points $\beta_1, \beta_2 \in \mathcal{A}$ such that $\zeta(\beta_1)\zeta(\beta_2) \le 0$. We say that an $m$-dimensional function $\beta\mapsto \zeta (\beta) = (\zeta_1 (\beta), ... \zeta_m (\beta))'$ has a crossing of zero at a point $\beta^*$, if $\beta^*$ is a crossing of zero of each component $\beta\mapsto\zeta_j(\beta)$ for every $j\in [m].$
\end{defn}
We study the asymptotic behavior of SSCE $(\hat{\eta}, \hat{\theta})$ in Theorem~\ref{thm:SSE_consis} and find the asymptotic distribution of $\hat{\theta}$ in Theorem~\ref{thm:Asymp_norml}. Following the work of~\cite{balabdaoui2020profile} in a single index model, one can show that the SSCE is asymptotically equivalent to the following minimizer ${\theta}^\dagger:= \argmin_{\theta\in \Theta}\big|\mathbb{M}_n(\theta)\big|$. However, we do not pursue this extension here. 

\section{Asymptotic analysis of the estimators} 
\label{sec:asymptotic_analysis}

We start our analysis by establishing properties of the simple profile least squares estimator for $\tilde\eta_\theta$. Henceforth, we also require the following assumption on the distribution of $\epsilon.$

\begin{enumerate}[label=\bfseries (A\arabic*)]
\setcounter{enumi}{2}

\item The error $\epsilon$ in model~\eqref{eq:model_sen} has finite  $q$-th moment, i.e., $K_q\coloneqq \big[\E(|\epsilon|^q)\big]^{1/q}  < \infty$ where $q\ge 2$. Further, $\mathbb{E}(\epsilon|X) = 0, $ $P_X$ a.e.~and $\sigma^2(x)\coloneqq \E(\epsilon^2|X=x) \le \sigma^2 < \infty$ for all $x\in \rchi.$ \label{assum:err_mom}
  \end{enumerate}
 The above assumption on $\epsilon$ is fairly general and allows for heteroscedastic errors. Further, contrary to most existing work for similar models that require sub-Gaussian or sub-exponential errors (see, e.g.,~\cite{balabdaoui2016least,balabdaoui2019score,balabdaoui2020profile}), we allow the error distribution to have only finitely many moments. 
The following result, established in~Section~\ref{sec:proof_unif_eta}, shows that $\tilde\eta_\theta$ (defined in~\eqref{eq:eta_tilde}) converges to $\eta_\theta$ (defined in~\eqref{eq:eta_profile_pop}) uniformly in $\theta$ in a neighborhood of $\theta_0$.
\begin{theorem}\label{thm:unif_eta} Suppose assumptions~\ref{a1}--\ref{assum:err_mom} hold, then
  \begin{equation}\label{eq:eta_bound_unif}
  \sup_{\theta \in \Theta} \|\tilde{\eta}_\theta\|_\infty = O_p(n^{1/q}).
  \end{equation}
 Moreover, let $P_X$ denotes the distribution of $X$,  then there exists a fixed  $\delta_0>0$, such that
\begin{equation}\label{eq:unif_eta}
\sup_{\theta\in B(\theta_0, \delta_0)} \int \Big\{\tilde{\eta}_\theta(|\theta-x|^2) -\eta_\theta(|\theta-x|^2)\Big\}^2 dP_X(x) =O_p\big(n^{-2/3} n^ {2/q}\big).
\end{equation}
\end{theorem}
The profiled estimator $\tilde{\eta}_\theta$ of $\eta_\theta$ plays a crucial role in the definition and analysis of  both the LSE and SSCE. The first part of Theorem~\ref{thm:unif_eta} shows that even though $\M$ is an unbounded class, $\|\tilde{\eta}_\theta\|_{\infty}$ is not too large.  Moreover, the above uniform convergence result helps us study the behavior of the criterion/loss function around $\theta_0$ for both the LSE and SSCE.

\subsection{Asymptotic analysis of the LSE} 
\label{sub:asymptotic_analysis_of_the_lse}

In this section, we compute upper bounds on  the rate of convergence of $\check{\eta}(|\check\theta-x|^2)$ and $\check{\theta},$ to ${\eta_0}(|\theta_0-x|^2)$ and $\theta_0$, respectively.  Since $\check\theta$ is a minimizer $ \theta \mapsto \Q_n(\tilde{\eta}_\theta, \theta)$  over $\rchi$ and $\tilde{\eta}_\theta \in \M$ for all $\theta \in \rchi$, the next step in characterizing the asymptotic behavior of $\check{\theta}$ is to calculate the {metric entropy} of the class of functions $\{\eta(|\theta-\cdot|^2): \eta\in\M, \theta\in \Theta\}$. However, by~\eqref{eq:eta_bound_unif}, we have that for large enough $n$, $P(\sup_{\theta \in B(\theta_0, \delta_0)} \|\tilde{\eta}_\theta\|_\infty > C n^{1/q}) \le \epsilon$ for some $C$ depending only on $\epsilon.$ 
 Thus, we will study the following class of functions,
\begin{equation}\label{eq:F_1}
\mathcal{F}_K := \big\{x\mapsto \eta(|\theta-x|^2): \eta\in\M, \theta\in \Theta, \|\eta\|_\infty \le K\big\}.
\end{equation}
Let $N_{[\, ]}(\varepsilon, \mathcal{F}_{K}, L_2(P_X))$ denote the $\varepsilon$-bracketing number of $\F_K$ in the $L_2(P_X)$ metric (see Section 2.1.1 of~\cite{VdVW96} for a formal definition).
The following lemma, proved in~Section~\ref{sec:ent_calc}, computes the bracketing entropy of $\F_K$.
\begin{lemma}\label{thm:entropy_FK}
Let $\varepsilon >0$ and $K>\varepsilon$. Then, there exists  a constant $A_1>0$ depending only on $d$, $\Theta$, and $\rchi$ such that $\log N_{[\, ]}(\varepsilon, \mathcal{F}_{K}, L_2(P_X)) \le {A_1 K}/{\varepsilon}.$
\end{lemma}

In Section~\ref{sec:proof_of_theorem_LSE_joint}, we use the above result to establish the following upper bounds on the rate of convergence for $x \mapsto \check{\eta}(|\check\theta-x|^2)$. To show that $\check{\theta}$ inherits the rate of convergence of the (joint) regression function, we will need the following assumption.
\todo[inline]{Where do we need A4 In theorem 3.2? Do we need twice differentiability? ANS: we don't. We would have needed it in the proof of $d_n$ but we do not need to constrain ourselves $\theta_0^\perp$.} 
\begin{enumerate}[label=\bfseries (A4$'$)]
\setcounter{enumi}{-1}
\item There exists an open set $\mathcal{A}\subset \{|\theta_0- x|^2 : x \in \rchi\}$, such that $t\mapsto \eta_0(t)$ is  continuously differentiable on $\mathcal{A}$, $\inf_{t\in \mathcal{A}}|\eta_0'(t)|>0$, and $\P(|\theta_0-X|^2\in \mathcal{A}) >0$. 
\label{a4prime}
\end{enumerate}
\begin{theorem}\label{thm:LSE_joint}
If assumptions~\ref{a1}--\ref{assum:err_mom} hold and $q\ge 5$, then we have
\begin{equation}\label{eq:conver_LSE_joint}
 \int \Big\{\check{\eta}(|\check\theta-x|^2) -\eta_0(|\theta_0-x|^2)\Big\}^2 dP_X(x) =O_p\big(n^{-2/3} n^{2/q} \big).
\end{equation}
Moreover, if assumption~\ref{a4prime} also holds, then 
\begin{equation}\label{eq:conver_LSE}
 |\check\theta-\theta_0| =O_p\big(n^{-1/3} n^{1/q} \big).
\end{equation}
\end{theorem}
 Assumption~\ref{a4prime} is inspired by~\cite[Assumption (A5)]{balabdaoui2016least} and allows us get the rate of convergence of $\check{\theta}$ from the rate of convergence of $\check{\eta}(|\check{\theta}-\cdot|^2)$.
 If~assumption~\ref{a4prime} doesn't hold then, the second part of proof of Theorem~\ref{thm:LSE_joint} can be easily modified to show that there exists a positive semi-definite matrix $I$ such that $ |I_{d\times d}(\check\theta-\theta_0)| =O_p\big(n^{-1/3} n^{1/q} \big).$ Assumption~\ref{a4prime} essentially says that $\eta_0$ must be smooth and non-constant on a region (in $\rchi$) of positive mass. 

The LSE discussed above is a natural  tuning parameter free estimator for $\theta_0$.  The sub-$\sqrt{n}$ upper bound on rate of convergence, however, raises the question whether the rate bound above is tight or the LSE actually converges at the much faster $\sqrt{n}$ rate.  To investigate this, we have done an extensive simulation study in Section~\ref{sec:simulation_study} of the paper. The simulations suggests that, indeed the above rate upper bound is not tight. However, it is still unknown whether  the LSE is $\sqrt{n}$ consistent, let alone its asymptotic distribution.  The rate of convergence of the LSE from the various simulation settings considered in Section~\ref{sec:simulation_study} is inconclusive, e.g., in Table~\ref{tab:location}, the  $\sqrt{n}\times \text{Var}(\check{\theta})$ appears to decrease, while Figures~\ref{fig:exp1 all} and~\ref{fig:linear_all} display an almost opposite trend. We believe that the difficulties in finding the true rate of convergence of the LSE for $\theta_0$ stems from the fact that $\tilde{\eta}_{\check{\theta}}$ is not  continuous and $\eta_0$ and $\theta_0$ are intertwined.     A similar phenomenon is observed for the LSE in the monotone single index models, where faster than $n^{1/3}$ rate (under sub-exponential errors) is conjectured and observed but not proved~\citep{tanaka2008semiparametric,balabdaoui2016least}. Lastly, the dependence of $q$ in the rate upper bound above in Theorem~\ref{thm:LSE_joint} can be improved by using the techniques developed in~\cite[Theorem 3.1 and Corollary 3.1]{2019arXiv190902088K} and~\cite[Theorems~4.1 and 7.3]{balabdaoui2016least}. However we do not pursue this marginal improvement in the current paper. Instead, we focus on studying a $\sqrt{n}$ consistent estimator with a tractable limit distribution that is practically useful, namely the SSCE.

 \todo[inline]{Need to add comments about $\sqrt{n}$ rate.}

\subsection{Asymptotic analysis of the SSCE} 
\label{sub:asymptotic_analysis_of_the_sse}
In this section, we study the asymptotics of the SSCE defined in~\eqref{eq:score_estim}. We will first prove its existence and consistency. Before stating the main results of this section, let us define $M: \rchi \to \R$, the population version of $\mathbb{M}_n(\cdot):$
\begin{equation}
M(\theta):=\E\Big(\left[Y - \eta_\theta\big(|\theta-X|^2\big)\right]({X-\theta})\Big).
\end{equation}
It is easy to see that $M(\theta_0)={\bf 0}_d$. However, it is not clear whether $\theta_0$ is the unique zero of $\theta\mapsto M(\theta)$ and/or ``well-separated'' in the sense of~\cite[Theorem 5.9]{vanderVaart98}. In Lemmas~\ref{lem:zeroCrossing} and~\ref{lem:deriv_M} (stated and proved in Section~\ref{sec:PropertyofM}), we will use the following  two assumptions to show that $\theta \mapsto M(\theta)$ has a unique zero at $\theta_0$, is differentiable at $\theta_0$, and $M'(\theta_0)$ is non-singular.


 \begin{enumerate}[label=\bfseries (A\arabic*)]
\setcounter{enumi}{3}

\item $\E\Big(\eta'_0(|\theta_0-X|^2) \text{Cov}\big(X\big||\theta_0-X|^2 \big)\Big)$ is a positive definite matrix.\footnote{In Lemma~\ref{lem:deriv_M}, we show that $M'(\theta_0)= \E(\eta'_0(|\theta_0-X|^2) \text{Cov}(X|\,|\theta_0-X|^2 ))$.}  \label{a4}

\item There exists a $\delta_0>0$ such that for all $\theta\in B(\theta_0, \delta_0)$ and $\theta\neq \theta_0$, the random  variable\label{assum:NonzeroEverywhere}
\[\text{Cov}\Big((\theta- \theta_0)^\top X, \eta_0\big(|\theta_0-X|^2\big)\big| |\theta-X|^2\Big)\neq 0 \quad \text{almost everywhere}.\]
\end{enumerate}
Note that~\ref{a4prime} is a sufficient condition for both~\ref{a4} and \ref{assum:NonzeroEverywhere}.  Assumptions~\ref{a4} and \ref{assum:NonzeroEverywhere} are similar to~\cite[Assumption A4]{2017arXiv170800145K},~\cite[Assumption A6]{balabdaoui2019score}, and~\cite[Theorem 4.1]{groeneboom2016current} among others.  The following result, proved in Section~\ref{sec:thm:SSE_consis}, shows that $\hat{\theta}$ exists with probability approaching one and is consistent. 


\begin{theorem}\label{thm:SSE_consis}
Suppose assumptions~\ref{a1}--\ref{assum:NonzeroEverywhere} hold, then the SSCE exists with probability approaching one and $\hat{\theta}$ is consistent for $\theta_0$, i.e., $\hat{\theta}\stackrel{P}{\to} \theta_0$.
\end{theorem}
If assumption~\ref{a4} does not hold, then just as in the case of Theorem~\ref{thm:LSE_joint}, we can show that $A^\top  (\hat \theta -\theta_0 ) =o_p(1)$, where $A:= \E(\eta'_0(|\theta_0-X|^2) \mathrm{Cov}(X||\theta_0-X|^2 ))$. To prove the asymptotic normality of $\hat \theta$, the SSCE, we will require the following smoothness assumption on the conditional expectation of $X$ given $|\theta_0-X|$.
\begin{enumerate}[label=\bfseries (A\arabic*)]
\setcounter{enumi}{5}
  \item The function  $ u\mapsto \E[X |\, |\theta -  X|^2=u]$ is twice continuously differentiable, except possibly at a finite number of points, and there exists a finite constant  $\bar{M}>0$ such that for every $\theta_1, \theta_2 \in \Theta$, \label{ContConditional}
\begin{equation}\label{eq:lip_h_beta}
\sup_{u \in [0, 4T^2]}\Big|\E\big[X \big|\, |\theta_1 -  X|^2=u\big] -\E\big[X \big|\, |\theta_2-  X|^2=u\big]\Big| \le \bar{M} |\theta_1-\theta_2|.
\end{equation}
\end{enumerate}
The assumption~\ref{ContConditional} is standard and widely used for semiparametric regression models of similar nature. Assumption~\ref{ContConditional} is similar to those in~\cite[Theorem 3.2]{VANC},~\cite[Assumption A5]{groeneboom2016current},~\cite[Assumption A5]{balabdaoui2019score},~and~\cite[Assumption B3]{2017arXiv170800145K}; also see~\cite[Assumption G2 (ii)]{song2014semiparametric}. The above papers, discuss many distributions of $X$ such that~\ref{ContConditional} holds. 

Based on the discussion preceding~\eqref{eq:score_estim} in Section~\ref{sec:Ident}, it is intuitively apparent that $\hat{\theta}$ is not semiparametrically efficient as~\eqref{eq:fake_score_estim} is not the efficient score equation. The following result (proved in~Section~\ref{sub:asymptotic_normality_of_}) shows that $\hat{\theta}$ is nonetheless asymptotically normal and finds its asymptotic distribution.
\begin{theorem}\label{thm:Asymp_norml}
Suppose assumptions~\ref{a1}--\ref{ContConditional} hold and $t\mapsto \eta_0(t)$ is strictly decreasing. Furthermore, suppose $q\ge 6$\footnote{Using Theorem 3.1 and Corollary 3.1 of~\cite{2019arXiv190902088K} and techniques used in the proof of Theorems~4.1 and 7.3 of~\cite{balabdaoui2016least} one can improve the assumptions that $q\ge 6$. However we do not pursue this marginal improvement in the current paper.} and let  
\[A:=\E\Big(\eta'_0(|\theta_0-X|^2) \text{Cov}\big(X\big||\theta_0-X|^2 \big)\Big)\] 
and \[\Sigma:= \E\Big[\sigma^2(X) \big(X -\E(X |\, |\theta_0 -  X|^2\big)\big(X^\top -\E(X^\top |\, |\theta_0 -  X|^2\big)\Big].\]
Then
\begin{equation}\label{eq:limit}
\sqrt{n}(\hat\theta-\theta_0) \stackrel{D}{\to} N(0, A^{-1} \Sigma A^{-1}).
\end{equation}
\end{theorem}
\begin{remark}[Efficient estimation of $\theta_0$]\label{rem:Efficient_est}
As discussed above, the SSCE is not semiparametrically efficient. We can improve upon it by considering the following estimator:
\begin{equation}\label{eq:ESE}
\tilde{\theta}_h:= \argmin_{\theta\in \Theta} \widetilde{\mathbb{M}}_n(\theta)\quad \text{where} \quad \widetilde{\mathbb{M}}_n(\theta) = \Big\|\sum_{i=1}^{n} \left(Y_i - \tilde\eta_{\theta}\big(|\theta-X_i|^2\big)\right)\tilde\eta_{\theta, h}'\big(|\theta-X_i|^2\big)X_i \Big\|,
\end{equation}
where for every $u\in [0, 4T^2]$ and $h>0$, we define
\begin{equation}\label{eq:eta_tilde_der}
\tilde\eta_{\theta, h}' (u) := \frac{1}{h}\int_{0}^{4T^2} K\left(\frac{u-x}{h} \right)d\tilde{\eta}_{\theta}(x),
\end{equation}
 where $x\to K(x)$ is a twice differentiable kernel with support $[0, 1]$. Using the techniques developed in this paper and~\cite{balabdaoui2019score}, one can show that if $h\asymp n^{-1/7}$, then $\tilde{\theta}_h$ is an efficient estimator for $\theta_0$ under some additional smoothness assumptions on $\eta_0$. However, we do not consider this estimator any further, since it involves a tuning parameter and the finite sample performance of $\tilde{\theta}_h$ can depend heavily on the choice of the bandwidth $h$. Finally, it is important to note that $\tilde{\theta}_h$ will be efficient \textit{only} when the errors are homoscedastic. In case of heteroscedastic errors, it may be the case that the SSCE has lower asymptotic variance than the above estimator. In fact, in a closely related model~\cite{balabdaoui2020profile} give an example where the efficient estimator (under homoscedastic error) has worse finite sample (and asymptotic) variance than a non-efficient estimator.  
\end{remark}

\section{Performance evaluation} 
\label{sec:simulation_study}


To investigate the performance of the LSE and SSCE, we carry out several simulation experiments. The \texttt{R} codes developed to implement these estimators and the scripts replicating the numerical results are available at \if1\blind{~\url{http://stat.ufl.edu/~rohitpatra/}.}\fi\if0\blind{the author's website. }\fi  We consider i.i.d. observations from
\begin{equation}\label{eq:simul_model}
Y = \eta_0(| \mathbf{0}_d-X|^2) + \epsilon\footnote{Not that here we have choosen $\theta_0= \mathbf{0}_d$ without loss of generality.},
\end{equation}
where the distribution of the covariates,  the attenuation function, and distribution of the errors  vary across a wide range of options. In the following subsections, we study the behavior of $\check{\theta}_1$ and $\hat{\theta}_1$, where $\check{\theta}:= (\check{\theta}_1,\ldots,\check{\theta}_d)$ and $\hat{\theta}:= (\hat{\theta}_1,\ldots,\hat{\theta}_d).$ However, before undertaking a comprehensive comparison across different settings, we focus on the potential
statistical inefficiency of the LSE, shown in the next example.

\subsection{A simple example showcasing the inefficiency of the LSE} \label{sub:_a_simple_exam} 

We consider the following setup for the model~\eqref{eq:simul_model}:
\begin{equation}\label{eq:simple_example}
\eta_0(t)= 1/(1+0.1 t), \qquad X \sim \text{Uniform}[-3, 3]^2, \qquad and \qquad \epsilon| X\sim \text{Normal}\big( 0, (0.1)^2\big).
\end{equation}
In Table~\ref{tab:location}, we list the sample variance of (centered and scaled) LSE and SSCE for the above setting. It can be seen that the variance of $\sqrt{n}(\hat{\theta}_1-\theta_{0,1})$ stabilizes and converges to its asymptotic limit for even relatively small sample size ($\geq 500$). On the other hand, the variance of $\sqrt{n}(\check{\theta}_1- \theta_{0, 1})$ appears to grow with $n$.  This finding together with similar results presented in the sequel suggest that the LSE might not be $\sqrt{n}$ consistent. Note that analogous inconclusive behavior is observed for the LSE estimator in the closely related monotone single index models in~\cite{tanaka2008semiparametric, balabdaoui2016least}.\footnote{Under the stronger assumption of convexity, ~\cite{2017arXiv170800145K} show that a minor variant of the LSE is not only $\sqrt{n}$ consistent, but also \textit{semiparametrically efficient}.} It is remarkable, though, that the estimated sample variance of SSCE is significantly lower than that of LSE (even under homoscedastic error). Thus, even though it is unknown whether the LSE is $\sqrt{n}$ consistent, it is safe to conclude that the LSE is not efficient for the estimation of 
$\theta_0$.

\begin{table}[!ht]
      \caption{\label{tab:location} Scaled sample variance for the LSE ($\check{\theta}_1$) and SSCE ($\hat{\theta}_1$) for $\theta_{0, 1}$ over $500$ replications as the sample size grows.}
      \centering
  \begin{tabular}{ccccccc}

  Sample size& $500$ & $1000$ &$3000$ &$5000$ &$10000$ & $15000$\\
  \midrule
    $n\times\text{var}(\check{\theta}_{1}) $&$1.43 $&$ 1.50$  &$ 1.57$ &$ 1.46 $&$ 1.75 $& $1.82$\\    $n\times\text{var}(\hat{\theta}_{1})$& $0.74 $&$ 0.75$  &$ 0.74$ &$ 0.73 $&$ 0.74 $& $0.74$   \\
  \end{tabular}
  \end{table}

\begin{table}[ht]
      \caption[Table showing choices for the attenuation function and the distribution of the covariates and error.]{\label{tab:choices_table}Table showing choices for the attenuation function and the distribution of the covariates and error for the simulation considered in Section~\ref{sub:independent_covariates}. }
      \centering

\begin{minipage}{0.35\textwidth}

  \begin{tabular}{l}
  
  Choice of attenuation function\\
  \midrule
$ 5+(1+t/5)^{-3}$\\ 
$5+\exp(-t/4)$ \\
$5+(1-t/10)\times \mathbf{1}_{[0, 10]}(t)$\\
  \bottomrule
  \end{tabular}
\end{minipage}
\begin{minipage}{0.63\textwidth}

  \begin{tabular}{l}
  
  Choices of covariate distribution\\
\midrule
$X_1\sim \text{Unif}[-3, 3]$ and $X_2| X_1 \sim 0.2 X_1+ .8\text{Unif}[-3, 3]$\\
    $X\sim \text{Unif}[-3, 3]^2$\\
    $X\sim \text{Normal}(0, I_{2\times2})$\\
  \bottomrule
  \end{tabular}
\end{minipage}
\bigskip

\begin{minipage}{0.4\textwidth}

  \begin{tabular}{l}
  
Homoscedastic error distributions \\
\midrule
    $\epsilon|X\sim \text{Normal}(0, 1)$\\
    $\epsilon|X\sim(-1)^{\text{Ber}(1)} \times\text{Beta}(2, 3)$\\
    $\epsilon|X\sim t_3$ \\
    $\epsilon|X\sim t_7$ \\
  \bottomrule
  \end{tabular}
\end{minipage}
\begin{minipage}{0.59\textwidth}
  \begin{tabular}{l }
  
Heteroscedastic error distributions\\
\midrule
   $ \epsilon|X\sim \log\big(2+ |\theta_0-X|^2\big)\times \text{Normal}(0, 1) $\\
    $\epsilon|X\sim \log\big(2+ |\theta_0-X|^2\big)\times (-1)^{{\text{Ber}(1)} } \times\text{Beta}(2, 3)$\\
    $\epsilon|X\sim \log\big(2+ |\theta_0-X|^2\big)\times  t_3$ \\
    $\epsilon|X\sim \log\big(2+ |\theta_0-X|^2\big)\times t_7$  \\ 
    \bottomrule
  \end{tabular}
\end{minipage}
  \end{table}
  \begin{figure}
\centering
\includegraphics[width=.9\textwidth]{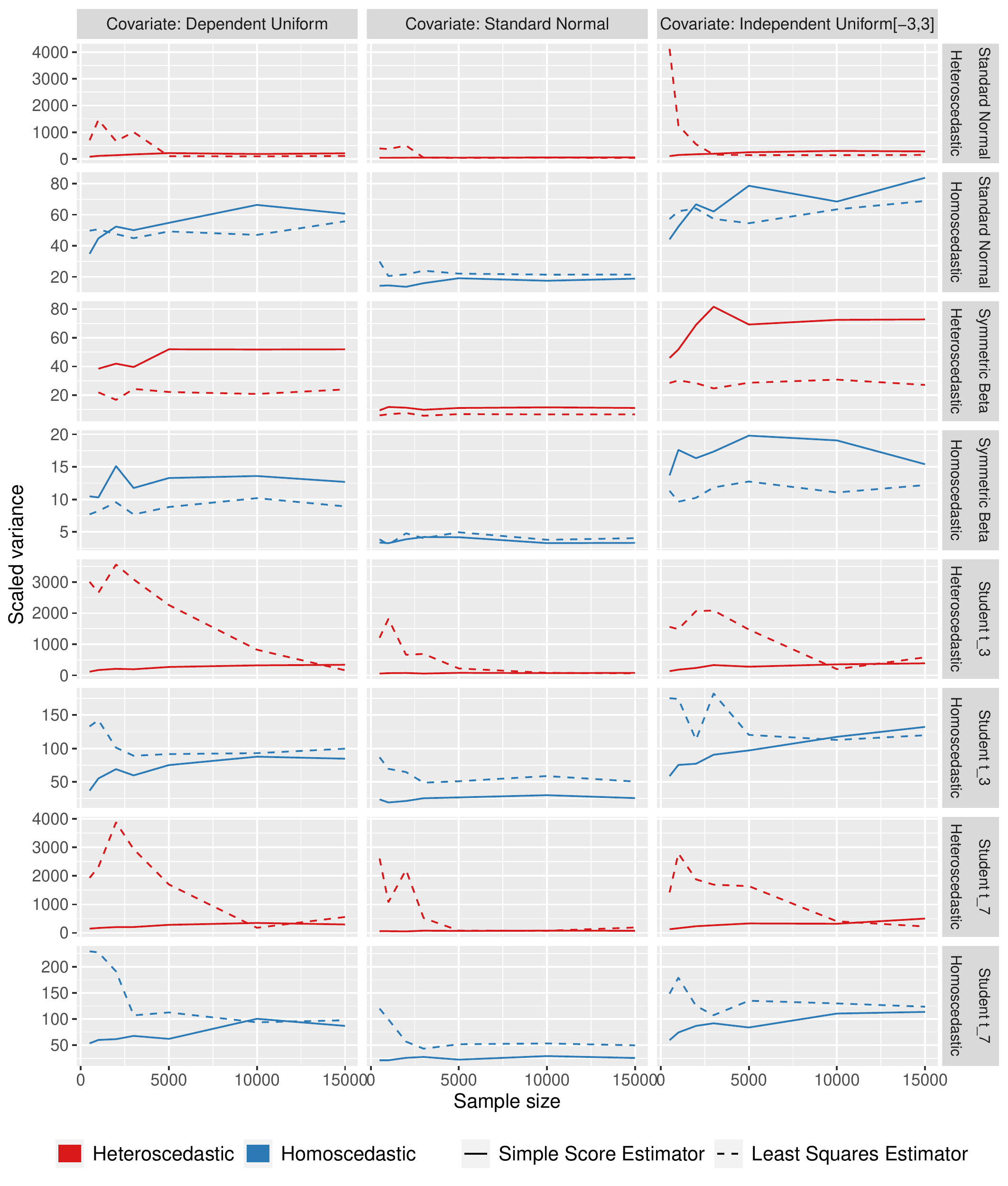}
  \caption[]{Plot of scaled sample variance ($n\times \text{var}(\theta^\dagger)$) for the LSE (dashed) and SSE (solid) for the first co-ordinate of $\theta_0$ when $\eta_0(t) = 5+(1+t/5)^{-3}$. The sample variance is calculated by using $200$ replications.}
  \label{fig:poly_2 all}
\end{figure}

\begin{figure}[h!]
\centering
\includegraphics[width=.9\textwidth]{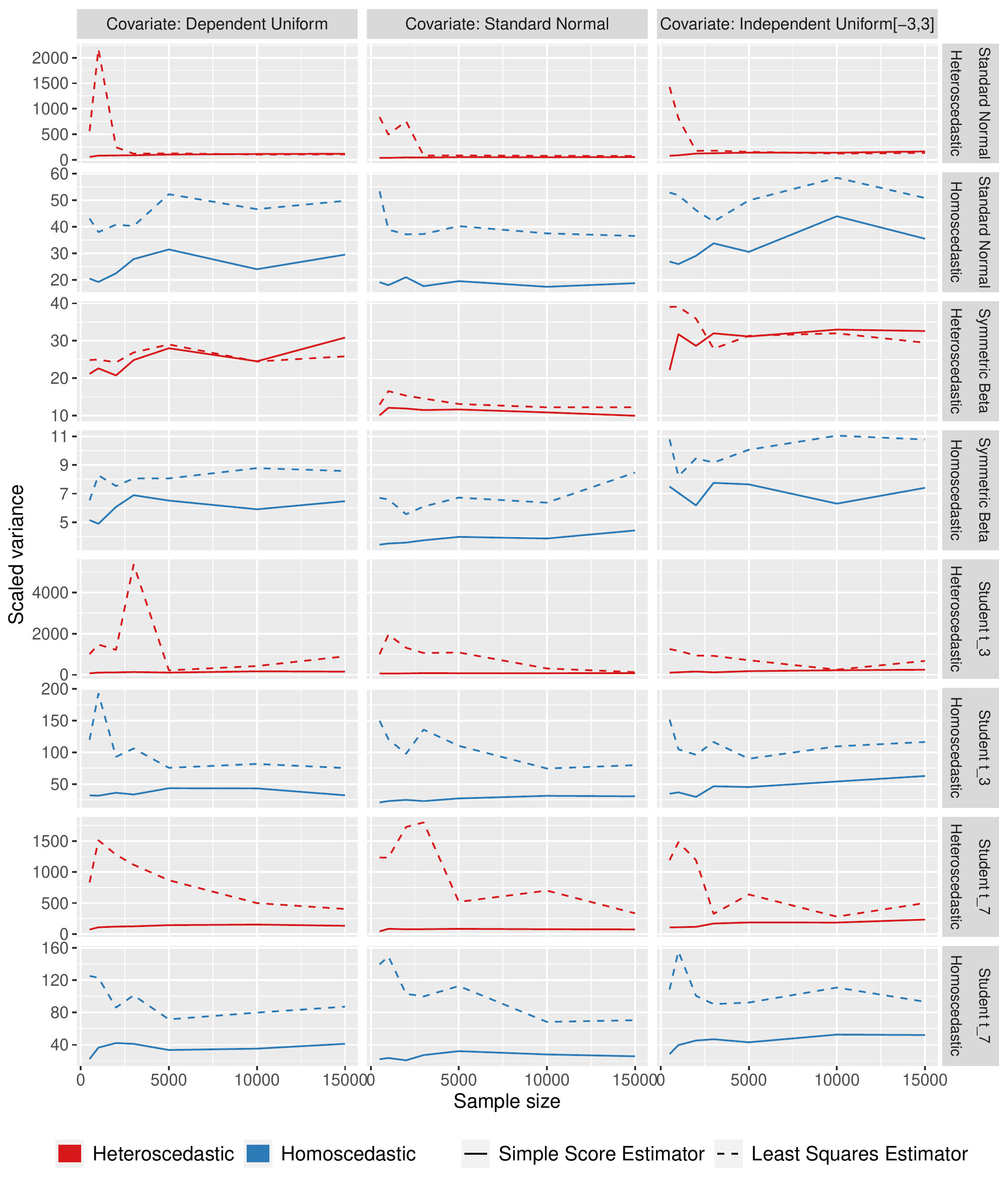}
  \caption[]{Plot of scaled variance ($n\times \text{var}(\theta^\dagger)$) for the LSE (dashed) and SSE (solid) for the first co-ordinate of $\theta_0$ when $\eta_0(t) = 5+\exp(-t/4)$.}
  \label{fig:exp1 all}
\end{figure}
  \begin{figure}[!h]
\centering
\includegraphics[width=.9\textwidth]{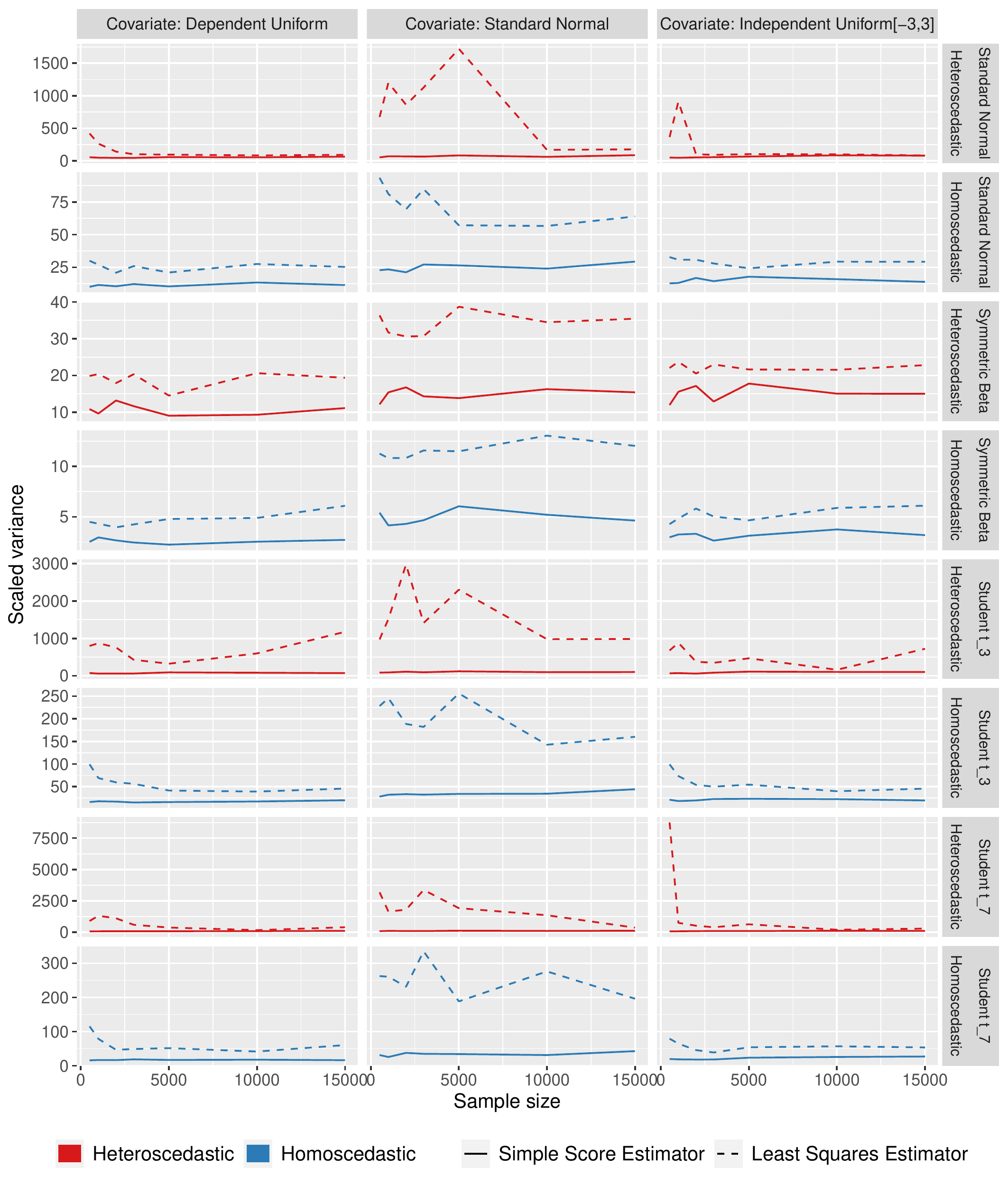}
  \caption[]{Plot of scaled variance ($n\times \text{var}(\theta^\dagger)$) for the LSE (dashed) and SSE (solid) for the first co-ordinate of $\theta_0$ when $\eta_0(t) = 5+(1-t/10)\times \mathbf{1}(0 \le t\le10)$.}
  \label{fig:linear_all}
\end{figure}

\subsection{Extensive comparison of LSE and SSCE} 
\label{sub:independent_covariates}
We consider a grid of settings for the attenuation function, the distribution of both the covariates, and the errors as described in Table~\ref{tab:choices_table}. The dimension is fixed as $d=2$ and $\theta_0 = (0, 0).$ In figures \ref{fig:poly_2 all}--\ref{fig:linear_all}, we  illustrate the finite sample performance of the LSE and SSCE by plotting the sample variance (scaled by $n$) of the two estimators of $\theta_{0, 1}$, the first co-ordinate of $\theta_0$. In each figure, we fix the choice of the attenuation function and vary across the choices for the distribution of covariates and errors. The figures illustrate that in almost all settings considered, the empirical variance of the SSCE is significantly lower than that of the LSE. Further, the sample biases (not shown here) of both the estimators are close to zero. We see that LSE has a smaller finite sample variance than the SSCE in only seven out of $74$ simulation settings considered in Figures~\ref{fig:poly_2 all}--\ref{fig:linear_all}. Indeed, a similar behavior of the SSCE was observed in the  monotone single models by~\cite{balabdaoui2019score} and \cite{balabdaoui2020profile}.

The analysis in~Section~\ref{sub:asymptotic_analysis_of_the_lse} established that $|\check{\theta}-\theta_0| = O_p(n^{-1/3} n^{1/q})$. However, simulations suggest that the above rate is not tight. We conjecture that $\check{\theta}$ converges to $\theta_0$ at a rate faster than $n^{1/3}$, but is not  $\sqrt{n}$ consistent. This is in line with the extensive work on the monotone current status and single index models (see \cite{balabdaoui2016least, groeneboom2016current, balabdaoui2019score, balabdaoui2020profile}). Another interesting empirical observation is that $\hat{\theta}$, the estimator based on a simple score equation, thoroughly outperforms the LSE.

\begin{figure}[!h]
\centering
\includegraphics[width=.91\textwidth]{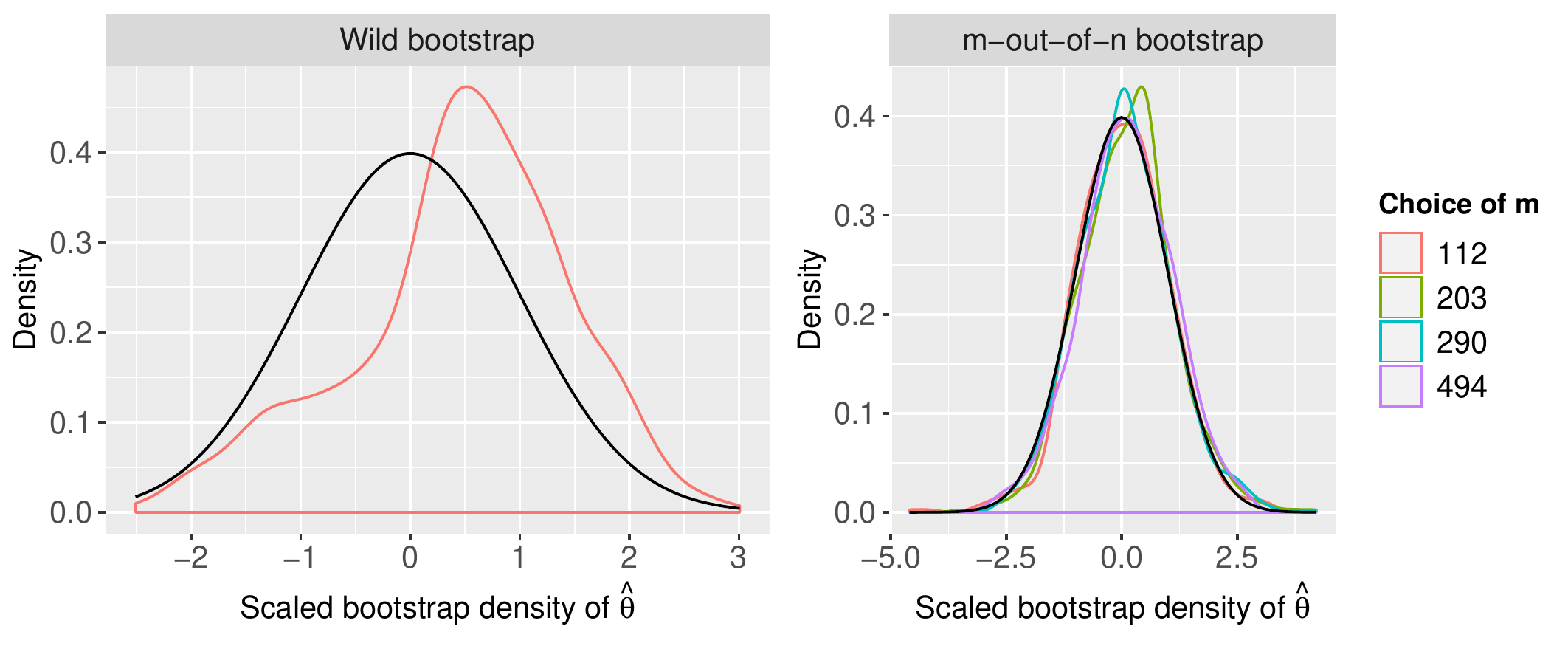}
  \caption[]{Kernel density estimate of the (scaled and centered) empirical bootstrap distribution of $\hat{\theta}_1$ based on $800$ bootstrap replications for both wild (scaled by $\sqrt{n}$)  and $m$-out-of-$n$ (scaled by $\sqrt{m}$) bootstrap. The data has a sample size of $1200$ and is generated according to $Y= 5+ (1+ |X-0|/5)^{-3}+ \text{Normal} (0, 1)$ and $X\sim \text{Unif}[-3, 3]^2$. For wild bootstrap, we use Mammen's two-point distribution~\cite{mammen1993bootstrap}. For $m$-out-of-$n$, the choices of $m$ are $(\floor{n^{2/3}}, \floor{n^{3/4}}, \floor{n^{4/5}}, \floor{n^{7/8}})$. The solid black line in both of the plots represents the true asymptotic density of $\sqrt{n}(\hat{\theta}_1-\theta_{0, 1})$.}
  \label{fig:DensityBoot}
\end{figure}
\subsection{Confidence intervals for the SSCE based on the Bootstrap} 
\label{sub:confidence_interval_and_coverage}

The goal of this subsection is to compute a confidence interval for $\theta_0$ based on the SSCE. Theorem~\ref{thm:Asymp_norml} establishes that the asymptotic distribution of the SSCE depends on nuisance parameters such as the $\sigma^2(X)$, $E(X| |X- \theta_0|^2)$, and $\eta_0'$. One can use consistent estimates of these quantities to estimate the asymptotic variance of the SSCE and create an asymptotic confidence interval for $\theta_0$. However, such estimators often involve tuning parameters. Following the theme of this paper for a tuning parameter free approach, we will use the standard $m$-out-of-$n$ bootstrap procedure on the data $\{(X_i, Y_i)\}_{i=1}^n$ to compute a confidence interval for $\theta_0$. Figure~\ref{fig:DensityBoot} shows the empirical bootstrap distribution of $\hat{\theta}_1$ for the wild bootstrap~\cite{mammen1993bootstrap} (left panel) and $m$-out-of-$n$ bootstrap~\cite{bickel2012resampling} (right panel). Extensive simulation results suggest that the wild bootstrap with Mammen's two-point distribution~\cite{mammen1993bootstrap} is inconsistent for the SSCE. However, for all the settings considered in the paper, the $m$-out-of-$n$ bootstrap performs well for most (valid) choices of $m$. Figure~\ref{fig:BootSimul} depicts the empirical coverage of both bootstrap procedures for sample sizes ranging from $300$ to $1200$ for a number of the settings described in Table~\ref{tab:choices_table}.
The results suggest that coverage for the $m$-out-of-$n$ bootstrap is close to the nominal level (90\%) for all settings involving an exponential attenuation function, but about 5\% below for the polynomial attenuation one. On the other hand, the wild bootstrap's coverage falls significantly short of the nominal level,
even for very large sample sizes.

\section{Locating an individual from video surveillance footage} 
\label{sec:real_data_analysis_CCTV_footage}

Next, we use the proposed estimators in a surveillance application. Specifically, video footage is available from a wide
angle CCTV camera for the entrance lobby of the INRIA Labs at Grenoble, France. Individuals walk in and out of the lobby and
the objective is to determine their locations. The video frames can be downloaded from the Context Aware Vision using Image-based Active Recognition (CAVIAR) project~\cite{fisher2005caviar}. The specific data set (Walk1) employed in this paper can be downloaded from~\url{http://groups.inf.ed.ac.uk/vision/CAVIAR/CAVIARDATA1/Walk1/Walk1_jpg.tar.gz}.
\begin{figure}
\centering
\includegraphics[width=.9\textwidth]{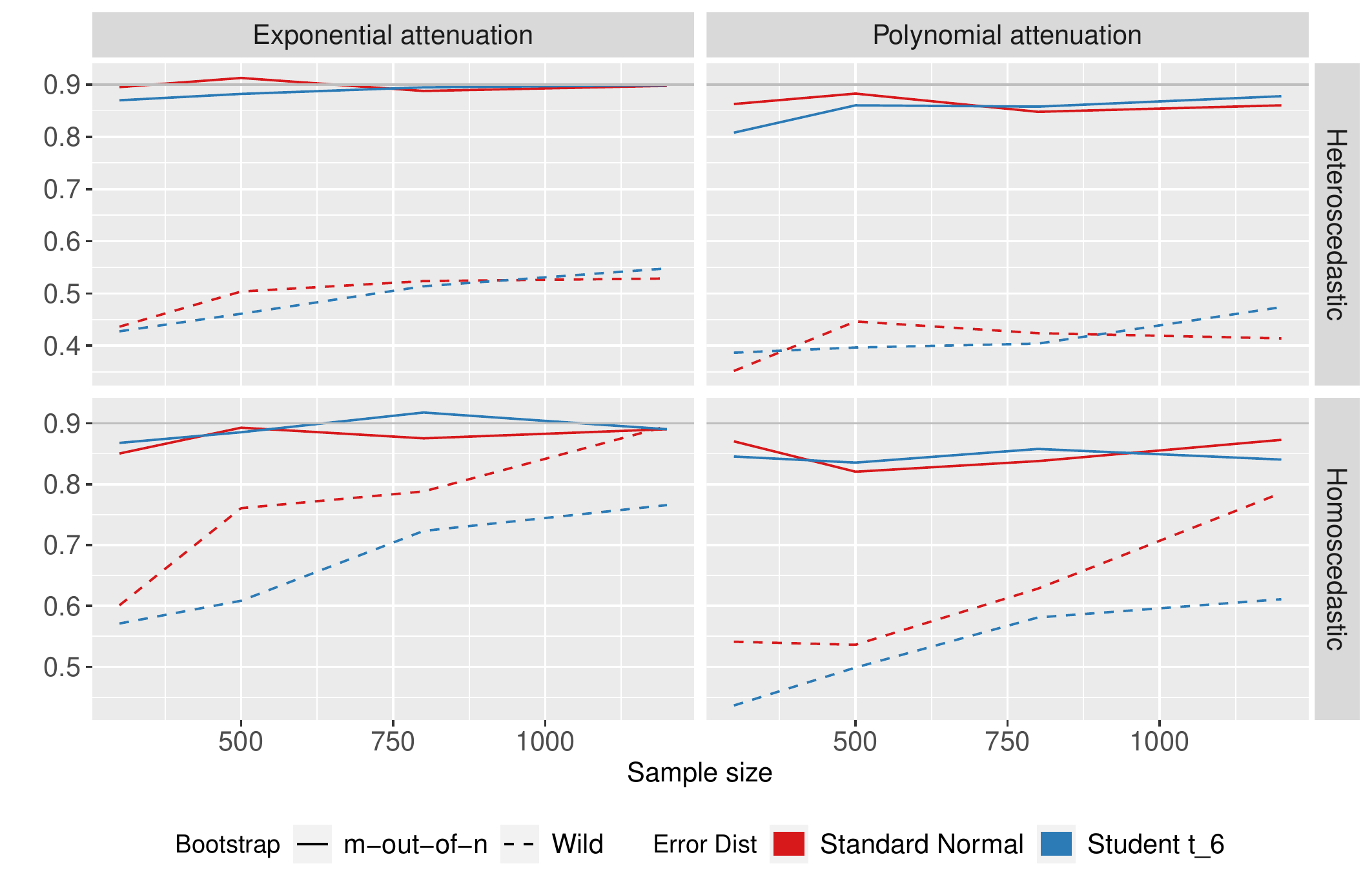}
  \caption[]{Empirical coverage for two $90\%$ nonparametric bootstrap confidence intervals for $\theta_{0, 1}$ based on the SSCE under model~\eqref{eq:model_sen} for a number of settings for the attenuation functions and distributions of the error listed in~Table~\ref{tab:choices_table}. For wild bootstrap, we use Mammen's two-point distribution~\cite{mammen1993bootstrap} and for $m$-out-of-$n$ bootstrap we used $m= \floor{n^{7/8}}$. The computed empirical coverage is based on $500$ bootstrap replications.}
  \label{fig:BootSimul}
\end{figure}

In the video under consideration, there were two time windows  (frames 1002--1019 and 1182--1230) where there was no movement in the lobby. Starting at frame $1235$ and ending at frame $1525$ a person walks across the lobby. Given an image, the objective is to locate the person in the lobby. The frames are stored in an RGB format. Thus, each frame consists of three channels: Red, Green, and Blue; for each channel, we have a gray scale matrix of dimension $384\times288$. Thus, a single frame can be represented as a $384\times288\times 3$ tensor.
 \begin{figure}[h!]
\centering
\includegraphics[width=.9\textwidth]{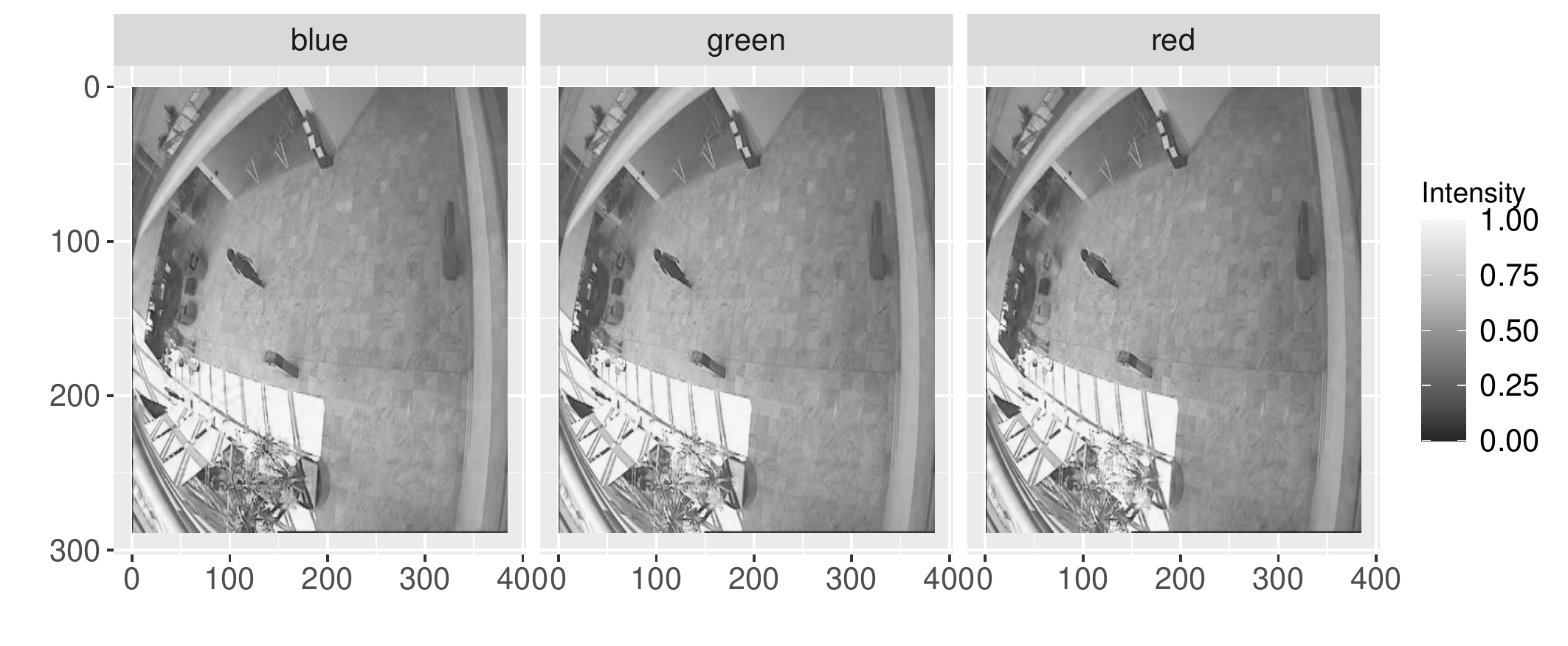}
  \caption[]{Intensity plot for the frame 1380 from the ``Walk1'' benchmark data set of the CAVIAR project for  the red, green, and blue channels.}
  \label{fig:1380_fgnd_foreground}
\end{figure}
To leverage our model, we analyze the three channels of the frame independently. For each channel, the person in the frame is considered to be the target and each pixel corresponds to a sensor measurement. In this section, we analyze frame number 1380. 
To adapt the data into our framework, we convert the data in the frame of interest (\# 1380) into a long vector of size $384\times288$ independent observations from the model
  \begin{equation}\label{eq:Total_model}
  Z = f(X)+ \epsilon,
  \end{equation}
  where $X$ is a two dimensional vector denoting the location of the pixel, $Y$ is the measurement at the corresponding pixel, and $\epsilon$ is a noise term. We further assume that the function $f$ can be modeled via an additive structure as
\[f(X)= b_0(X)+ s_0(X), \]
where $b_0(X)$ denotes the unknown ``background'' and $s_0(X)$ denotes the unknown ``signal''. The 
function $X\mapsto b_0(X)$ is unknown. However, each of the 66 frames (wherein the lobby is empty) provides independent and identically distributed observations from $Z= b_0(X)+ \epsilon$ for every $X$. Thus, we can estimate $X\mapsto b_0(X)$ consistently for each $X$ via a simple sample average. Thus, for the remainder of the section, we treat $X\mapsto b_0(X)$ as a known function and assume we have measurements from the following model:
\begin{equation}\label{eq:final_data_model}
  Y= Z- b_0(X)= s_0(X) + \epsilon.
\end{equation}
\begin{figure}
\centering
\includegraphics[width=.9\textwidth]{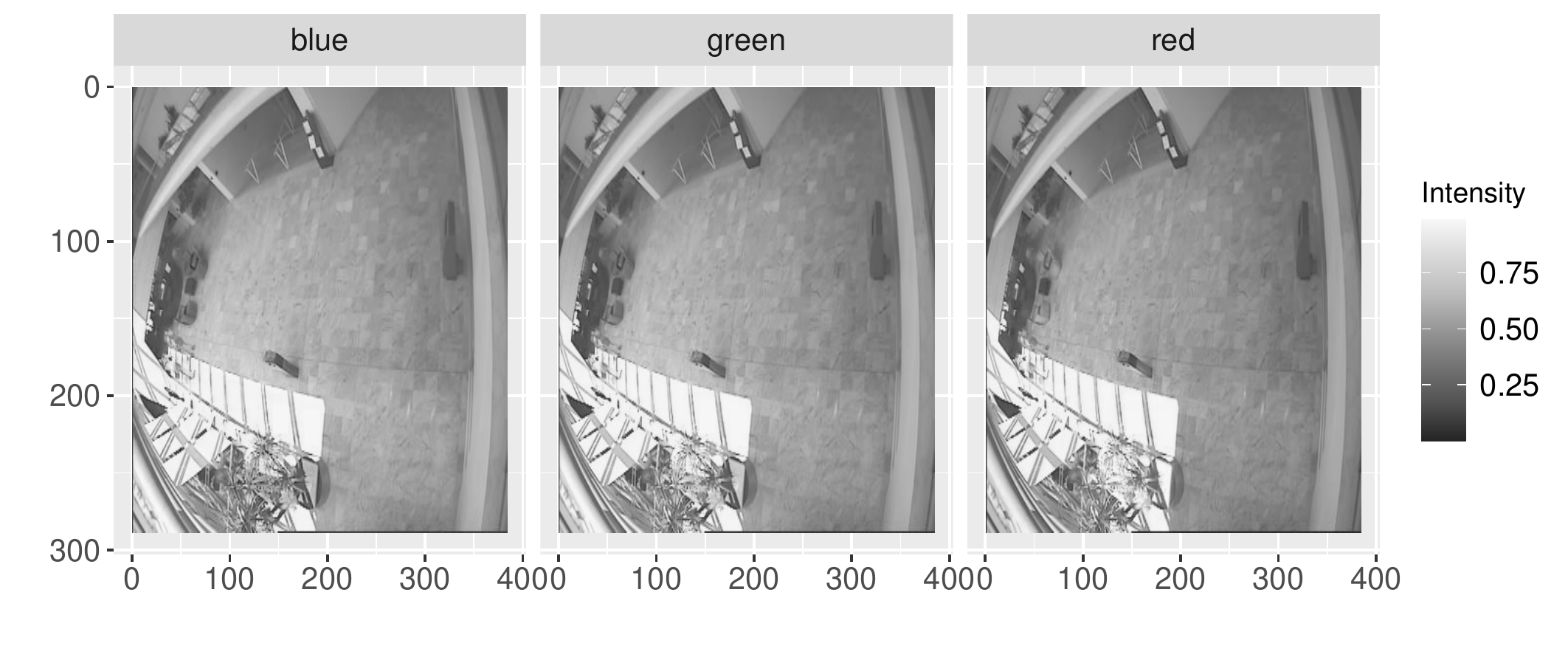}
  \caption[]{Intensity plot of the estimate of $b(X)$ (i.e., the average of 66 frames where the lobby is empty) for the red, green, and blue channels.}
  \label{fig:background}
\end{figure}

\begin{figure}
\centering
\includegraphics[width=.9\textwidth]{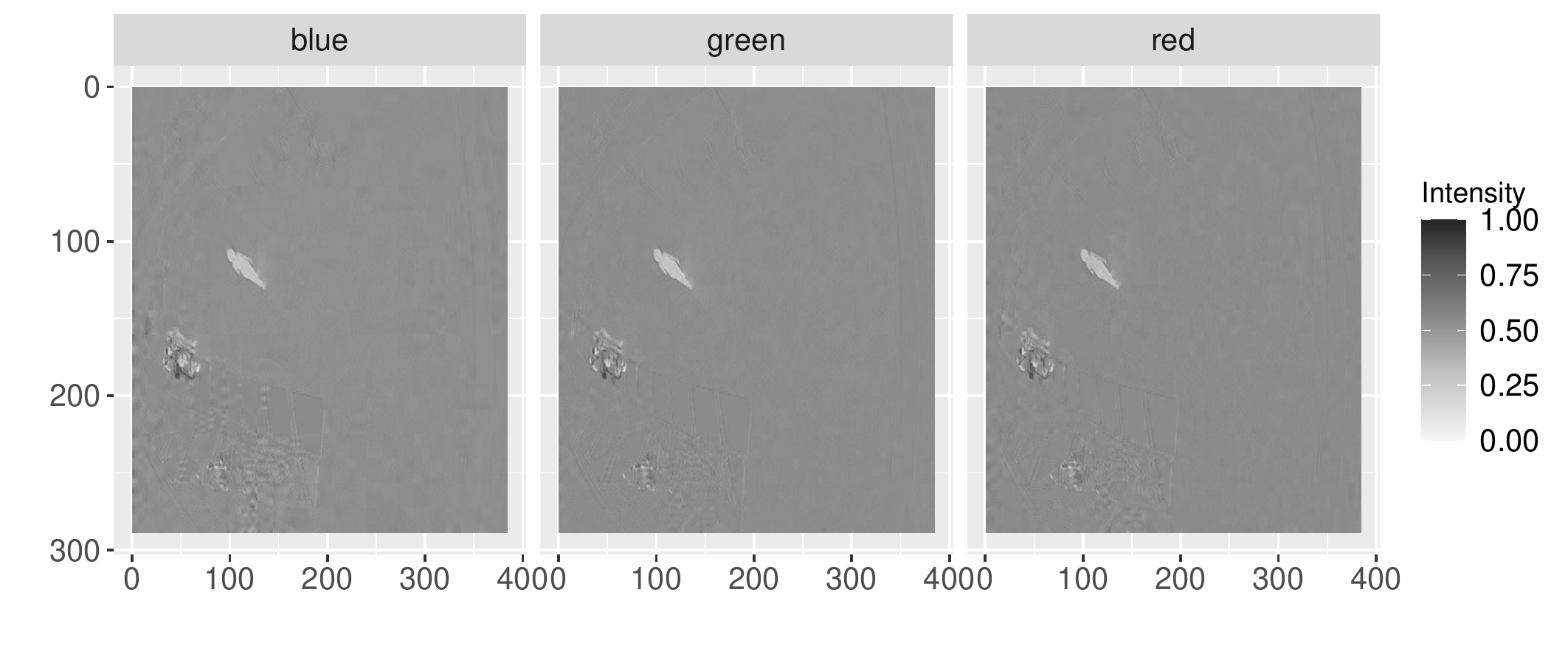}
  \caption[]{Intensity plot of all pixels  of re-centered image (see~\eqref{eq:final_data_model}) for the frame number 1380.}
  \label{fig:diff}
\end{figure}
Since in the current setup the target corresponds to a person walking through the lobby, in an ideal world (one without light bleeding),\footnote{Light bleeding is the phenomenon, where  photo-charge from an pixel bleeds/leaks into other nearby pixels, thus affecting their detected sensitivity.} $x \mapsto s_0(x)$ would a step function of the form $c_1 \times \mathbf{1}(x\in A)$ for some constant $c_1$ and subset $A\subset \rchi$. The presence of the target at a pixel location will elevate (or deprecate) the true gray scale intensity, while the gray scale intensity at the other locations is zero. Observe that the true gray scale intensity can be taken to be zero, because we assume that $b_0(X)$ has been estimated well from all preceding frames. Further, in this paper, we will assume that the target (set $A$) can be well approximated by a disk. This implies that, in the ideal world, we can assume that $s_0(X)$ is a step function of the form $c_1 \times \mathbf{1}(|\theta_0-x|\le c_2)$ for some constants $c_1$ and $c_2$. However, to accommodate for light bleeding in the data and to allow for the developed theory to be applicable to this dataset, we posit that $s_0(\cdot)$ can be approximated by a rapidly (strictly) decreasing function around the target. For the rest of the section, we assume that
\begin{equation}\label{eq:S_def}
s_0(X) = \eta_0(|X-\theta_0|^2),
\end{equation}
where $\theta_0$ can be interpreted as the location of the target and $\eta_0$ is a monotone function. For an example of $\eta_0$ see Figure~\ref{fig:Fucntion_Detection} where we plot the estimated functions for the three color bands. Combining~\eqref{eq:final_data_model} and~\eqref{eq:S_def}, the data are converted to our posited framework with the following exception: the design points $\{X_i\}_{i=1}^n$ are fixed grid locations in this application, while in the technical developments we assume a random design setup. To remedy this, we sample the grid locations uniformly at random.

\begin{remark}\label{rem:dec_inc_choose}
A natural question is, how does the practitioner know whether to assume that $t\mapsto \eta(t)$ is decreasing or increasing? This is an important question, because the intensity (values of $Y$) of the person relative to the background can be higher or lower depending on factors such as lighting and colors in the picture.  A simple way to address this, is to fit both an increasing and a decreasing function for~\eqref{eq:final_data_model} and choose the fit that has the smaller squared error loss ($ \sum_{i=1}^n (Y_i-\eta(|X_i-\hat\theta|^2))$) at the SSCE. In fact, this is how we decided to fit a decreasing function for $\eta_0$ in~\eqref{eq:S_def}. 
\end{remark}

 The image for frame 1380 from the ``Walk1'' benchmark data set is shown in Figure~\ref{fig:1380_fgnd_foreground}. As mentioned earlier, we use frames 1002--1019 and 1182--1230 to estimate the background levels (i.e., the function $b(\cdot)$). As the above 66 frames do not record any movement, we assume that the data in the model follows model~\eqref{eq:Total_model} with $f(x)=b(x)$ for all $x$ and estimate it by the sample mean of the 66 frames. In Figure~\ref{fig:background}, we plot the estimates of $x\mapsto b_0(X)$ for each of the three channels. We now treat $b_0(\cdot)$ as known,  and assume that we have observations $(X, Y)$ from~\eqref{eq:final_data_model}. Figure~\ref{fig:diff} corresponds to a heat map of the centered frame; i.e., we treat the image as being generated by~\eqref{eq:final_data_model} with $Y$ being the intensity and $X$ being the location of the pixel.

\begin{figure}[h!]
\centering
\includegraphics[width=.9\textwidth]{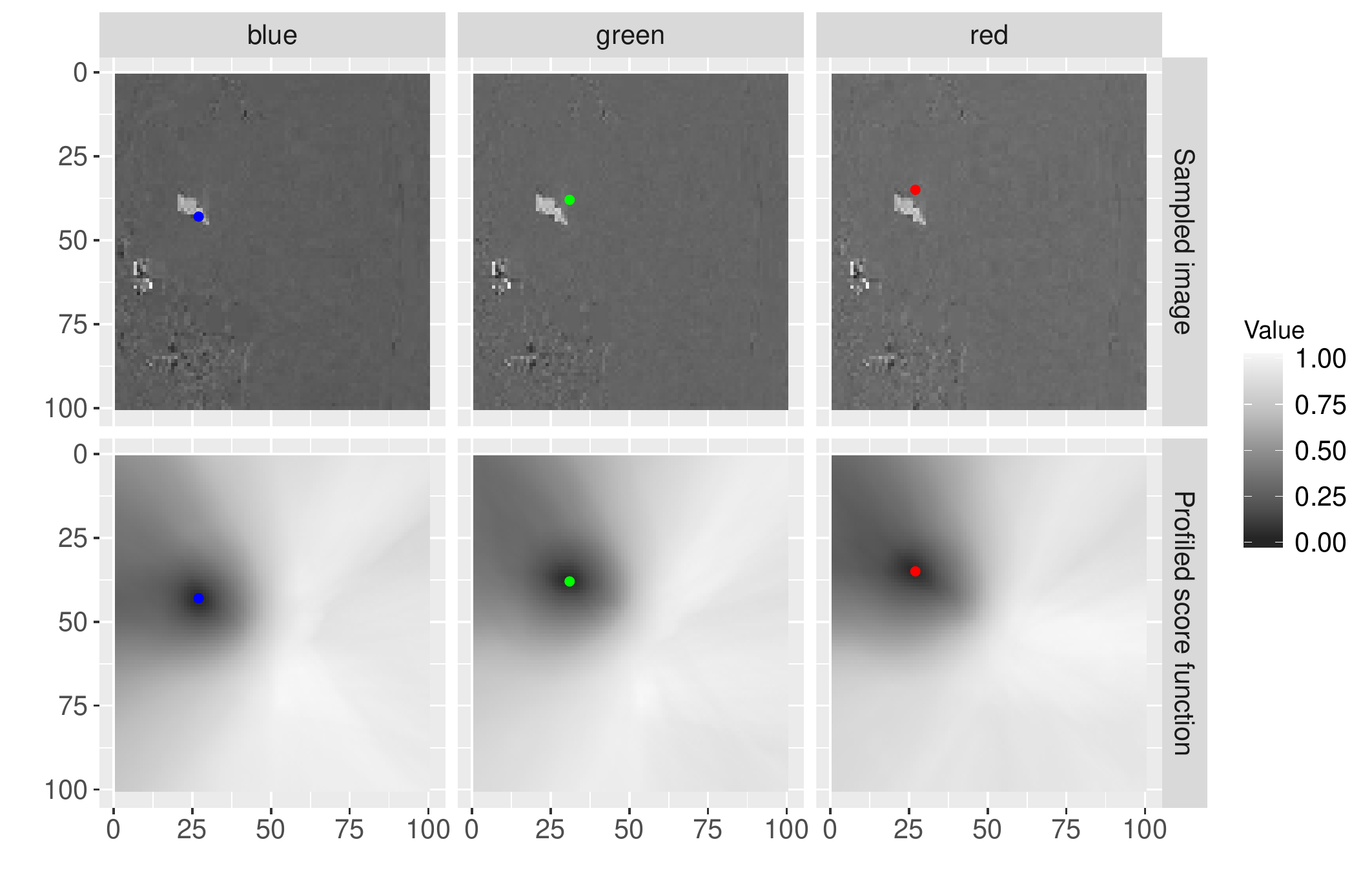}
  \caption[]{First row:  estimates of $\theta_0$ corresponding to each of the three bands overlaid on the heat map of the sampled image. Second row: scaled heat map of the profiled score function ($\theta\mapsto |\mathbb{M}_n(\theta)|$) at each of the sensor location overlaid with the estimates of $\theta_0$ (the location of the minimum of the profiled score function) for each of the three bands.}
  \label{fig:detected_image_sample}
\end{figure}

To replicate the random design scenario, we sample (uniformly) a grid of size $100\times 100$ from the images in Figure~\ref{fig:diff}. The top row of Figure~\ref{fig:detected_image_sample} depicts this ``observed'' data set. The next step is to compute the SSCE. The second row of Figure~\ref{fig:detected_image_sample}  depicts the heat map of $\theta\mapsto |\MM_n(\theta)|$ as $\theta$ varies over the location of the sensors. In each of the heat maps, the location of the minimum is marked with solid dots that correspond to the estimated location of the individual. 
Finally, in the left panel of Figure~\ref{fig:Fucntion_Detection}, we overlay the original image with the detected location from each of the channels. The right panel plots the $t\mapsto \hat\eta(t)$ corresponding to the three channels - recall that $\hat\eta$ is defined as $\tilde{\eta}_{\hat{\theta}}$ as in \eqref{eq:eta_hat}. Finally, in Figure~\ref{fig:bootstrap}, we plot the ellipsoid confidence regions based on a normal approximation, with the dispersion matrix based on an $m$-out-of-$n$ bootstrap with $m= \floor{n^{7/8}}$. It is worth noting that the confidence ellipsoids contain the target for all three of the channels. 
\begin{figure}[h!]
\centering
\includegraphics[width=.9\textwidth]{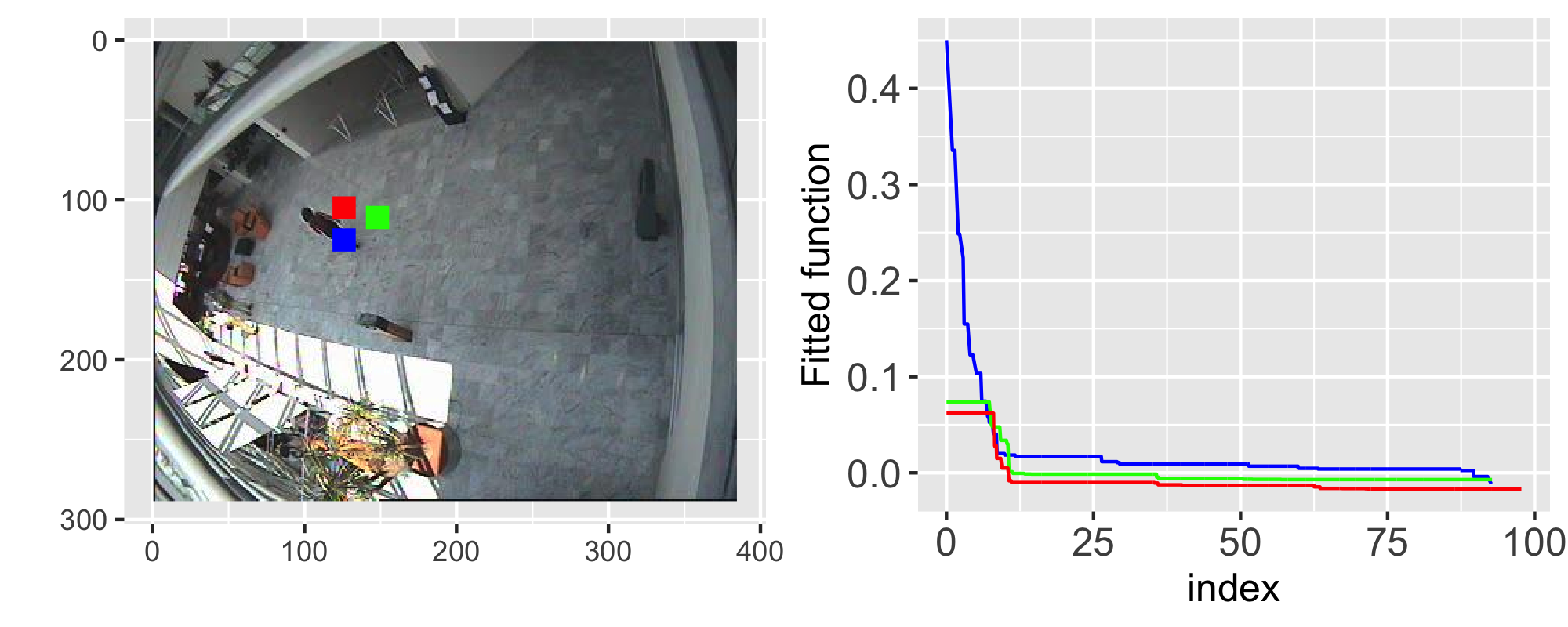}
  \caption[]{Left panel: input image overlaid with estimated location (SSCE, $\hat\theta$) of person based on the three color channels. Right Panel: isotonic estimates of the attenuation function at the estimated location ($\tilde\eta_{\hat\theta} = \hat\eta$) for each of the three color channels. }
  \label{fig:Fucntion_Detection}
\end{figure}

\begin{figure}[h!]
\centering
\includegraphics[width=.9\textwidth]{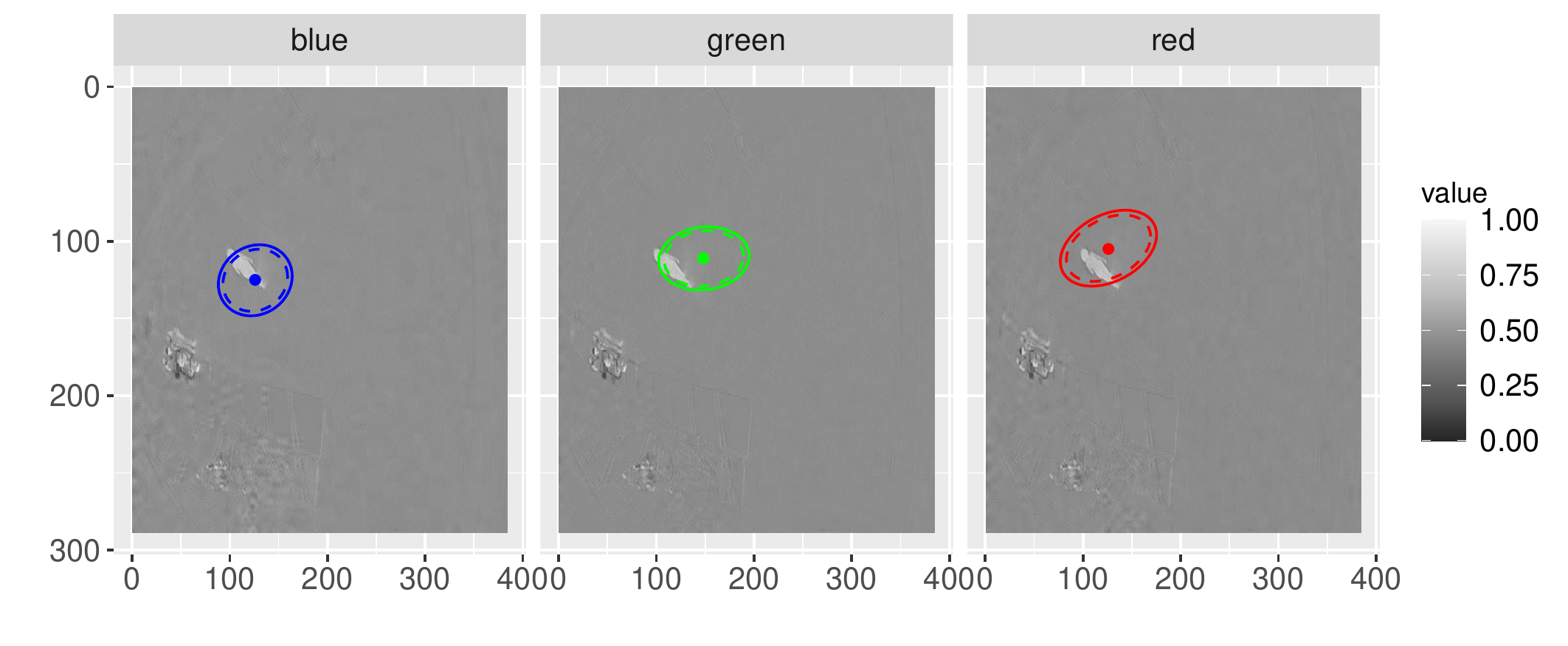}
  \caption[]{Ellipsoid confidence regions ($90\%$ and $95\%$) based on normal approximation with the variance estimated via $200$ $m$-out-of-$n$ bootstrap replicates with $m= \floor{n^{7/8}}$. }
  \label{fig:bootstrap}
\end{figure}

\section{Concluding remarks} 
\label{sec:concluding_remarks}
This paper proposed two tuning parameter free estimators (the LSE and SSCE) for the location of a target based on measurements acquired from distributed sensors using an index type regression model. We proved that the SSCE for the unknown location is $\sqrt{n}$ consistent and asymptotically normal under heavy-tailed and heteroscedastic errors. A numerical comparison between the SSCE and LSE reveals that proposed score estimator performs well in wide variety of settings. Unlike most work in the target detection literature, we do not assume a parametric model for the signal strength attenuation function. Further, the estimation procedure is completely automated and doesn't require any tuning parameters. These advantages make the proposed methodology applicable to a wide variety of problems that leverage sensing infrastructure in distributed systems.

We conclude by outlining some exciting future research directions. The rate of convergence of the LSE and limiting distribution of the monotone LSE is an open problem in the wider field of semi-parametric inference. Further, the current paper focuses only on locating a target for a fixed time point. Extending the current methodology to \textit{tracking} one or multiple targets is a challenging but important problem. When tracking a moving target, one can potentially ``combine'' the estimates of the attenuation function from different time points to provide an accurate and tuning parameter free estimate of the current location. 

\appendix
\addcontentsline{toc}{section}{Appendix} 
\part{Appendix} 
\parttoc
\section{Proof of Lemma~\ref{thm:Ident}} 
\label{sec:proof_of_thmIdent}
If we could prove that $\theta_0=\beta$, this would imply that $\eta_0(\cdot)\equiv g(\cdot)$ on  $\{|x-\theta_0|^2: x\in  \rchi \}$. Hence, it suffices to show that $\theta_0=\beta$. To show that $\theta_0=\beta$, we first notice that because of the convexity of $\rchi$, for small enough $L > 0$ we can find an open ball $B$ with radius $L$ included in $\rchi$ on which $x\mapsto \eta_0(|x- \theta_0|^2 )$ is not constant and 
 \begin{equation}\label{eq:tempa1}
 \eta_0(|\theta_0-x|^2) = g(|\beta-x|^2) \quad \text{ for all } x\in B.
 \end{equation}
Since, $x\mapsto \eta_0(|x- \theta_0|^2 )$ is not constant on $B, $ there  exists a point $b \in \{|x- \theta_0|^2 : x\in B\}$ and $\epsilon>\delta>0$ such that $b+\epsilon \in \{|x- \theta_0|^2 : x\in B\}$ and
\begin{equation}\label{eq:epsilion_b}
  t_0=:\eta_0(b+\epsilon) < \eta_0(b+\delta) := t_1 < \eta_0(b):=t_2.
  \end{equation}  Thus
\begin{equation}\label{eq:subset}
\phi \subsetneq \{x \in \rchi: \eta_0(|x-\theta_0|^2) \le t_0\} \cap B \subsetneq \{x \in \rchi: \eta_0(|x-\theta_0|^2) \le t_1\} \cap B \subsetneq B
\end{equation}
Observe that $\{x\in \R^d: \eta_0(|x-\theta_0|^2) \le t_0\}$ and $\{x\in \R^d: \eta_0(|x-\theta_0|^2) \le t_1\}$ are two distinct (by~\eqref{eq:subset}) concentric discs centered at $\theta_0$. Similarly, by~\eqref{eq:tempa1},  we have that  
$\{x\in \R^d: g(|x-\beta|^2)\le t_0\}$ and $\{x\in \R^d: g(|x-\beta|^2) \le t_1\}$ are two (by~\eqref{eq:subset}) concentric discs centered at $\beta$. They are distinct because
\begin{align*}\label{eq:t2}
\{x\in \R^d: \eta_0(|x-\theta_0|^2) \le t_0\}\cap B &= \{x\in \R^d: g(|x-\beta|^2) \le t_0\}\cap B\\
\{x\in \R^d: \eta_0(|x-\theta_0|^2) \le t_1\}\cap B &= \{x\in \R^d: g(|x-\beta|^2) \le t_1\}\cap B, 
\end{align*}
and $ \{x \in \rchi: \eta_0(|x-\theta_0|^2) \le t_1\}\cap B \neq B.$ Thus $\theta_0=\beta$. And the proof is complete. 


\section{Proof of Theorem~\ref{thm:unif_eta}}\label{sec:proof_unif_eta}
We will first prove~\eqref{eq:eta_bound_unif}. For each $\theta\in \Theta$, recall that \[\tilde\eta_\theta :=\argmin_{\eta \in \M}\Q_n(\eta, \theta).\] From~\citet[Theorem~1.3.4]{RWD88}, for any $q\le j \le n$, we have 
\[ \min_{1\le i \le n} Y_i \le \tilde{\eta}_\theta (|\theta-X_j|^2)\le \max_{1\le i \le n} Y_i.\]
Thus the proof of~\eqref{eq:eta_bound_unif} will be complete if we can show that  $\max_{1\le i \le n} |Y_i| =O_p\left(n^{1/q}\right).$ In this regard observe that for any $C>0$,
 \begin{align}
 \begin{split}
 \P\Big(\max_{1\le i \le n} |Y_i| \ge C n^{1/q}\Big) \le \sum_{i=1}^n\P\Big( |Y_i| \ge C n^{1/q}\Big) \le \sum_{i=1}^n\P\Big( |\epsilon_i| \ge C n^{1/q}/2\Big),
 \end{split}
 \end{align}
 for all $n$ such that $n^{1/q} \ge\|\eta_0\|_\infty$. By assumption~\ref{assum:err_mom} and Markov's inequality, we have
\begin{align}\label{eq:max_bound}
 \begin{split}
 \P\Big(\max_{1\le i \le n} |Y_i| \ge C n^{1/q}\Big)  \le \sum_{i=1}^n\P\Big( |\epsilon_i| \ge C n^{1/q}/2\Big) \le n \frac{2 ^q K_q^q}{C^q n } = 2^q K_q^q C^{-q}.
 \end{split}
 \end{align}
 The above upper bound converges to zero as $C\uparrow \infty$ for all large enough $n$ (independent of $C$).  Thus $\max_{1\le i \le n} |Y_i| =O_p\left(n^{1/q}\right).$

In the following we prove~\eqref{eq:unif_eta}. We will use the following newly defined quantities in the proof:
\begin{enumerate}

  \item $\M^{K}:= \{\eta:\D \to \R: \eta \text{ is non-decreasing and }   \|\eta\|_{\infty} \le K\}$
  \item $M_0:= \|\eta_0\|_\infty$

  \item $M_1 := \sup_{\theta\in B(\theta_0, r)} \|\eta_\theta\|_\infty$

\item $C_\epsilon:= 8 \E\big(\max_{1\le i\le n} |\epsilon_i|\big)\le 8\big[\E\big(\max_{1\le i\le n}|\epsilon_i|^q\big)\big]^{1/q} \le 8\left(n\E{|\epsilon|^q}\right)^{1/q} \le 8 \|\epsilon\|_q n^{1/q}. $
\end{enumerate}

 We will use arguments similar to Theorem 3.2.5 of~\cite{vdVaartWellner2000} and~\cite{balabdaoui2016least}. Recall that \begin{equation}\label{eq:temp_1}
\tilde{\eta}_\theta:=\argmin_{\eta\in \M} \Q_n(\eta, \theta) = \argmax_{\eta\in \M} \sum_{i=1}^{n} \bigg\{Y_i \eta\circ\theta(x_i) -\frac{1}{2}\eta^2\circ\theta(x_i)\bigg\}
\end{equation}
and
\begin{equation}\label{eq:temp_123}
{\eta}_\theta:=\argmin_{\eta\in \M} \Q(\eta, \theta) = \argmax_{\eta\in \M} \E \bigg\{Y\eta(|\theta-X|^2) -\frac{1}{2}\eta^2(|\theta-X|^2)\bigg\}
\end{equation}

Observe that  
\begin{align}\label{eq:q_diff}
\begin{split}
 &\Q_n(\eta, \theta) - \Q_n(\eta_\theta, \theta)\\
  ={}&  \sum_{i=1}^{n} \bigg\{Y_i \eta\circ\theta(x_i) -\frac{1}{2}\eta^2\circ\theta(x_i) - Y_i \eta_\theta\circ\theta(x_i) +\frac{1}{2}\eta_\theta^2\circ\theta(x_i)\bigg\}\\
 ={}&  \sum_{i=1}^{n} \bigg\{\epsilon_i \big(\eta\circ\theta(x_i)- \eta_\theta\circ\theta(x_i)\big) +  \eta_0\circ\theta_0(x_i) \big(\eta\circ\theta(x_i)- \eta_\theta\circ\theta(x_i)\big) -\frac{1}{2}\eta^2\circ\theta(x_i)  +\frac{1}{2}\eta_\theta^2\circ\theta(x_i)\bigg\}\\
 ={}&  \sum_{i=1}^{n} \epsilon_i \big(\eta\circ\theta(x_i)- \eta_\theta\circ\theta(x_i)\big) +\sum_{i=1}^{n}   \big(\eta_0\circ\theta_0(x_i)- \eta_\theta\circ\theta(x_i)\big) \big(\eta\circ\theta(x_i)- \eta_\theta\circ\theta(x_i)\big)\\&\qquad +\sum_{i=1}^{n}  -\frac{1}{2}\big(\eta\circ\theta(x_i)  - \eta_\theta\circ\theta(x_i)\big)^2.
\end{split}
\end{align}
We first will show that  for each $\theta$
\begin{equation}\label{eq:dist}
\Q(\eta, \theta) -\Q(\eta_\theta, \theta) = \frac{-1}{2} d_\theta^2(\eta, \eta_\theta) \lesssim -d_\theta^2(\eta, \eta_\theta),
\end{equation}
where for any $\eta_1, \eta_2 \in \M, $
\begin{equation}\label{eq:d_theta}
d_\theta^2(\eta_1, \eta_2) =\int \big(\eta_1(|\theta- X|^2)-\eta_2(|\theta- X|^2)\big)^2 d\P_0(X) = \int (\eta(t)-\eta_2(t))^2 f_{|\theta-X|^2} (t) dt.
\end{equation}
Observe that
\begin{align}\label{eq:pr_1}
\begin{split}
&\Q(\eta, \theta) -\Q(\eta_\theta, \theta) \\
={}& \E\left( Y\eta(|\theta-X|^2) -\frac{1}{2}\eta^2(|\theta-X|^2) - Y\eta_\theta(|\theta-X|^2) +\frac{1}{2}\eta_\theta^2(|\theta-X|^2) \right)\\
={}& \E\Bigg( \E\left[ Y\eta(|\theta-X|^2) -\frac{1}{2}\eta^2(|\theta-X|^2) - Y\eta_\theta(|\theta-X|^2) +\frac{1}{2}\eta_\theta^2(|\theta-X|^2)\bigg\vert|\theta-X|^2\right] \Bigg)\\
={}& \E\Bigg( \E\bigg(Y\Big\vert|\theta-X|^2 \bigg) \big[\eta(|\theta-X|^2)-\eta_\theta(|\theta-X|^2)\big] -\frac{1}{2}\eta^2(|\theta-X|^2) +\frac{1}{2}\eta_\theta^2(|\theta-X|^2)\Bigg)\\
={}& \E\Bigg(  \eta_\theta(|\theta-X|^2)\big[\eta(|\theta-X|^2)-\eta_\theta(|\theta-X|^2)\big] -\frac{1}{2}\eta^2(|\theta-X|^2) +\frac{1}{2}\eta_\theta^2(|\theta-X|^2)\Bigg)\\
={}&-\frac{1}{2} d_\theta^2(\eta, \eta_\theta).
\end{split}
\end{align}
Let us fix $r>0$ and let $\phi_n : \R^+\to \R^+$, be a function such that
\begin{equation}\label{eq:325_main}
\sqrt{n} \E^* \sup_{ \theta \in B(\theta_0, r) } \sup_{\eta \in \M_\theta^K(\delta)}\Big| (\Q_n-\Q)(\eta, \theta)- (\Q_n-\Q)(\eta_\theta, \theta)\Big| \lesssim  \phi_n(\delta),
\end{equation}
where \begin{equation}\label{eq:M_dleta}
\M_\theta^K(\delta):= \{\eta: \|\eta\|_\infty <K\text{ and } d_\theta(\eta, \eta_\theta) \le \delta \}
\end{equation}
and there exists an $\alpha <2$ such that $\delta \mapsto \phi_n(\delta)/\delta^\alpha$ is decreasing and $r_n^2\phi(1/r_n) \le \sqrt{n}.$ Our goal is to show that
\[\sup_{\theta\in B(\theta_0, r)} d_\theta(\tilde\eta_\theta, \eta_\theta) =O_p\left(n^{-1/3} n^{1/q}\right).\]
Note that
\begin{align}
&P\Big( r_n \sup_{\theta\in B(\theta_0, r)} d_\theta(\tilde\eta_\theta, \eta_\theta) >\delta \Big)\label{eq:p1}\\
 \le{}& \P(\sup_{\theta\in B(\theta_0, r)}\|\tilde{\eta}_\theta\|\ge K)+ \sum_{j>M} \P\Big(  2^{j-1} \delta< r_n \sup_{\theta\in B(\theta_0, r)} d_\theta(\tilde\eta_\theta, \eta_\theta) \le 2^j \delta \text{ and } \sup_{\theta\in B(\theta_0, r)}\|\tilde{\eta}_\theta\|\le K\Big).\nonumber
\end{align}
By~\eqref{eq:eta_bound_unif}, we can make the first probability on the right hand side of the equation small by choosing $K=C n^{1/q}$ for an appropriate choice of $C.$
We will now try to bound the second probability.
Now observe that $ 2^{j-1}\delta < r_n \sup_{\theta\in B(\theta_0, r)} d_\theta(\tilde\eta_\theta, \eta_\theta) \le 2^j\delta$ implies that   there exists $\theta_a \in B(\theta_0, r)$ such that  $2^{j-1}\delta< r_n d_{\theta_a}(\tilde{\eta}_{\theta_a}, \eta_{\theta_a}) \le 2^{j}\delta$. Now let us define a set
\begin{equation}\label{eq:set_a}
\mathcal{A}_{n, j} := \Big\{(\eta, \theta)\, :\, \theta\in B(\theta_0, r), \|\eta\|_\infty \le K, \text{ and } 2^{j-1}\delta < r_n d_\theta(\eta, \eta_\theta) \le 2^j\delta \Big\}.
\end{equation}
As $(\theta_a, \tilde{\eta}_{\theta_a}) \in \A_{n, j}$, we have that
 \begin{equation}\label{eq:p2}
 2^{j-1}\delta < r_n \sup_{\theta\in B(\theta_0, r)} d_\theta(\tilde\eta_\theta, \eta_\theta) \le 2^j\delta \Rightarrow \sup_{(\eta, \theta) \in \A_{n, j}} \Q_n(\eta, \theta)- \Q_n(\eta_\theta, \theta) \ge \Q_n(\tilde{\eta}_{\theta_a}, \theta_a)- \Q_n(\eta_{\theta_a}, \theta_a) \ge 0.
 \end{equation}
Moreover, by~\eqref{eq:dist} we have that
\begin{equation}\label{eq:p3}
 \sup_{(\eta, \theta) \in \A_{n, j}}  \Q(\eta, \theta)- \Q(\eta_\theta, \theta) = \sup_{(\eta, \theta) \in \A_{n, j}}   \frac{-1}{2} d_\theta^2(\eta, \eta_\theta)= \frac{-2^{2j -2} \delta^2}{2 r_n^2}.
\end{equation}
Combining~\eqref{eq:p2} and~\eqref{eq:p3}, we have that if $2^{j-1} \delta< r_n \sup_{\theta\in B(\theta_0, r)} d_\theta(\tilde\eta_\theta, \eta_\theta) \le 2^j\delta$ then
\[0 \le \sup_{(\eta, \theta) \in \A_{n, j}}  (\Q_n-\Q)(\eta, \theta)- (\Q_n-\Q)(\eta_\theta, \theta) + \sup_{(\eta, \theta) \in \A_{n, j}}  \Q(\eta, \theta)- \Q(\eta_\theta, \theta).\]
Thus
\begin{equation}\label{eq:pd}
2^{j-1}\delta < r_n \sup_{\theta\in B(\theta_0, r)} d_\theta(\tilde\eta_\theta, \eta_\theta) \le 2^j\delta \Rightarrow \sup_{(\eta, \theta) \in \A_{n, j}}  (\Q_n-\Q)(\eta, \theta)- (\Q_n-\Q)(\eta_\theta, \theta) \ge \frac{2^{2j -2}\delta^2}{2 r_n^2}
\end{equation}
Now combining~\eqref{eq:p1} and~\eqref{eq:pd}, we have that
\begin{align}
&P\Big( r_n \sup_{\theta\in B(\theta_0, r)} d_\theta(\tilde\eta_\theta, \eta_\theta) > \delta\Big) \\
={}&\P\left( \exists j \ge 1  \sup_{(\eta, \theta) \in \A_{n, j}}  \Big[(\Q_n-\Q)(\eta, \theta)- (\Q_n-\Q)(\eta_\theta, \theta)\Big] \ge \frac{2^{2j -2}\delta^2}{2 r_n^2}\right)\\
\le{}&  \sum_{j\ge 1}  \frac{2 r_n^2}{\sqrt{n}2^{2j -2}\delta^2} \E \left(\sup_{(\eta, \theta) \in \A_{n, j}}  \sqrt{n}\Big|(\Q_n-\Q)(\eta, \theta)- (\Q_n-\Q)(\eta_\theta, \theta)\Big|  \right)\\
\le{}&  \sum_{j\ge 1}  \frac{2 r_n^2}{\sqrt{n}2^{2j -2}\delta^2} \E \left(\sup_{(\eta, \theta) \in \F_j}  \sqrt{n}\Big|(\Q_n-\Q)(\eta, \theta)- (\Q_n-\Q)(\eta_\theta, \theta)\Big|  \right), \label{eq:pd1}
\end{align}
where
\begin{equation}\label{eq:F_def}
\F_j :=\cup_{i=1}^j \A_{n, i} =  \{(\eta, \theta) : \theta \in B(\theta_0, r) \text{ and } \eta \in \M_\theta^K (2^j \delta/r_n)\}.
\end{equation}
We will now compute an upper bound for the expectation in  the display. Observe that 
\begin{align}\label{eq:g_equiv}
\begin{split}
&\sqrt{n}\big|(\Q_n-\Q)(\eta, \theta)- \sqrt{n}(\Q_n-\Q)(\eta_\theta, \theta)\big|\\
\le{}& \big|\G_n\epsilon \big(\eta\circ\theta- \eta_\theta\circ\theta\big)\big| + \big| \G_n\big[\big(\eta_0\circ\theta_0(x_i)- \eta_\theta\circ\theta(x_i)\big) \big(\eta\circ\theta(x_i)- \eta_\theta\circ\theta(x_i)\big)\big]\big|+ \frac{1}{2}\big|\G_n \big(\eta\circ\theta(x_i)  - \eta_\theta\circ\theta(x_i)\big)^2\big|\\
\end{split}
\end{align}
Thus we have that
\begin{align}
&\E \left(\sup_{(\eta, \theta) \in \F_j}  \sqrt{n}\Big|(\Q_n-\Q)(\eta, \theta)- (\Q_n-\Q)(\eta_\theta, \theta)\Big|  \right)\\
\le{}&\E \left(\sup_{(\eta, \theta) \in \F_j} \big|\G_n\epsilon \big(\eta\circ\theta- \eta_\theta\circ\theta\big)\big| \right)+ \E \left(\sup_{(\eta, \theta) \in \F_j} \big| \G_n\big[\big(\eta_0\circ\theta_0(x_i)- \eta_\theta\circ\theta(x_i)\big) \big(\eta\circ\theta(x_i)- \eta_\theta\circ\theta(x_i)\big)\big]\big|\right)\\\
&\hspace{2.4in}+ \frac{1}{2}\E \left(\sup_{(\eta, \theta) \in \F_j} \big|\G_n \big(\eta\circ\theta(x_i)  - \eta_\theta\circ\theta(x_i)\big)^2\big|\right).\label{eq:split_123}
\end{align}
Next we give upper bounds for each of the terms in the right of~\eqref{eq:split_123}. As 
\[ \sup_{(\eta, \theta)\in \F_j} \|\big(\eta\circ\theta- \eta_\theta\circ\theta\big)\|_\infty \le M_1+K,
\] by  Lemma~S.5.1 of~\cite{2017arXiv170800145K}, we have that 
\begin{equation}\label{eq:eps_term}
\E \left(\sup_{(\eta, \theta) \in \F_j} \big|\G_n\epsilon \big(\eta\circ\theta- \eta_\theta\circ\theta\big)\big| \right) \le \E \left(\sup_{(\eta, \theta) \in \F_j} \big|\G_n\bar\epsilon \big(\eta\circ\theta- \eta_\theta\circ\theta\big)\big| \right) + 2 \frac{(M_1+ K)C_\epsilon}{\sqrt{n}},
\end{equation}
where for $i=1, \ldots, n$
\begin{equation}\label{eq:ebar_def}
\bar{\epsilon}_i := \epsilon \mathbbm{1}_{|\epsilon|\le C_\epsilon} \qquad \text{and} \qquad \epsilon_i^* := \epsilon_i -\bar\epsilon_i.
\end{equation}

We will bound the last expectation on the right of~\eqref{eq:split_123} via symmetrization and contraction (Theorem 3.1.21 and Corollary 3.2.2 of \cite{Gine16}). Note that 
\begin{align}\label{eq:3rdterm}
\begin{split}
\E\bigg(\sup_{\F_j} \bigg|\G_n \big(\eta\circ\theta  - \eta_\theta\circ\theta\big)^2\Big|\bigg) &\le 2 \E\bigg(\sup_{\F_j} \bigg|\P_n R \big(\eta\circ\theta  - \eta_\theta\circ\theta\big)^2\Big|\bigg) \\
&\le  8 \sup_{\F_j}\|\eta\circ\theta  - \eta_\theta\circ\theta\|_{\infty} \E\bigg(\sup_{\F_j} \bigg|\P_n R \big(\eta\circ\theta  - \eta_\theta\circ\theta\big)\Big|\bigg)\\
&=  8 \sup_{\F_j}\|\eta\circ\theta  - \eta_\theta\circ\theta\|_{\infty} \E\bigg(\sup_{\F_j} \bigg|\G_n R \big(\eta\circ\theta  - \eta_\theta\circ\theta\big)\Big|\bigg),
\end{split}
\end{align}
here $R_1, \ldots, R_n$ are i.i.d.~Rademacher random variables (i.e., $\P(R=1)= \P(R=0) = 1/2$) independent of $(X_i, \epsilon_i)_{i=1}^n$. 

Let
\[\F(\gamma) := \{(\eta, \theta) : \theta \in B(\theta_0, r) \text{ and } \eta \in \M_\theta^K (\gamma)\}.\]
By  combining~\eqref{eq:split_123} and~\eqref{eq:3rdterm} and Lemma~\ref{lem:outsideterm}, we have that 
\begin{align}\label{eq:final}
\begin{split}
&\E \left(\sup_{(\eta, \theta) \in \F(\gamma)}  \sqrt{n}\Big|(\Q_n-\Q)(\eta, \theta)- (\Q_n-\Q)(\eta_\theta, \theta)\Big|  \right)\\
\lesssim{}&  \sigma [K\gamma]^{1/2} + \frac{\sigma^2 \gamma^{-1}K^2 C_\epsilon}{\sqrt{n}}  + \frac{KC_\epsilon}{\sqrt{n}} + K [K\gamma ]^{{1/2}} + \frac{K^3\gamma^{-1}}{\sqrt{n}}+ [K\gamma ]^{{1/2}} + \frac{K^2\gamma^{-1}}{\sqrt{n}} \\
\lesssim{}& \sigma K^{3/2} \gamma^{1/2}  + \frac{K^2(K+C_\epsilon) \gamma^{-1}}{\sqrt{n}} + \frac{KC_\epsilon}{\sqrt{n}}
\end{split}
\end{align}
Let $\gamma_j = 2^j \delta/r_n$, $C_\epsilon=  C n^{1/q}$, and $K = C n^{1/q}$, then 
\begin{align}
&P\Big( r_n \sup_{\theta\in B(\theta_0, r)} d_\theta(\tilde\eta_\theta, \eta_\theta) > \delta\Big) \\
\le{}& \frac{ 8r_n^2}{\sqrt{n}\delta^2} \sum_{j\ge 1}  \frac{1 }{2^{2j }} \E \left(\sup_{(\eta, \theta) \in \F(\gamma_j)}  \sqrt{n}\Big|(\Q_n-\Q)(\eta, \theta)- (\Q_n-\Q)(\eta_\theta, \theta)\Big|  \right)\\
\lesssim{}& \frac{ 8r_n^2}{\sqrt{n}\delta^2}\Bigg[\sigma K^{3/2}  \sum_{j\ge 1}  \frac{1 }{2^{2j }}  \left[2^j \delta/r_n\right]^{1/2}    + \frac{K^2(K+C_\epsilon)}{\sqrt{n}} \sum_{j\ge 1}  \frac{1 }{2^{2j }}\left[2^j \delta/r_n\right]^{-1} +\frac{KC_\epsilon}{\sqrt{n}}\sum_{j\ge 1}  \frac{1 }{2^{2j }}\Bigg]\\
\le{}& \frac{ 8r_n^2}{\sqrt{n}\delta^2}\Bigg[\sigma K^{3/2} \left[\frac{\delta}{r_n}\right]^{1/2} \sum_{j\ge 1}  \frac{1 }{2^{2j }}  \left[2^j\right]^{1/2}    + \frac{K^2(K+C_\epsilon)}{\sqrt{n}} \frac{r_n}{\delta}\sum_{j\ge 1}  \frac{1 }{2^{2j }}\left[2^j\right]^{-1} +\frac{KC_\epsilon}{\sqrt{n}}\sum_{j\ge 1}  \frac{1 }{2^{2j }}\Bigg]\\
\le{}& \frac{ 8r_n^2}{\sqrt{n}\delta^2}\Bigg[\sigma K^{3/2} \left[\frac{\delta}{r_n}\right]^{1/2}     + \frac{K^2(K+C_\epsilon)}{\sqrt{n}} \frac{r_n}{\delta} +\frac{KC_\epsilon}{\sqrt{n}}\Bigg]\\
={}& \sigma K^{3/2} \frac{ 8r_n^{3/2}}{\sqrt{n}\delta^{3/2}}     + \frac{ 8r_n^3}{{n}\delta^3}K^2(K+C_\epsilon)  +\frac{8 KC_\epsilon r_n^2 }{n \delta^2}\\
\le{}& \sigma n^{3/2q} \frac{ 8r_n^{3/2}}{\sqrt{n}\delta^{3/2}}     + \frac{ 8r_n^3}{{n}\delta^3}n^{3/q}  +\frac{8 n^{2/q} r_n^2 }{n \delta^2}
\end{align}

Thus if $r_n = n^{1/3 -1/q}$, then 
\[P\Big( r_n \sup_{\theta\in B(\theta_0, r)} d_\theta(\tilde\eta_\theta, \eta_\theta) > \delta\Big) \lesssim \frac{8 \sigma}{\delta^{3/2}} + \frac{8}{\delta^3} + \frac{8n^{-1/3}}{\delta^2}.\]

\begin{lemma}\label{lem:outsideterm}
Let
\[\F(\gamma) := \{(\eta, \theta) : \theta \in B(\theta_0, r) \text{ and } \eta \in \M_\theta^K (\gamma)\}, \]
then 
\begin{align}\label{eq:outside_lem}
\begin{split}
\E\bigg(\sup_{\F(\gamma)} \Big|\G_n\bar\epsilon \big(\eta\circ\theta- \eta_\theta\circ\theta\big)\Big|\bigg) \lesssim{}& \sigma [K\gamma]^{1/2} + \frac{\sigma^2 \gamma^{-1}K^2 C_\epsilon}{\sqrt{n}}, \\
\E\bigg(\sup_{\F(\gamma)} \Big|\G_n R \big(\eta\circ\theta- \eta_\theta\circ\theta\big)\Big|\bigg) \lesssim{}& \left[K \gamma\right]^{1/2} +   \frac{ K^2 \gamma^{-1} }{\sqrt{n} }, \\
\E\bigg(\sup_{\F(\gamma)} \Big|\G_n  \big(\eta_0\circ\theta_0- \eta_\theta\circ\theta\big) \big(\eta\circ\theta -\eta_\theta\circ\theta\big)\Big|\bigg)\lesssim{}& [K\gamma ]^{{1/2}} + \frac{K^2\gamma^{-1}}{\sqrt{n}}.
\end{split}
\end{align}
\end{lemma}
\begin{proof}
Note that $\epsilon \big(\eta\circ\theta- \eta_\theta\circ\theta\big)$, $R \big(\eta\circ\theta- \eta_\theta\circ\theta\big)$, $\big(\eta_0\circ\theta_0- \eta_\theta\circ\theta\big) \big(\eta\circ\theta- \eta_\theta\circ\theta\big)$ are uniformly bounded by $C_\epsilon (M_1 + K)$, $C_\epsilon (M_1 + K)$, and $(M_1 + K) (M_1+ M_0)$, respectively. In Lemma~\ref{lem:basic_joint}, we show that 
\begin{equation}\label{eq:entropy_eta_diff}
\log N_{[\, ]}(\nu, \{\eta\circ\theta -\eta_\theta\circ\theta:  (\eta, \theta) \in \F(\gamma)\}, \|\cdot \|)\le \frac{2A (M_1+K)}{\nu}
\end{equation}
and 
\begin{align}
\begin{split}\label{eq:eta_theta_ent}
  &\log N_{[\, ]}(\nu, \{ \eta_\theta\circ\theta:  (\eta, \theta) \in \F(\gamma)\}, \|\cdot \|)\\
={}&\log N_{[\, ]}(\nu, \{\eta_0\circ\theta_0 - \eta_\theta\circ\theta:  (\eta, \theta) \in \F(\gamma)\}, \|\cdot \|)\le \frac{A M_1}{\nu}.
\end{split}
\end{align}
Thus it is clear that 
\[\log N_{[\, ]}(\nu, \{\bar\epsilon(\eta\circ\theta -\eta_\theta\circ\theta):  (\eta, \theta) \in \F(\gamma)\}, \|\cdot \|)\le \frac{2A (M_1+ K) \sigma}{\nu},
\]
as $\|\E(\bar\epsilon^2|X=\cdot)\|_\infty\le \sigma^2 $  and
\[\log N_{[\, ]}(\nu, \{R(\eta\circ\theta -\eta_\theta\circ\theta):  (\eta, \theta) \in \F(\gamma)\}, \|\cdot \|)\le \frac{2A (M_1+K) }{\nu}.
\]
Finally, in Lemma~\ref{lem:basic_joint}, we show that 
\[\log N_{[\, ]}(\nu, \{ \big(\eta_0\circ\theta_0- \eta_\theta\circ\theta\big) \big(\eta\circ\theta -\eta_\theta\circ\theta\big):  (\eta, \theta) \in \F(\gamma)\}, \|\cdot \|)\le  \frac{8A(M_1+K) M_1}{\nu}, \]
Then by Lemma~3.4.2 of~\cite{VdVW96} we have that 
\begin{align}\label{eq:Ep_outside_term}
\begin{split}
\E\bigg(\sup_{\F(\gamma)} \Big|\G_n\bar\epsilon \big(\eta\circ\theta- \eta_\theta\circ\theta\big)\Big|\bigg) 
\lesssim{}&  {\sigma} \sqrt{A(M_1+K)}\gamma^{1/2}  \left(1+\frac{{\sigma} \sqrt{A(M_1+K)} \gamma^{1/2}C_\epsilon (M_1+ K)}{\gamma^{2} \sqrt{n} } \right)\\
={}&  {\sigma} \sqrt{A(M_1+K)}\gamma^{1/2}  + {\sigma}^2 A(M_1+K) \frac{ \gamma^{-1}C_\epsilon (M_1+ K)}{ \sqrt{n} } \\
\lesssim{}& \sigma [K\gamma]^{1/2} + \frac{\sigma^2 \gamma^{-1}K^2 C_\epsilon}{\sqrt{n}}
\end{split}
\end{align}
and 
\begin{align}\label{eq:R_outside_term}
\begin{split}
\E\bigg(\sup_{\F(\gamma)} \Big|\G_n R \big(\eta\circ\theta- \eta_\theta\circ\theta\big)\Big|\bigg)
\lesssim& \left[\sqrt{A(M_1+K)} \gamma\right]^{1/2}  \left(1+\frac{ \sqrt{A(M_1+K)} \gamma^{1/2} (M_1+ K)}{\gamma^{2} \sqrt{n} } \right)\\
\lesssim& \left[K \gamma\right]^{1/2} +   \frac{ K^2 \gamma^{-1} }{\sqrt{n} }
\end{split}
\end{align}
and
\begin{align}
\begin{split}
&\E\bigg(\sup_{\F(\gamma)} \Big|\G_n  \big(\eta_0\circ\theta_0- \eta_\theta\circ\theta\big) \big(\eta\circ\theta -\eta_\theta\circ\theta\big)\Big|\bigg)\\ 
\lesssim{}& \sqrt{A(M_1+K)M_1} \left[\gamma (M_0+ M_1)\right]^{1/2}  \left(1+\frac{ \left[\gamma(M_0+ M_1)\right]^{1/2} \sqrt{A(M_1+K)M_1} (M_1+ K) (M_0+ M_1)}{\left[\gamma(M_0+ M_1)\right]^{2} \sqrt{n} } \right)\\
={}& \sqrt{A(M_1+K)M_1} \left[\gamma (M_0+ M_1)\right]^{1/2} +  \frac{ \left[\gamma(M_0+ M_1)\right]^{-1} A(M_1+K)M_1 (M_1+ K) (M_0+ M_1)}{ \sqrt{n} } \\
={}& \sqrt{A(M_1+K)M_1} \left[\gamma (M_0+ M_1)\right]^{1/2} +  \frac{ \left[\gamma(M_0+ M_1)\right]^{-1} A(M_1+K)M_1 (M_1+ K) (M_0+ M_1)}{ \sqrt{n} }\\
\lesssim{}& [K\gamma ]^{{1/2}} + \frac{K^2\gamma^{-1}}{\sqrt{n}}. \qedhere
\end{split}
\end{align}
\end{proof}

\section{Entropy Calculations}\label{sec:ent_calc}
The following lemma, proves~Lemma~\ref{thm:entropy_FK} and finds metric entropies of other related function classes. 
\begin{lemma}\label{lem:basic_joint} If $\nu >0, $ then
\begin{align}
\log N_{[\, ]}(\nu, \{\eta\circ\theta : \theta \in B(\theta_0, r) \text{ and } \eta \in \M^K \}, \|\cdot \|)&\le  \frac{AK}{\nu}\label{eq:joint_ent}\\
\log N_{[\, ]}(\nu, \{\eta_\theta\circ\theta : \theta \in B(\theta_0, r)\}, \|\cdot \|)&\le  \frac{AM_1}{\nu}, \label{eq:eta_theta_ent1}\\
\log N_{[\, ]}(\nu, \{\eta\circ\theta -\eta_\theta\circ\theta: \theta \in B(\theta_0, r) \text{ and } \eta \in \M^K \}, \|\cdot \|)&\le  \frac{2A(M_1+K)}{\nu}, \label{eq:ent_diff}
\end{align}
where $A$ is a universal constant depending only on $\D$ and $d$ (the dimension of $X$). Further, let \[\A^K := \{\big(\eta_0\circ\theta_0- \eta_\theta\circ\theta\big) \big(\eta\circ\theta -\eta_\theta\circ\theta\big) : \theta \in B(\theta_0, r) \text{ and } \eta \in \M^K \}, \]
then 
\begin{equation}\label{eq:ak_ent}
\log N_{[\, ]}(\nu, \A^{K}, \|\cdot \|)\le  \frac{8A(M_1+K) M_1}{\nu},
\end{equation}
\end{lemma}
\begin{proof}
By triangle inequality, we have that 
\[\big||\theta_2-x| -|\theta_1 -x|\big| \le |\theta_1-\theta_2|.\]
Thus the proof of~\eqref{eq:joint_ent}--\eqref{eq:ent_diff} follows directly from Lemma K.1 of~\cite{2017arXiv170800145K}; also see~Lemma~4.9 of~\cite{balabdaoui2016least}. We will now prove~\eqref{eq:ak_ent}.  Consider the following two classes of functions:
\begin{align}\label{eq:parts_pdt}
\begin{split}
 \A^K_1 &:= \{  \big(\eta\circ\theta -\eta_\theta\circ\theta\big)\eta_0\circ\theta_0:  \theta\in B(\theta_0, r) \text{ and } \eta\in \M^K\}\\
\A^K_2 &:= \{ \big(\eta\circ\theta -\eta_\theta\circ\theta\big)\eta_\theta\circ\theta :  \theta\in B(\theta_0, r) \text{ and } \eta\in \M^K\}
\end{split}
\end{align}
Note that $\|\eta_0\circ\theta_0\|_{\infty} \le M_0$. Thus by \eqref{eq:ent_diff}, we have that 
\begin{equation}\label{eq:ent_A1}
\log N_{[\, ]}(\nu, \A^K_1, \|\cdot\|) \le \frac{2A (M_1+K)M_0}{\nu}.
\end{equation}
Now we will compute the entropy of $\A_2^K.$ Note that for any $\theta\in B(\theta_0, r)$ and $\eta\in \M_K$ $f(t):= \eta(t)\eta_\theta(t)$ is a monotone function on $\D\to\R$ and bounded by $K M_1.$
Thus by Lemma~\ref{lem:basic_joint}, we have that 
\begin{equation}\label{eq:ent_A1_1}
\log N_{[\, ]}(\nu, \{\eta_\theta\circ\theta\; \eta\circ\theta:  \theta\in B(\theta_0, r) \text{ and } \eta\in \M^K\}, \|\cdot\|) \le \frac{A M_1K}{\nu}.
\end{equation}
Similarly, we have that 
\begin{equation}\label{eq:ent_A2_1}
\log N_{[\, ]}(\nu, \{(\eta_\theta\circ\theta)^2:  \theta\in B(\theta_0, r) \text{ and } \eta\in \M^K\}, \|\cdot\|) \le \frac{A M_1^2}{\nu}.
\end{equation}
Thus 
\begin{align}\label{eq:ent_A2}
\begin{split}
\log N_{[\, ]}(\nu, \A^K_2, \|\cdot\|)\le{}& \log N_{[\, ]}(\nu/2, \{(\eta_\theta\circ\theta)^2:  \theta\in B(\theta_0, r) \text{ and } \eta\in \M^K\}, \|\cdot\|)\\ &+ \log N_{[\, ]}(\nu/2, \{(\eta_\theta\circ\theta)^2:  \theta\in B(\theta_0, r) \text{ and } \eta\in \M^K\}, \|\cdot\|)\\
\le{}&  \frac{2A M_1K}{\nu} + \frac{2A M_1^2}{\nu} = \frac{2A (M_1+K)M_1}{\nu} 
\end{split}
\end{align}
Finally as $\A^K \subset \A^K_1 + \A^K_2$, we have 
\begin{align}
\log N_{[\, ]}(\nu, \A^K, \|\cdot\|) &\le \log N_{[\, ]}(\nu/2, \A^K_1, \|\cdot\|)+ \log N_{[\, ]}(\nu/2, \A^K_2, \|\cdot\|)\\ &\le\frac{4A (M_1+K)M_0}{\nu}+  \frac{4A (M_1+K)M_1}{\nu}\\
&\le \frac{8A (M_1+K)M_1}{\nu}\qedhere
\end{align}

\end{proof}

\section{Proof of Theorem~\ref{thm:LSE_joint}} 
\label{sec:proof_of_theorem_LSE_joint}
\paragraph{Proof of joint rate of the LSE} 
\label{par:joint_rate_}

We will first show that for every $q\ge 5$,  we have 
\begin{equation}\label{eq:CombinedRate}
 \|\check{\eta} (|\check{\theta}- X|^2)- {\eta_0} (|{\theta_0}- X|^2)\| = O_p\left( n^{-1/3} n^{1/q}\right).
\end{equation} 
We will prove~\eqref{eq:CombinedRate} via an application of Theorem~2.1 of~\cite{2019arXiv190902088K}. Fix $C>0, $ define $\widetilde{\F}_n := \F_{C n^{1/q}}$, where for any $K>0$, $\F_K$ is defined as in~Lemma~\ref{thm:entropy_FK}.  Now define, $f^\dagger := \argmin_{f\in \widetilde{\F}_n} \sum_{i=1}^n (Y_i - f(X_i))^2$.  Observe that by Theorem~\ref{thm:unif_eta}, we have that $\P(\check\eta(|\check{\theta} - \cdot|^2) \notin \widetilde{\F}_n) =o(1).$ Thus $\P(f^\dagger(\cdot) \equiv\check{\eta}(|\check \theta - \cdot|^2) ) = 1- o(1)$ and the rate of convergence of $f^{\dagger}(\cdot)$ and $\check{\eta}(|\check \theta - \cdot|^2)$ coincide.  To complete the proof of~\eqref{eq:CombinedRate}, we will find the rate of convergence of $f^{\dagger}$ by applying Theorem~2.1 of~\cite{2019arXiv190902088K}. Note that $\widetilde{\F}_n$ satisfies the assumption of Theorem~2.1 of~\cite{2019arXiv190902088K} with $A= C n^{1/q}$, $\alpha=1$, $s=0$ and $\Phi = \max\{K_q, C n^{1/q}\} $. Thus we have that 
\[ \|f^{\dagger}(X) - \eta_0 (|\theta_0 -X|^2)\| = O_p \left(\max \left\{n^{-1/3} n^{1/q}, n^{- (q-3)/ 2q}, n^{- (q- 1- (3q-1)/q)/(2q-1)}\right\} \right).\]
Now observe that $n^{- (q- 1- (3q-1)/q)/(2q-1)} \le n^{-1/3} n^{1/q}$ when $q\ge 5$ and $ n^{- (q-3)/ 2q}\le n^{-1/3} n^{1/q} $ when $q\ge 3$. Thus we have~\eqref{eq:CombinedRate} when $q\ge 5.$

\paragraph{Consistency of the separated parameters:} 
\label{par:consistency_of_the_separated_parameters_} We will now use the above result to prove that $|\check{\theta}-\theta_0|=o_p(1)$. The following argument is similar to the proof of Theorem~5.2 of~\cite{balabdaoui2016least}. To show the dependence of $n$ in the definition of $\check{\eta}$ and $\check{\theta}$, we will use $\check{\eta}_n$ and $\check{\theta}_n$, respectively. By~\eqref{eq:CombinedRate}, we have that for every subsequence $\{n_k\}$, there exists a further subsequence $\{n_{k_l}\}$ such that $\|\check{\eta}_{n_{k_l}}(|\check{\theta}_{n_{k_l}}- X|^2)- {\eta_0} (|{\theta_0}- X|^2)\|$ converges to $0$ almost surely; see Theorem 2.3.2 of~\cite{Durrett}. 

Now fix $\omega$.\footnote{All the subsequences used in the following arguments depend on $\omega$. } Since we will argue along the subsequences, to avoid this messy subsequence notation, in what follows, we will assume without loss of generality that $\|\check{\eta}_{n}(|\check{\theta}_{n}- X|^2)- {\eta_0} (|{\theta_0}- X|^2)\|$ converges to $0$ almost surely.    We will use compactness based arguments to show that $\check{\theta}_n$ is consistent.  Define $\bar\eta_n: \R^+ \to \R^+$ as follows:
\begin{equation}\label{eq:eta_trunc}
\bar\eta_n(t) := \begin{cases}
  \check\eta_n(t) &\text{ if } \check\eta_n(t) \in [0, \|\eta_0\|_{\infty}]\\
  \|\eta_0\|_{\infty} &\text{otherwise.} 
\end{cases}
\end{equation}
Now recall that the space of bounded, monotone, left-continuous functions are compact under pointwise convergence.  Moreover, since $\Theta$ is compact. Let $m_0$ and $\alpha_0$ be a limit points of $\{\check{\eta}_n\}$ and $\{\check{\theta}_n\}$, respectively. If we can show that
\begin{equation}\label{eq:limitpoint}
 \|m_0(|\alpha_0- X|^2)- {\eta_0} (|{\theta_0}- X|^2)\| =0,
 \end{equation}
  identifiability (Lemma~\ref{thm:Ident}) and monotonicity of $\eta_0$ and $m_0$ will imply that all limit points of $\{\check{\eta}_n\}$ and $\{\check{\theta}_n\}$ are $\eta_0$ and $\theta_0$, respectively. We will now prove~\eqref{eq:limitpoint}. By triangle inequality, we conclude that
  \begin{align}\label{eq:limitsplit}
  \|m_0(|\alpha_0- X|^2)- {\eta_0} (|{\theta_0}- X|^2)\| &\le \|m_0(|\alpha_0- X|^2)- {m_0} (|{\check\theta_n}- X|^2)\|\\
  &\qquad + \|m_0(|\check\theta_n- X|^2)- \bar{\eta}_n (|{\check\theta_n}- X|^2)\|+ \|\bar{\eta}_n (|{\check\theta_n}- X|^2)- {\eta_0} (|{\theta_0}- X|^2)\|. 
  \begin{split}  
  \end{split}
  \end{align}
  We now provide a bound for the first term. Since $X$ has density with respect to the Lebesgue measure (by \ref{a1}) and $m_0$ has only countably many discontinuities, we have that  $\|m_0(|\alpha_0- X|^2)- {m_0} (|{\check\theta_n}- X|^2)\| \to 0$ along a subsequence as $n\to \infty. $  For the third term, observe that 
  \[\|\bar{\eta}_n (|{\check\theta_n}- X|^2)- {\eta_0} (|{\theta_0}- X|^2)\|  \le \|\check{\eta}_n (|{\check\theta_n}- X|^2)- {\eta_0} (|{\theta_0}- X|^2)\|, \]
by definition of $\bar\eta_n$. Thus $\|\bar{\eta}_n (|{\check\theta_n}- X|^2)- {\eta_0} (|{\theta_0}- X|^2)\| \to 0$ along a subsequence as $n\to \infty.$ For the second term, observe that $\bar\eta_n$ converges pointwise (along a subsequence) to $m_0$ at each continuity point of $m_0$ and both functions are bounded. Moreover, since $|{\check\theta_n}- X|^2$ has density wrt Lebesgue measure, we have that the points of discontinuities of $m_0$ wrt to density of $|{\check\theta_n}- X|^2$ is measure zero. Thus by the Dominated convergence theorem, we have that $ \|m_0(|\check\theta_n- X|^2)- \bar{\eta}_n (|{\check\theta_n}- X|^2)\| \to 0$ along a subsequence as $n\to \infty$. Thus, we have~\eqref{eq:limitpoint}. Thus, we have that all limit points of $\check{\theta}_n$ are identical to $\theta_0$. Thus $|\check{\theta}-\theta_0| =o_p(1)$.  

\paragraph{Proof of rate of convergence of $\check{\theta}$:} 
\label{par:proof_of_rate_of_convergence_of_}
We will use the above two results and assumption~\ref{a4prime} to show that the $\check{\theta}$ inherits the rate of convergence of $x \mapsto \check{\eta}(|\check{\theta}-x|^2).$ The proof borrows from~\cite[Lemma 5.7]{VANC}, ~\cite[Theorem 3.8]{2017arXiv170800145K}, and~\cite[Corollary 5.3]{balabdaoui2016least}. Let $g_1(x) :=[\check{\eta} (|\check{\theta}- x|^2)- {\eta_0} (|{\check\theta}- x|^2)]  $ and $g_2(x):=[ \eta_0 (|\check{\theta}- x|^2)- {\eta_0} (|{\theta_0}- x|^2)] $. For any set $B \subset \rchi$ with nonempty interior, let $X_B$ a random variable such that $\P(X\in A) = \P(X\in A\cap B)/\P(X\in B)$.\footnote{We will choose an appropriate $B$ later.} By the Cauchy-Schwarz inequality, we have 
\begin{align}\label{eq:pdt_to_sum}
\begin{split}
\big(P_{X_B}[g_1 g_2]\big)^2
&= \big(P_{X_B} \big[ [\check{\eta} (|\check{\theta}- \cdot|^2)- {\eta_0} (|{\check\theta}- \cdot|^2)] g_2(\cdot)\big]\big)^2\\
&= \Big(P_{X_B}\big[ g_1(\cdot) P_{X_B} [g_2(X_B)\big||\check{\theta}- X_B|]\big]\Big)^2\\
&\le P_{X_B}\big[ g_1^2\big] P_{X_B}\Big[P_{X_B}^2 \big[g_2(X_B)\big||\check{\theta}- X_B|\big]\Big]\\
&= c_n P_{X_B} g_1^2 P_{X_B}g_2^2,
\end{split}
\end{align}
where \[c_n:= \frac{  P_{X_B}\Big[P_{X_B}^2 \big[g_2(X_B)\big||\check{\theta}- X_B|\big]\Big]}{P_{X_B}g_2^2} = \frac{P\Big(\big[ \eta_0 (|\check{\theta}- X_B|^2)- P\big({\eta_0} (|{\theta_0}- X_B^2)\big| |\check{\theta}- X_B|\big) \big]^2\Big)}{P\Big(\big[ \eta_0 (|\check{\theta}- X_B|^2)-{\eta_0} (|{\theta_0}- X_B^2)\big]^2\Big)} .
\]
If $c_n<1$, then by Lemma~5.7~\cite{VANC}, we can infer that 

\[ P_{X_B}g^2_1+ P_{X_B}g^2_2 \le \frac{1}{1-\sqrt{c_n} } P_{X_B}(g_1+g_2)^2 =\frac{1}{(1-\sqrt{c_n}) \P(X\in B) }  \int_B \Big\{\check{\eta}(|\check\theta-x|^2) -\eta_0(|\theta_0-x|^2)\Big\}^2 dP_X(x).
\]
Now fix $\varepsilon >0$. 
 By consistency of $\check\theta$, we can easily find $B\subset \rchi$ such that the interior of $B$ is not empty,  $\P(X\in B) >0$, $\{|{\theta}_0-x|^2: x\in B\} \subset \mathcal{A}$, and $\P( \{|\check{\theta}-x|^2: x\in B\} \subset \mathcal{A}) \ge 1- \varepsilon$ for all $n>n_0$. If we can show that $c_n <1$ with probability tending to 1 for the above choice of $B$. Then previous two parts of the proof, we can find large constants $M_1, M_2, M_3$, and $n_0$ such that for any $n>n_0$, the following three inequalities hold:
\begin{align}\label{eq:1}
\begin{split}
\P\left(\int \Big\{\check{\eta}(|\check\theta-x|^2) -\eta_0(|\theta_0-x|^2)\Big\}^2 dP_X(x) >M_1 n^{-2/3}n^{2/q}\right) &\le \varepsilon, \\
\P\big(|\check{\theta}-\theta_0|\ge 1/M_2\big) &\le \varepsilon, \\
 \text{and }\qquad \P(c_n \ge 1- 1/M_3) &\le \varepsilon.
\end{split}
\end{align}
Combining the above two displays for the above the choice of $B$ above, for all $n> n_0$, we get
\begin{equation}\label{eq:g_bound}
\P\left(P_{X_B} g_2^2 \lesssim n^{-2/3} n^{2/q}\right ) \ge 1 -4 \varepsilon.
\end{equation}
We will now show that~\eqref{eq:g_bound}, implies that $\P(|\check\theta-\theta_0| \lesssim n^{-2/3} n^{2/q}  ) \le 4 \varepsilon.$
\begin{align}\label{eq:g_theta}
\begin{split}
1-4 \varepsilon &\le   \P\left(P_{X_B} g_2^2 \lesssim n^{-2/3} n^{2/q}\right )\\
 &= \P\left(\int_B \Big\{{\eta}_0(|\check\theta-x|^2) -\eta_0(|\theta_0-x|^2)\Big\}^2 dP_X(x)  \lesssim  \P(X\in B) n^{-2/3} n^{2/q}\right ) \\ 
&\le \P\left( k_0 \int_B \Big\{|\check\theta-x|^2-|\theta_0-x|^2\Big\}^2 dP_X(x)  \lesssim  \P(X\in B) n^{-2/3} n^{2/q}\right ) \\ 
&\le \P\left( k_0 |\check\theta-\theta_0|^2 \inf_{\beta \in S_{d-1}} \int_B \Big\{ \beta^\top (\check\theta+\theta_0 - 2x) \Big\}^2 dP_X(x)  \lesssim  \P(X\in B) n^{-2/3} n^{2/q}\right ),
\end{split}
\end{align}
where $k_0 := \min_{t \in \mathcal{A}} [\eta_0'(t)]^2$. Recall that by continuity of $t \mapsto \eta_0'(t)$ (by~\ref{a4prime}), we have that  $k_0 > 0$. Since $B$ does not depend on particular value of $\check\theta(\omega)$  and $dP_X$ is positive everywhere on $B$, we obtain
\begin{align}\label{eq:g_theta1}
\begin{split}
1-4 \varepsilon &\le \P\left(k_0 |\check\theta-\theta_0|^2 \inf_{\beta \in S_{d-1}} \int_B \Big\{ \beta^\top (\check\theta+\theta_0 - 2x) \Big\}^2 dP_X(x)  \lesssim  \P(X\in B) n^{-2/3} n^{2/q}\right )\\
&\le \P\left( |\check\theta-\theta_0|^2 \lesssim n^{-2/3} n^{2/q}\right ).
\end{split}
\end{align}
Since we can do this for every $\varepsilon>0, $ we have that $\|\check{\theta}-\theta_0\| =O_p\left(n^{-1/3}n^{1/q}\right).$

We will now complete the proof by showing that $c_n <1$ with probability tending to 1. By continuous differentiability of $\eta_0$ on $\{|{\theta}_0-x|^2: x\in B\}$, for all $x\in B, $ we have 
\begin{equation}\label{eq:eta_deruv}
 \eta_0 (|{\theta}_0- x|^2) = \eta_0 (|\check{\theta}- x|^2) +\eta_0'(|\check{\theta}- x|^2)  \big(|\theta_0- x|^2- |\check{\theta}- x|^2\big) + o\big(\big||\theta_0- x|^2- |\check{\theta}- x|^2|\big|\big).
\end{equation}
Since $|\theta_0- x|^2- |\check{\theta}- x|^2 := (\theta_0-\check{\theta})^\top (\theta_0+ \check{\theta}- 2x)$, and both $\Theta$ and $\rchi$ are bounded, we conclude that $||\theta_0- x|^2- |\check{\theta}- x|^2|| \le 3T|\theta_0-\check{\theta}| $ 
Letting $r_n := |\theta_0-\check{\theta}|$ and by Taylor series expansion (in~\eqref{eq:eta_deruv}), we conclude
\begin{align}\label{eq:c_n1}
\begin{split}
&P\Big(\big[ \eta_0 (|\check{\theta}- X_B|^2)- P\big({\eta_0} (|{\theta_0}- X_B^2)\big| |\check{\theta}- X_B|\big) \big]^2\Big)\\
={}&P\Big(\Big[  P\Big(\eta_0'(|\check{\theta}- X_B|^2)  \big(|\theta_0- X_B|^2- |\check{\theta}- X_B|^2\big) + o(r_n) \Big| |\check{\theta}- X_B|\Big) \Big]^2\Big)\\
={}&P\Big(\Big[  P\Big(\eta_0'(|\check{\theta}- X_B|^2)  \big(|\theta_0- X_B|^2- |\check{\theta}- X_B|^2\big)  \Big||\check{\theta}- X_B|\Big) \Big]^2\Big) + o(r_n^2) \\
&\qquad+ o(r_n)  P\Big(\eta_0'(|\check{\theta}- X_B|^2)  \big(|\theta_0- X_B|^2- |\check{\theta}- X_B|^2\big)\Big).
\end{split}
\end{align}
Another Taylor expansion implies that
\begin{align}\label{eq:c_n2}
\begin{split}
P\Big(\big[ \eta_0 (|\check{\theta}- X_B|^2)-{\eta_0} (|{\theta_0}- X_B^2)\big]^2\Big)={}&P\Big[  \Big(\eta_0'(|\check{\theta}- X_B|^2)  \big(|\theta_0- X_B|^2- |\check{\theta}- X_B|^2\big) \Big)^2 \Big] + o(r_n^2) \\
&\qquad+ o(r_n)  P\Big(\eta_0'(|\check{\theta}- X_B|^2)  \big(|\theta_0- X_B|^2- |\check{\theta}- X_B|^2\big)\Big).
\end{split}
\end{align}
Recall that by~\ref{a4prime}, $t\mapsto\eta_0'(t)$ is continuous on $\mathcal{A}$ and $0  < \inf_{\mathcal{A}} |\eta_0'(t) | \le \sup_{t\in \{|\check{\theta}-x|^2: x\in B\}}|\eta_0'(t) | <\infty.$ Hence we deduce that $r_n \lesssim P\Big(\eta_0'(|\check{\theta}- X_B|^2)  \big(|\theta_0- X_B|^2- |\check{\theta}- X_B|^2\big)\Big)\lesssim r_n$. Thus, we have that 
\begin{align}\label{eq:c_nproof1}
\begin{split}
c_n &=\frac{P\Big(\big[\eta_0'(|\check{\theta}- X_B|^2)\big]^2 \Big[  P\Big(  |\theta_0- X_B|^2- |\check{\theta}- X_B|^2  \Big||\check{\theta}- X_B|\Big) \Big]^2\Big) + o(r_n^2)}{P\Big(\big[\eta_0'(|\check{\theta}- X_B|^2)\big]^2  \Big[  |\theta_0- X_B|^2- |\check{\theta}- X_B|^2 \Big]^2 \Big) + o(r_n^2)}.
\end{split}
\end{align}
We will now simplify parts of the numerator and denominator of $c_n$. Observe that
\begin{align}\label{eq:simpli_}
\begin{split}
\Big[  P\Big(  |\theta_0- X_B|^2- |\check{\theta}- X_B|^2  \Big||\check{\theta}- X_B|\Big) \Big]^2={}& 4 (\check{\theta}-\theta_0)^\top P(X_B \big| |\check{\theta}- X_B|)P(X_B^\top  \big| |\check{\theta}- X_B|) (\check{\theta}-\theta_0)\\
 &+ (|\check{\theta}|^2-|\theta_0|^2)^2  +4 (|\check{\theta}|^2-|\theta_0|^2) (\check{\theta}-\theta_0)^\top P(X_B \big| |\check{\theta}- X_B|),
\end{split}
\end{align}
and 
\begin{align}\label{eq:simpli_1}
\begin{split}
\Big[  |\theta_0- X_B|^2- |\check{\theta}- X_B|^2 \Big]^2
 =& 4 (\check{\theta}-\theta_0)^\top X_B X_B^\top   (\check{\theta}-\theta_0)+ (|\check{\theta}|^2-|\theta_0|^2)^2\\
 &\qquad  +4 (|\check{\theta}|^2-|\theta_0|^2) (\check{\theta}-\theta_0)^\top X_B.
\end{split}
\end{align}
Substituting these in $c_n$, we derive
\begin{align}\label{eq:c_nproof11}
\begin{split}
c_n &= \frac{ 4(\check{\theta}-\theta_0)^\top P\Big(\big[\eta_0'(|\check{\theta}- X_B|^2)\big]^2 P\big(X_B\big||\check{\theta}- X_B|\big)P\big(X_B^\top \big||\check{\theta}- X_B|\big)\Big)(\check{\theta}-\theta_0) +b_n+ o(r_n^2)}{ 4 (\check{\theta}-\theta_0)^\top P\Big(\big[\eta_0'(|\check{\theta}- X_B|^2)\big]^2 X_B X_B^\top  \Big) (\check{\theta}-\theta_0)+b_n + o(r_n^2)},
\end{split}
\end{align}
where 
\begin{align}\label{eq:b_n_def}
\begin{split}
b_n :=& (|\check{\theta}|^2-|\theta_0|^2)^2 P\Big(\big[\eta_0'(|\check{\theta}- X_B|^2)\big]^2\Big) + 4(|\check{\theta}|^2-|\theta_0|^2) P\Big(\big[\eta_0'(|\check{\theta}- X_B|^2)\big]^2 (\check{\theta}-\theta_0)^\top X \Big).
\end{split}
\end{align}
Thus 
\begin{equation}\label{eq:c_n_final}
c_n =   d_n + \frac{(1-d_n) (b_n+ o(r_n^2))}{P\Big(\big[\eta_0'(|\check{\theta}- X_B|^2)\big]^2  \Big[  |\theta_0- X_B|^2- |\check{\theta}- X_B|^2 \Big]^2 \Big) + o(r_n^2)},
\end{equation}
where 
\begin{equation}\label{eq:d_n_def}
d_n := \frac{ (\check{\theta}-\theta_0)^\top P\Big(\big[\eta_0'(|\check{\theta}- X_B|^2)\big]^2 P\big(X_B\big||\check{\theta}- X_B|\big)P\big(X_B^\top \big||\check{\theta}- X_B|\big)\Big)(\check{\theta}-\theta_0) }{ (\check{\theta}-\theta_0)^\top P\Big(\big[\eta_0'(|\check{\theta}- X_B|^2)\big]^2 X_B X_B^\top  \Big) (\check{\theta}-\theta_0)},
\end{equation}
Note that $|(b_n+ o(r_n^2))/{P([\eta_0'(|\check{\theta}- X_B|^2)]^2  [  |\theta_0- X_B|^2- |\check{\theta}- X_B|^2 ]^2 ) + o(r_n^2)}|\le 1$, thus if we can show that $d_n(\omega)<1$, then we have that $c_n(\omega) <1$. 
Define $\gamma := (\check{\theta} - \theta_0)/{|\check{\theta} - \theta_0|}$ and $S_{d-1} := \{x\in \R^d: |x|=1\}$. Since $t\mapsto \eta_0'(t)$ is continuous on $\mathcal{A}$, we have that 
\begin{align}\label{eq:d_n_gamma}
\begin{split}
1- d_n &= \frac{ (\check{\theta}-\theta_0)^\top P\Big(\big[\eta_0'(|\check{\theta}- X_B|^2)\big]^2 \big[X_B -P\big(X_B\big||\check{\theta}- X_B|\big)\big]\big[X_B^\top -P\big(X_B^\top \big||\check{\theta}- X_B|\big)\big]\Big)(\check{\theta}-\theta_0) }{ (\check{\theta}-\theta_0)^\top P\Big(\big[\eta_0'(|\check{\theta}- X_B|^2)\big]^2 X_B X_B^\top  \Big) (\check{\theta}-\theta_0)}\\
\ge & \inf_{\gamma \in S_{d-1}} \frac{ \gamma^\top P\Big(\big[\eta_0'(|\check{\theta}- X_B|^2)\big]^2  \big[X_B -P\big(X_B\big||\check{\theta}- X_B|\big)\big]\big[X_B^\top -P\big(X_B^\top \big||\check{\theta}- X_B|\big)\big]\Big)\gamma }{ \gamma^\top P\Big(\big[\eta_0'(|\check{\theta}- X_B|^2)\big]^2 X_B X_B^\top  \Big) \gamma}\\
\ge & k_0\inf_{\gamma \in S_{d-1}} \frac{ \gamma^\top P\Big(  \big[X_B -P\big(X_B\big||\check{\theta}- X_B|\big)\big]\big[X_B^\top -P\big(X_B^\top \big||\check{\theta}- X_B|\big)\big]\Big)\gamma }{ \gamma^\top P\Big(\big[\eta_0'(|\check{\theta}- X_B|^2)\big]^2 X_B X_B^\top  \Big) \gamma} \\
\ge & k_0 \frac{\inf_{\gamma \in S_{d-1}}  \gamma^\top P\Big( \big[X_B -P\big(X_B\big||\check{\theta}- X_B|\big)\big]\big[X_B^\top -P\big(X_B^\top \big||\check{\theta}- X_B|\big)\big]\Big)\gamma }{ \sup_{\gamma \in S_{d-1}} \gamma^\top P\Big(\big[\eta_0'(|\check{\theta}- X_B|^2)\big]^2 X_B X_B^\top  \Big) \gamma}\\
= & k_0 \frac{\inf_{\gamma \in S_{d-1}}   P\Big( \big[\gamma^\top X_B -P\big(\gamma^\top X_B\big||\check{\theta}- X_B|\big)\big]^2\Big) }{ \sup_{\gamma \in S_{d-1}}  P\Big(\big[\eta_0'(|\check{\theta}- X_B|^2)\big]^2 (\gamma^\top  X_B )^2 \Big)} 
\end{split}
\end{align}
where $k_0 := \min_{t \in \mathcal{A}} [\eta_0'(t)]^2$. We will now show that the numerator is strictly positive. Let $\{e_1, \ldots, e_{d-1}, e_d\}$ be  orthonormal basis of $\R^d$. Then 
\[\inf_{\gamma \in S_{d-1}}   P\Big( \big[\gamma^\top X_B -P\big(\gamma^\top X_B\big||\check{\theta}- X_B|\big)\big]^2\Big) = \min_{1 \le i \le d} P\Big( \big[e_i^\top X_B -P\big(e_i^\top X_B\big||\check{\theta}- X_B|\big)\big]^2\Big).\] Since $B$ has a nonempty interior, $\P(X\in B)>0 $, and $X$ has a Lebesgue density, we have that $P\Big( \big[e_i^\top X_B -P\big(e_i^\top X_B\big||\check{\theta}- X_B|\big)\big]^2\Big) >0$ for all $i.$ Thus we have that $d_n <1$ with probability tending to 1. And the proof is now complete.



\section{Proof of Theorem~\ref{thm:SSE_consis}} 
\label{sec:thm:SSE_consis}
Recall that 
 \begin{equation}
 \mathbb{M}_n(\theta):=\sum_{i=1}^{n} \left(Y_i - \tilde\eta_\theta\big(|\theta-X_i|^2\big)\right)({X_i-\theta}).
\end{equation}
and $M: \Theta \mapsto \R^d$ is defined as 
\begin{equation}
M(\theta):=\int_{\D}\left[Y - \eta_\theta\big(|\theta-X|^2\big)\right]({X-\theta}) dP_X,
\end{equation}
where \begin{equation}
{\eta}_\theta(u^2):=\argmin_{\eta\in \M} \Q(\eta, \theta) = \argmax_{\eta\in \M} \E \bigg(Y\eta(|\theta-X|^2) -\frac{1}{2}\eta^2(|\theta-X|^2)\bigg) = \E\big(\eta_0(|\theta_0-X|^2) \big| |\theta-X| = u\big).
\end{equation}
Since $\theta\mapsto \mathbb{M}_n(\theta)$ is a piecewise constant function with finitely many jumps, it is not clear if $\hat\theta$ (a zero crossing) exists. The proof of the theorem is split into two parts: (1) existence  of a zero crossing (Section~\ref{sub:Exist}) and (2) consistency of the zero crossing (Section~\ref{sub:consistency_og_}). 

\subsection{Existence of a zero crossing}\label{sub:Exist}
\begin{theorem}\label{thm:ExistenceZero}
Suppose the conditions of Theorem~\ref{thm:SSE_consis} hold. Then for all $\epsilon>0$, there exists a $N(\epsilon)$ such that 
\[\P\big( \MM_n(\cdot) \text{ has a zero crossing in } B(\theta_0, r)\big) \ge 1-\epsilon, \qquad \forall n> N(\epsilon).\] 
\end{theorem}


\begin{proof}
In Lemma~\ref{lem:deriv_M}, we show that $M(\theta)$ is a differentiable function. Hence, multivariate Taylor's theorem implies that
\begin{equation}\label{eq:M_taylor}
M(\theta) = M(\theta_0) + M'(\theta_0) (\theta-\theta_0) + o(|\theta -\theta_0|)
\end{equation}
Thus by Lemma~\ref{lem:Approx_M}, we have 
\begin{equation}\label{eq:Taylor+Approx}
\MM_n(\theta) = M'(\theta_0) (\theta-\theta_0) + o(|\theta -\theta_0|) + B_n(\theta),
\end{equation}
where $B_n(\theta)$ is such that $\sup_{\theta \in B(\theta_0, r)} |B_n(\theta)| =o_p(1).$ Let us define
\begin{equation}\label{eq:Rn_def}
R_n(\theta):= \MM_n(\theta) - M'(\theta_0) (\theta-\theta_0),
\end{equation}
 and for any $h >0$, $K_h(\cdot)$ be any (fixed) $d$-dimensional product differentiable kernel with finite support. Let us now define
\[ \tilde{R}_{nh}(\theta) :=\int_{\R^d} R_n(v) K_h(v-\theta)dv, \]
and 
\begin{equation}\label{eq:M_tilde}
\widetilde{\MM}_{n, h}(\theta) := \int \MM_n(v)  K_h(v-\theta)dv = M'(\theta_0) (\theta-\theta_0) + \tilde{R}_{nh}(\theta).
\end{equation}
We will first show that for every large  $n$, $\widetilde{\MM}_{n, h}(\theta)$ has a zero in $B(\theta_0, r)$ with probability tending to one. We will then show that, this in turn implies that for every large $n$, $\MM_{n}(\theta)$ has a zero crossing in $B(\theta_0, r)$ with probability tending to one. Consider the following reparameterization of $\theta, $ $\gamma := M'(\theta_0)  \theta$ and $\gamma_0 := M'(\theta_0)  \theta_0$. Now consider the function 
\begin{equation}\label{eq:new_gamm_maping}
k_{n, h}(\gamma):= \gamma_0 - \tilde{R}_{n, h}\big([M'(\theta_0)]^{-1} \gamma\big). 
\end{equation}

Note that $M'(\theta_0)$ is invertible and $\tilde{R}_{n, h}$ is a continuous map. 
 Thus by~\eqref{eq:Taylor+Approx} and \eqref{eq:Rn_def}, we have for each small enough $\delta>0$,  $h$ small enough such that 

 \[ \P\big(k_{n, h}(B(\gamma_0, \delta)) \subset B(\gamma_0, \delta)\big) \ge 1-\epsilon, \] for all large enough $n$.  If $k_{n, h}(B(\gamma_0, \delta)) \subset B(\gamma_0, \delta)$, then by Brouwer's fixed point theorem, we have that there exists a $\gamma_{n, h}$ such that  $k_{n, h}(\gamma_{n, h})= \gamma_{n, h}$, i.e., $  \gamma_{n, h} =\gamma_0- \tilde{R}_{n, h}([M'(\theta_0)]^{-1} \gamma_{n, h})$. Defining $\theta_{n, h} := [M'(\theta_0)]^{-1} \gamma_{n, h}$, we get that 
\begin{equation}\label{eq:theta_tilde}
\widetilde{\MM}_{n, h}(\theta_{n, h})= M'(\theta_0) (\theta_{n, h}-\theta_0) + \tilde{R}_{nh}(\theta_{n, h}) =  \gamma_0- \gamma_{n, h} - \tilde{R}_{n, h}\big([M'(\theta_0)]^{-1} \gamma_{n, h}\big)= 0.
\end{equation}

For each fixed $n$, consider the sequence of $\theta_{n, h_i}$ as $h_i \downarrow 0.$ By compactness of $B(\gamma_0, \delta)$, we have that $\theta_{n, h_i}$'s  have a limit point. Let us denote this point by $\theta_n.$ 

In the following, we will show that $\MM_{n, j}(\theta)$ (the $j$th component of $\MM_{n}(\theta)$) has a zero crossing at $\theta_n$ for all $j\le d$. Suppose $\MM_{n, j}(\theta)$  does not have a zero crossing at $\theta_{n}$. Then there exists a $\delta>0, $ such that  $\MM_{n, j}(\theta)$ must have the same sign for all $\theta \in B({\theta}_n, \delta)$. Let $\MM_{n, j}(\theta) >0$ for all $\theta \in B({\theta}_n, \delta)$. Since $\MM_{n, j}(\theta)$ takes only finitely many values, there exists a $c>0$ such that $\MM_{n, j}(\theta) \ge c> 0$ for all $\theta \in B({\theta}_n, \delta)$. Observe that 
 $ \widetilde{\MM}_{n, h, j}(\theta) = \int \MM_{n, j}(u) K_h(u-\theta) du > c/2, $
 for all $\theta \in B({\theta}_n, \delta)$. A contradiction to~\eqref{eq:theta_tilde}, since ${\theta}_{n, h} \in B({\theta}_n, \delta)$ for large $h.$
\end{proof}

\subsection{Consistency of the SSCE} 
\label{sub:consistency_og_}

\begin{theorem}\label{thm:Consis}
Suppose the conditions of Theorem~\ref{thm:SSE_consis} hold, then $|\hat\theta -\theta_0| =o_p(1).$
\end{theorem}

\begin{proof}
 By Theorem~\ref{thm:ExistenceZero}, Lemma~\ref{lem:Approx_M}, and the Borel-Cantelli Lemma, there exists a sequence $\{n_k\}$ such that 
 \begin{equation}\label{eq:th12}
   \sum_{k=1}^\infty \P(\MM_{n_k}\text{ does not have a zero crossing}) < \infty \text{  and  } \sup_{\theta \in B(\theta_0, r)} \big|\MM_{n_k}(\theta) - M(\theta)\big| \stackrel{a.s.}{\rightarrow} 0.
 \end{equation}
  If it exists, let $\hat\theta_n$ be any zero crossing of $\MM_n$. Thus for almost every $\omega$, $\hat{\theta}_{n_k}$ exists for all but finitely many $k's$. Hence we can find a subsequence $\{n_{k_j}\}$ (can depend on $\omega$) such that $\hat\theta_{n_{k_j}}(\omega) \to \theta_*$ for some $\theta_* \in B(\theta_0, r).$

The fact that $M(\cdot)$ is a continuous function and ~\eqref{eq:th12} imply that  $|\MM_n(\theta_{n_k} (\omega)) -M(\theta_*)| \rightarrow 0$ for almost all $\omega$.
By the fact that limit of zero crossing become roots of the limit, we will have that, we have that $M(\theta_*) =0.$ We will next show that this implies that $\theta_* =\theta_0$. Recall that (see~\eqref{eq:m_cov})
\begin{align}\label{eq:test132}
\begin{split}
0&= (\theta_0 -\theta_*) M(\theta_*)\\
&=\E\left[\text{Cov}\Big((\theta- \theta_0)^\top ({X-\theta}), \eta_0\big(|\theta_0-X|^2\big)\big| |\theta-X|^2\Big)\right] 
\end{split}
\end{align}
In the proof of Lemma~\ref{lem:zeroCrossing}, we have shown that $\text{Cov}\Big((\theta- \theta_0)^\top ({X-\theta}), \eta_0\big(|\theta_0-X|^2\big)\big| |\theta-X|^2\Big) \ge 0$. Further by assumption~\ref{assum:NonzeroEverywhere}, we have that $\text{Cov}\Big((\theta- \theta_0)^\top ({X-\theta}), \eta_0\big(|\theta_0-X|^2\big)\big| |\theta-X|^2\Big)$ is not equal to $0$ almost surely for all $\theta\neq \theta_0$. Thus \eqref{eq:test132} implies $\theta_0 =\theta_*.$ As for almost all $\omega$, $\{\hat\theta_{n_{k_j}}(\omega)\}$ converges to $\theta_0$, by Theorem~2.3.2 of~\cite{Durrett}, we have that $\hat{\theta}_{n}$ converges to $\theta_0$ in probability 
 \end{proof}

 \section{Proof of Theorem~\ref{thm:Asymp_norml}} 
 \label{sub:asymptotic_normality_of_}
  To avoid messy and technical details, in the proof we will assume that $\hat{\theta}_n$ exists for each $n.$ We will first show that 
\begin{align}
\MM_n(\hat\theta)&= \int_{\D} \big(\eta_{\theta_0}(|\theta_0-X|^2)- \eta_{\hat{\theta}}(|\hat\theta-X|^2)\big) \big(X-\theta_0- h_{\theta_0}(|\theta_0 -X|)\big) dP_X\label{eq:MnAprrox}\\
&\quad +\int_{\D} \big(Y- \eta_{\theta_0}(|\theta_0-X|^2)\big) \big(X-\theta_0- h_{\theta_0}(|\theta_0 -X|)\big) d(\P_n- P_0) + o_p\big(n^{-1/2}+ \hat\theta-\theta_0\big), \nonumber
\end{align}
where 
\begin{equation}\label{eq:h_def}
 h_{\theta}(u) := \E(X\big||{X-\theta}| =u)-\theta.
 \end{equation}
 By Lemma~\ref{lem:Approx_M} and the fact that $\hat\theta\stackrel{p}{\rightarrow} \theta_0$, we have that 
\begin{align}\label{eq:norm_part1}
\begin{split}
&\int_{\D} \big(X- \E(X\big|{X-\theta_0})\big)\big(\eta_{\theta_0}(|\theta_0-X|^2)- \eta_{\hat{\theta}}(|\hat\theta-X|^2)\big) dP_X\\
={}& M'(\theta_0) (\hat\theta-\theta_0) + o_p(\hat\theta-\theta_0).
\end{split}
\end{align}
Recall that in Lemma~\ref{lem:deriv_M}, we show that $M'(\theta_0)= \E\Big(\eta'_0(|\theta_0-X|^2) \text{Cov}\big(X\big||\theta_0-X|^2 \big)\Big)$.  As $\MM_n(\hat\theta)=0$ and $M'(\theta_0)$ is invertible (see~\ref{a4}), we have that 
\begin{align}\label{eq:Final_step}
\begin{split}
\sqrt{n}(\hat\theta-\theta_0)=  -\{M'(\theta_0)\}^{-1}\sqrt{n}\int_{\D} &\big(Y- \eta_{\theta_0}(|\theta_0-X|^2)\big) \big(X-\theta_0- h_{\theta_0}(|\theta_0 -X|)\big) d(\P_n- P_0)\\
 &+ o_p\big(1+ \sqrt{n}(\hat\theta-\theta_0)\big),
\end{split}
\end{align}
We can then conclude that
\begin{equation}\label{eq:dist1}
\sqrt(\hat\theta-\theta_0) \stackrel{d}{\to} N_d(0, (M'(\theta_0))^{-1} \Sigma (M'(\theta_0))^{-1}),
\end{equation}
where \begin{equation}\label{eq:Sigma}
\Sigma = \E\left( \epsilon^2 [X- \E(X\big||{X-\theta_0}|)] [X- \E(X\big||{X-\theta_0}|)]^\top\right).
\end{equation}

\noindent\textbf{Proof of~\eqref{eq:MnAprrox}:}
Observe that 
\begin{align}\label{eq:proof_MnAprrox}
\begin{split}
\MM_n(\hat\theta) &= \int \left(Y - \tilde\eta_{\hat\theta}\big(|{\hat\theta}-X|^2\big)\right)\big[X-{\hat\theta}\big] d\P_n\\
&= \int \left(Y - \tilde\eta_{\hat\theta}\big(|{\hat\theta}-X|^2\big)\right)\big[X-{\hat\theta}- h_{\hat\theta}\big(|{\hat\theta}-X|\big)\big] d\P_n\\
&\quad+  \int \left(Y - \tilde\eta_{\hat\theta}\big(|{\hat\theta}-X|^2\big)\right)\big[h_{\hat\theta}\big(|{\hat\theta}-X|\big)- \bar{h}_{n, \hat\theta}\big(|{\hat\theta}-X|\big)\big] d\P_n\\
&\quad+ \int \left(Y - \tilde\eta_{\hat\theta}\big(|{\hat\theta}-X|^2\big)\right) \bar{h}_{n, \hat\theta}\big(|{\hat\theta}-X|\big) d\P_n\\
\end{split}
\end{align}
where $\bar{h}_{n, \theta}$ be a piecewise constant function defined as  follows:
\begin{equation}\label{eq:hbar_def}
\bar{h}_{n, \theta} :=\begin{cases}
h_{\theta}(\tau_{i, \theta} ) & \text{if } \eta_\theta(u) > \tilde{\eta}_\theta (\tau_{i, \theta}) \text{ for all } u \in (\tau_{i, \theta}, \tau_{i+1, \theta})\\
h_{\theta}(s ) & \text{if } \eta_\theta(s) = \tilde{\eta}_\theta (\tau_{i, \theta}) \text{ for some } s \in (\tau_{i, \theta}, \tau_{i+1, \theta})\\
h_{\theta}( \tau_{i+1, \theta} ) & \text{if } \eta_\theta(u) < \tilde{\eta}_\theta (\tau_{i, \theta}) \text{ for all } u \in (\tau_{i, \theta}, \tau_{i+1, \theta}),
\end{cases}
\end{equation}
where $\{\tau_{i, \theta}\}_{i=1}^n$ is the values $\{|X_i-\theta|\}_{i=1}^n$ in increasing order. By definition of $\bar{h}_{n, \theta}$ and the fact that $\tilde{\eta}_\theta$ is the minimizer of $\sum_{i=1}^{n} \left(Y_i - \tilde\eta_\theta\big(|\theta-X_i|^2\big)\right)^2$, we have that 
\[\int \left(Y - \tilde\eta_{\hat\theta}\big(|{\hat\theta}-X|^2\big)\right) \bar{h}_{n, \hat\theta}\big(|{\hat\theta}-X|\big) d\P_n= 0\footnote{See \cite[Page 332]{GJ14} for this property of the isotonic estimator; also see~\cite{groeneboom2016current}.}.\]
In Lemma~\ref{lem:II_proof}, we show that 
\begin{equation}\label{eq:II_orig}
\left|\int \left(Y - \tilde\eta_{\hat\theta}\big(|{\hat\theta}-X|^2\big)\right)\big[h_{\hat\theta}\big(|{\hat\theta}-X|\big)- \bar{h}_{n, \hat\theta}\big(|{\hat\theta}-X|\big)\big] d\P_n\right|= o_p(n^{-1/2} +|\hat\theta-\theta_0|).
\end{equation}
Thus by~\eqref{eq:proof_MnAprrox} and \eqref{eq:II_orig}, we have that 
\begin{align}\label{eq:second_split}
\begin{split}
\MM_n(\hat\theta) 
&= \int \left(Y - \tilde\eta_{\hat\theta}\big(|{\hat\theta}-X|^2\big)\right)\big[X-{\hat\theta}- h_{\hat\theta}\big(|{\hat\theta}-X|\big)\big] d\P_n+ o_p(n^{-1/2} +|\hat\theta-\theta_0|)\\
&= \int \left(Y - \eta_{\hat\theta}\big(|{\hat\theta}-X|^2\big)\right)\big[X-{\hat\theta}- h_{\hat\theta}\big(|{\hat\theta}-X|\big)\big] d\P_n\\
&\quad+ \int \left(\eta_{\hat\theta}\big(|{\hat\theta}-X|^2\big) - \tilde\eta_{\hat\theta}\big(|{\hat\theta}-X|^2\big)\right)\big[X-{\hat\theta}- h_{\hat\theta}\big(|{\hat\theta}-X|\big)\big] d\P_n+ o_p(n^{-1/2} +|\hat\theta-\theta_0|)\\
\end{split}
\end{align}
In Lemma~\ref{lem:Second_split}, we show that 
\begin{equation}\label{eq:Second_split_approx}
\int \left(\eta_{\hat\theta}\big(|{\hat\theta}-X|^2\big) - \tilde\eta_{\hat\theta}\big(|{\hat\theta}-X|^2\big)\right)\big[X-{\hat\theta}- h_{\hat\theta}\big(|{\hat\theta}-X|\big)\big] d\P_n =  o_p(n^{-1/2} +|\hat\theta-\theta_0|). 
\end{equation}
 Thus by~\eqref{eq:second_split} and~\eqref{eq:Second_split_approx}, we have that 
 \begin{align}\label{eq:third_split}
\begin{split}
\MM_n(\hat\theta) 
&= \int \left(Y - \eta_{\hat\theta}\big(|{\hat\theta}-X|^2\big)\right)\big[X-{\hat\theta}- h_{\hat\theta}\big(|{\hat\theta}-X|\big)\big] d\P_n+ o_p(n^{-1/2} +|\hat\theta-\theta_0|)\\
&= \int \left(Y - \eta_{\hat\theta}\big(|{\hat\theta}-X|^2\big)\right)\big[X-{\hat\theta}- h_{\hat\theta}\big(|{\hat\theta}-X|\big)\big] d(\P_n-P_0)\\
&\quad +\int \left(Y - \eta_{\hat\theta}\big(|{\hat\theta}-X|^2\big)\right)\big[X-{\hat\theta}- h_{\hat\theta}\big(|{\hat\theta}-X|\big)\big] dP_0+ o_p(n^{-1/2} +|\hat\theta-\theta_0|)\\
\end{split}
\end{align}
Observe that
\begin{align}\label{eq:third_split_part2}
\begin{split}
&\int \left(Y - \eta_{\hat\theta}\big(|{\hat\theta}-X|^2\big)\right)\big[X-{\hat\theta}- h_{\hat\theta}\big(|{\hat\theta}-X|\big)\big] dP_0\\
={}&\int \left(\eta_{\theta_0}\big(|{\theta_0}-X|^2 - \eta_{\hat\theta}\big(|{\hat\theta}-X|^2\big)\right)\big[X-{\hat\theta}- h_{\hat\theta}\big(|{\hat\theta}-X|\big)\big] dP_X\\
={}&\int \left(\eta_{\theta_0}\big(|{\theta_0}-X|^2 - \eta_{\hat\theta}\big(|{\hat\theta}-X|^2\big)\right)\big[X-{\theta_0}- h_{\theta_0}\big(|{\hat\theta}-X|\big)\big] dP_X +o_p(|\hat\theta-\theta_0|).
\end{split}
\end{align}
Here the last step is due to assumption~\ref{ContConditional} and Theorem~\ref{thm:unif_eta}.  Moreover
\begin{align}\label{eq:third_split_part1}
\begin{split}
& \int \left(Y - \eta_{\hat\theta}\big(|{\hat\theta}-X|^2\big)\right)\big[X-{\hat\theta}- h_{\hat\theta}\big(|{\hat\theta}-X|\big)\big] d(\P_n-P_0)\\
={}& \int \left(Y - \eta_{\theta_0}\big(|{\theta_0}-X|^2\big)\right)\big[X-{\hat\theta}- h_{\hat\theta}\big(|{\hat\theta}-X|\big)\big] d(\P_n-P_0)\\
& +\int \left(\eta_{\theta_0}\big(|{\theta_0}-X|^2\big) - \eta_{\hat\theta}\big(|{\hat\theta}-X|^2\big)\right)\big[X-{\hat\theta}- h_{\hat\theta}\big(|{\hat\theta}-X|\big)\big] d(\P_n-P_0)\\
={}& \int \left(Y - \eta_{\theta_0}\big(|{\theta_0}-X|^2\big)\right)\big[{X-\theta_0} -h_{\theta_0}\big(|\theta_0-X|\big)\big] d(\P_n-P_0)\\
&+ \int \left(Y - \eta_{\theta_0}\big(|{\theta_0}-X|^2\big)\right)\big[\theta_0+ h_{\theta_0}\big(|\theta_0-X|\big)-{\hat\theta}- h_{\hat\theta}\big(|{\hat\theta}-X|\big)\big] d(\P_n-P_0)\\
& +\int \left(\eta_{\theta_0}\big(|{\theta_0}-X|^2\big) - \eta_{\hat\theta}\big(|{\hat\theta}-X|^2\big)\right)\big[X-{\hat\theta}- h_{\hat\theta}\big(|{\hat\theta}-X|\big)\big] d(\P_n-P_0)\\
\end{split}
\end{align}
In Lemmas~\ref{lem:proof_third_1} and~\ref{lem:proof_third_2}, we show that 
\begin{equation}\label{eq:proof_third_11}
\left|\int \left(Y - \eta_{\theta_0}\big(|{\theta_0}-X|^2\big)\right)\big[\theta_0+ h_{\theta_0}\big(|\theta_0-X|\big)-{\hat\theta}- h_{\hat\theta}\big(|{\hat\theta}-X|\big)\big] d(\P_n-P_0)\right| = o_p(n^{-1/2})
\end{equation}
and 
\begin{equation}\label{eq:proof_third_21}
\left|\int \left(\eta_{\theta_0}\big(|{\theta_0}-X|^2\big) - \eta_{\hat\theta}\big(|{\hat\theta}-X|^2\big)\right)\big[X-{\hat\theta}- h_{\hat\theta}\big(|{\hat\theta}-X|\big)\big] d(\P_n-P_0)\right| = o_p(n^{-1/2})
\end{equation}
Thus by~\eqref{eq:third_split_part1}, ~\eqref{eq:proof_third_11}, and~\eqref{eq:proof_third_21}, we have that 
\begin{align}\label{eq:189}
\begin{split}
& \int \left(Y - \eta_{\hat\theta}\big(|{\hat\theta}-X|^2\big)\right)\big[X-{\hat\theta}- h_{\hat\theta}\big(|{\hat\theta}-X|\big)\big] d(\P_n-P_0)\\
={}&\int_{\D} \big(Y- \eta_{\theta_0}(|\theta_0-X|^2)\big) \big(X-\theta_0- h_{\theta_0}(|\theta_0 -X|)\big) d(\P_n- P_0) + o_p\big(n^{-1/2}\big).
\end{split}
\end{align}
Combining~\eqref{eq:third_split}, ~\eqref{eq:third_split_part2}, and~\eqref{eq:189}, we have that~\eqref{eq:MnAprrox}. Thus completing the proof.

\begin{lemma}\label{lem:II_proof}
Suppose assumptions~\ref{a1}--\ref{assum:NonzeroEverywhere} hold. If $q\ge 6$, then
\begin{equation}\label{eq:II}
\left|\int \left(Y - \tilde\eta_{\hat\theta}\big(|{\hat\theta}-X|^2\big)\right)\big[h_{\hat\theta}\big(|{\hat\theta}-X|\big)- \bar{h}_{n, \hat\theta}\big(|{\hat\theta}-X|\big)\big] d\P_n\right|= o_p(n^{-1/2} +|\hat\theta-\theta_0|)
\end{equation}
\end{lemma}
\begin{proof}
Let us split the quantity of interest into three parts,
\begin{align}\label{eq:II_proof_1}
\begin{split}
&\int \left(Y - \tilde\eta_{\hat\theta}\big(|{\hat\theta}-X|^2\big)\right)\big[h_{\hat\theta}\big(|{\hat\theta}-X|\big)- \bar{h}_{n, \hat\theta}\big(|{\hat\theta}-X|\big)\big] d\P_n\\
={}& \int\left(Y - \tilde\eta_{\hat\theta}\big(|{\hat\theta}-X|^2\big)\right)\big[h_{\hat\theta}\big(|{\hat\theta}-X|\big)- \bar{h}_{n, \hat\theta}\big(|{\hat\theta}-X|\big)\big] d(\P_n-P_0)\\
&\quad + \int\left(Y - \eta_{\hat\theta}\big(|{\hat\theta}-X|^2\big)\right)\big[h_{\hat\theta}\big(|{\hat\theta}-X|\big)- \bar{h}_{n, \hat\theta}\big(|{\hat\theta}-X|\big)\big] dP_0\\
&\quad + \int\left(\eta_{\hat\theta}\big(|{\hat\theta}-X|^2\big) - \tilde\eta_{\hat\theta}\big(|{\hat\theta}-X|^2\big)\right)\big[h_{\hat\theta}\big(|{\hat\theta}-X|\big)- \bar{h}_{n, \hat\theta}\big(|{\hat\theta}-X|\big)\big] dP_X\\
={}& \bf{I}+ \bf{II}+ \bf{III}
\end{split}
\end{align}
Let us start by providing an upper bound for $\bf{II}$. Observe that Lemma~\ref{lem:H_bar_diff}, the Dominated Convergence Theorem, Lemma~\ref{lem:deriv_M}, and consistency of $\hat\theta$,  imply that
\begin{align}\label{eq:II_1}
\begin{split}
&\bigg|\int\left(Y - \eta_{\hat\theta}\big(|{\hat\theta}-X|^2\big)\right)\big[h_{\hat\theta}\big(|{\hat\theta}-X|\big)- \bar{h}_{n, \hat\theta}\big(|{\hat\theta}-X|\big)\big] dP_0\bigg|\\
={}&\bigg|\int\left(\eta_{\theta_0}\big(|{\theta_0}-X|^2\big) - \eta_{\hat\theta}\big(|{\hat\theta}-X|^2\big)\right)\big[h_{\hat\theta}\big(|{\hat\theta}-X|\big)- \bar{h}_{n, \hat\theta}\big(|{\hat\theta}-X|\big)\big] dP_X\bigg|\\
={}& \bigg|(\hat\theta-\theta_0)(1+o(|\hat\theta-\theta_0|))\int\eta'_{\theta_0}\big(|{\theta_0}-X|^2\big) \big[X-\theta -\big[h_{\theta_0}\big(|{\theta_0}-X|\big)\big]\big[h_{\hat\theta}\big(|{\hat\theta}-X|\big)- \bar{h}_{n, \hat\theta}\big(|{\hat\theta}-X|\big)\big] dP_X\bigg|\\
\le{}&\bigg|(\hat\theta-\theta_0)(1+o(|\hat\theta-\theta_0|))\int\eta'_{\theta_0}\big(|{\theta_0}-X|^2\big) \big[X-\theta -\big[h_{\theta_0}\big(|{\theta_0}-X|\big)\big]\big[h_{\hat\theta}\big(|{\hat\theta}-X|\big)- \bar{h}_{n, \hat\theta}\big(|{\hat\theta}-X|\big)\big] dP_X\bigg|\\
\le{}& d R\|\eta'_{\theta_0}\|_\infty  \bigg|(\hat\theta-\theta_0)(1+o(|\hat\theta-\theta_0|))\bigg|\int\left|\eta_{\hat\theta}\big(|{\hat\theta}-X|^2\big) - \tilde\eta_{\hat\theta}\big(|{\hat\theta}-X|^2\big)\right| dP_X\\
={}&o_p(|\hat\theta-\theta_0|). 
\end{split}
\end{align}
We will now bound $\bf{III}$. By Lemma~\ref{lem:H_bar_diff} and~Theorem~\ref{thm:unif_eta}, we have that 
\begin{align}\label{eq:III_1}
\begin{split}
\bf{III}={}&\left|\int\left(\eta_{\hat\theta}\big(|{\hat\theta}-X|^2\big) - \tilde\eta_{\hat\theta}\big(|{\hat\theta}-X|^2\big)\right)\big[h_{\hat\theta}\big(|{\hat\theta}-X|\big)- \bar{h}_{n, \hat\theta}\big(|{\hat\theta}-X|\big)\big] dP_X\right|\\
\le{}&M^*_h \sqrt{d} \int\left(\eta_{\hat\theta}\big(|{\hat\theta}-X|^2\big) - \tilde\eta_{\hat\theta}\big(|{\hat\theta}-X|^2\big)\right)^2 dP_X\\
={}& O_p(n^{-2/3} n^{2/q}). 
\end{split}
\end{align}
The proof will be complete if we can show that $\mathbf{I} =o_p(n^{-1/2})$. Observe that 
\begin{align}\label{eq:I_split}
\begin{split}
\sqrt{n}\textbf{I}&= \left|\G_n\Big[\epsilon \big[h_{\hat\theta}\big(|{\hat\theta}-X|\big)- \bar{h}_{n, \hat\theta}\big(|{\hat\theta}-X|\big)\big]\Big]\right| \\
&{}+ \left|\G_n\Big[ \big[\eta_{\theta_0}\big(|{\theta_0}-X|^2\big)-\tilde\eta_{\hat\theta}\big(|{\hat\theta}-X|^2\big)\big] \big[h_{\hat\theta}\big(|{\hat\theta}-X|\big)- \bar{h}_{n, \hat\theta}\big(|{\hat\theta}-X|\big)\big]\Big]\right|.
\end{split}
\end{align}
In Lemma~\ref{lem:H_bar_diff}, we show that  that for any $\delta >0$
\begin{equation}\label{eq:G_n_ep_I}
\P\left(\left|\G_n\Big[\epsilon \big[h_{\hat\theta}\big(|{\hat\theta}-X|\big)- \bar{h}_{n, \hat\theta}\big(|{\hat\theta}-X|\big)\big]\Big]\right| \ge \delta \right)= o(1).
\end{equation}
In Lemma~\ref{lem:eta_h_hbar}, we show that the second term on the right side of~\eqref{eq:I_split} is also $o_p(1).$
\end{proof}

\begin{lemma}\label{lem:H_bar_diff}
Suppose assumptions~\ref{a1}--\ref{assum:NonzeroEverywhere} hold. Let $\{h_{\theta, i}\}_{i=1}^d$ denote the $d$ components of $d$-dimensional function $h_\theta$, i.e., $h_\theta(\cdot):= \left(h_{\theta, 1}(\cdot), \ldots, h_{\theta, d}(\cdot)\right)$. There exists constant $M_h$ and $M^*_h$ such that 
\begin{equation}\label{eq:h_boun}
\sup_{\theta \in B(\theta_0, r)}\|h_\theta\|_{2, \infty} \le M_h, \quad\quad \sup_{\theta \in B(\theta_0, r)} \sum_{i=1}^d TV(h_{\theta, i}) \le M_h,
\end{equation}
and 
\begin{equation}\label{eq:h_diff_o}
|h_{\theta}(u) -\bar{h}_{n, \theta}(u)| \le M^*_h \sqrt{d}  |\eta_{\theta}(u)- \tilde\eta_{\theta}(u)| \qquad \forall \theta\in B(\theta_0, r) \text{ and } u\in \D. 
\end{equation}
Moreover, let
\begin{equation}\label{eq:H_def} 
\bar{\h} := \big\{ h(|\theta-x|)\big| h: \D \to \R^d, \, h_i = f_{i, 1} -f_{i, 2}, f_{i, 1}, f_{i, 2} \in \M^{4M_h}, \text{ and } \theta \in B(\theta_0, r)\big\}.
\end{equation}
Then 
\begin{equation}\label{eq:log_ent_H}
\log N_{[]}(\nu, \bar{\h}, \|\cdot\|_{2, P_X}) \le \frac{8dAM_h}{\nu},
\end{equation}

Finally
\begin{equation}\label{eq:g_n_e_h_hbar}
\P\left(\left|\G_n\Big[\epsilon \big[h_{\hat\theta}\big(|{\hat\theta}-X|\big)- \bar{h}_{n, \hat\theta}\big(|{\hat\theta}-X|\big)\big]\Big]\right| \ge \delta \right) =o(1)
\end{equation}

\end{lemma}
\begin{proof}

Recall that by Lemma~\ref{lem:basic_joint}, we have that 
\[\log N_{[\, ]}(\nu, \{f\circ\theta : f :\D\to \R, \theta \in B(\theta_0, r) \text{ and } f \in \M^{4M_h} \}, \|\cdot \|_{P_X})\le  \frac{A4M_h}{\nu}.\]
Thus by stability property of Donsker classes, we have that 
\begin{equation}\label{eq:log_ent_H_proof}
\log N_{[]}(\nu, \bar{\h}, \|\cdot\|_{2, P_X}) \lesssim  \frac{8dAM_h}{\nu}.
\end{equation}
As $\sup_{x\in\rchi}|x| \le R$, it is easy to see that $\sup_{\theta\in B(\theta_0, r)}\|h_\theta\|_{2, \infty} \le R + |\theta| + r.$ The proof of finite total variation follows from the proof of Lemma~\ref{lem:deriv_M} and Lemma~F.4 of~\cite{balabdaoui2019score}.
Recall that $t\mapsto \eta_0(t)$ is strictly decreasing and continuously differentiable, then by Lemma~\ref{lem:eta_theta_def}, we can conclude that $\eta'_\theta(\cdot)$ is bounded away from zero for all $\theta\in B(\theta_0, r)$. We have also assumed that $u \mapsto h_\theta(u)$ has a totally bounded derivative (see~\ref{ContConditional}). Thus techniques used in (10.64) of~\cite{GJ14} imply that there exist constant $M_h^*$ such that
\begin{equation}
|h_{\theta}(u) -\bar{h}_{n, \theta}(u)| \le M^*_h \sqrt{d}  |\eta_{\theta}(u)- \tilde\eta_{\theta}(u)| \qquad \forall \theta\in B(\theta_0, r) \text{ and } u\in \D. 
\end{equation}
Observe that $h_{\theta} -\bar{h}_{n, \theta} \in \bar{\h}$, as both $h_\theta$ and $\bar{h}_{n, \theta}$ satisfy~\eqref{eq:h_boun} and by Lemma~F.5\footnote{A real-valued uniformly bounded function of bounded variation can be written as difference of two bounded and monotone functions.} of~\cite{balabdaoui2019score}, we have that $h_{\theta} -\bar{h}_{n, \theta} = \left(f_{1, 1} -f_{1, 2}, \ldots, f_{d, 1} -f_{d, 2}\right)$, for some functions $ f_{i, 1}, f_{i, 2} \in \M^{4M_h}$ for all $i \in\{1, \ldots, d\}.$
Recall that by Theorem~\ref{thm:unif_eta} and~\eqref{eq:h_diff_o}, we have that  
\begin{equation}\label{eq:L_2_bound}
\sup_{\theta\in B(\theta_0, r)} \|h_{\theta} (|\theta-X|) -\bar{h}_{n, \theta}(|\theta-X|)\|_{2, P_X} = O_p(dn^{-1/3} n^{1/q}).
\end{equation}
Let us define
\begin{equation}\label{eq:zeta_n} 
\bar{\h}_n := \big\{ h(x): h\in \bar\h, \text{ and } \|h(X)\|_{2, P_X} \le \zeta_n \big\}, \qquad\text{where}\qquad\zeta_n := d\log n n^{-1/3} n^{1/q}.
\end{equation}
Then by Chebyshev's inequality
\begin{align}
&\P\left(\left|\G_n\Big[\epsilon \big[h_{\hat\theta}\big(|{\hat\theta}-X|\big)- \bar{h}_{n, \hat\theta}\big(|{\hat\theta}-X|\big)\big]\Big]\right| \ge \delta \right)\nonumber\\
\le{}&\P\bigg(\sup_{\theta\in B(\theta_0, r)}\left|\G_n\Big[\epsilon \big[h_{\theta}\big(|{\theta}-X|\big)- \bar{h}_{n, \theta}\big(|{\theta}-X|\big)\big]\Big]\right|\ge \delta\bigg)\label{eq:P_E_dim}\\
\le{}&\P\bigg(\sup_{h(\cdot)\in \bar{\mathcal{H}}_n}\left|\G_n \epsilon h \right|\ge \delta\bigg) + \P\left( \sup_{\theta\in B(\theta_0, r)} \|h_{\theta} (|\theta-X|) -\bar{h}_{n, \theta}(|\theta-X|)\|_{2, P_X} \ge \zeta_n\right) \nonumber\\
\le{}& 2 \delta^{-1} \sqrt{d} \sum_{i=1}^{d}\E\bigg(\sup_{h(\cdot)\in \bar{\mathcal{H}}_n}\left|\G_n\epsilon h_i\right|\bigg) + o(1). \nonumber
\end{align}
By definition of $\bar\h$, we have that $\sup_{h\in \bar\h} \|h\|_{2, \infty} \le 2 M_h$. Thus by arguments similar to those in the proof of Theorem~\ref{thm:unif_eta} and Lemma F.4 of~\cite{2017arXiv170800145K}, we have that for every $i \in \{1, \ldots, d\}$
\begin{align}\label{eq:e_bar_split}
\begin{split}
\E\bigg(\sup_{h(\cdot)\in \bar{\mathcal{H}}_n}\left|\G_n\epsilon h_i\right|\bigg)\le{}&\E\bigg(\sup_{h(\cdot)\in \bar{\mathcal{H}}_n}\left|\G_n\bar\epsilon h_i\right|\bigg) + 2 \frac{2 M_h C_\epsilon}{\sqrt{n}},
\end{split}
\end{align}
where $\bar\epsilon_i := \epsilon_i\mathbbm{1}_{\{|\epsilon_i|\le C_{\epsilon}\}}$ and  $C_\epsilon \lesssim n^{1/q}$.  Thus by Lemma~F.7 of~\cite{2017arXiv170800145K}, we have that 
\begin{align}\label{eq:H_bar_diff}
\begin{split}
\E\bigg(\sup_{h(\cdot)\in \bar{\mathcal{H}}_n}\left|\G_n \bar\epsilon h_i\right|\bigg)\le{}& \sigma  \sqrt{4d A M_h} \zeta_n^{1/2}\left(1+ \frac{\sigma \sqrt{4d A M_h} \zeta_n^{1/2} C_\epsilon 2M_h}{\zeta_n^2 \sqrt{n} }\right) =o(1).
\end{split}
\end{align}
Thus combining~\eqref{eq:P_E_dim}, ~\eqref{eq:e_bar_split}, and~\eqref{eq:H_bar_diff}, we have~\eqref{eq:g_n_e_h_hbar}.\qedhere

\end{proof}

\begin{lemma}\label{lem:eta_h_hbar}
Suppose  the assumptions of Theorem~\ref{thm:Asymp_norml} hold. If $q\ge 6$, 
\begin{align}\label{eq:eta_h_hbar}
\begin{split}
\left|\G_n\Big[ \big[\eta_{\theta_0}\big(|{\theta_0}-X|^2\big)-\tilde\eta_{\hat\theta}\big(|{\hat\theta}-X|^2\big)\big] \big[h_{\hat\theta}\big(|{\hat\theta}-X|\big)- \bar{h}_{n, \hat\theta}\big(|{\hat\theta}-X|\big)\big]\Big]\right|=& o_p(1).
\end{split}
\end{align}
\end{lemma}
\begin{proof}
Let 
\begin{equation}\label{eq:f_eta}
\mathcal{F}^\eta_n:= \big\{ \eta_0\circ\theta_0 - f\circ\theta:  \theta \in B(\theta_0, r), f\in \M^{B_n}, \text{ and } \|\eta_0\circ\theta_0 - f\circ\theta\|_{P_X} \le \zeta_n\big\},
\end{equation}
where $\zeta_n$ is defined in~\eqref{eq:zeta_n} and 
\begin{equation}\label{eq:B_n_def}
B_n:= n^{1/q} \log  n .
\end{equation}
Then by Lemma~\ref{lem:basic_joint}  we have that 
\begin{equation}\label{eq:ent_1_eta}
\log N_{[\, ]}(\nu, \mathcal{F}_n^\eta, \|\cdot \|_{P_X})\le \frac{A (B_n+M_1)}{\nu},
\end{equation}Let $\mathcal{F}_n^\eta\times \bar{\h} := \{fh : f\in \mathcal{F}_n^\eta \bar{\h}\}.$ By~\eqref{eq:log_ent_H} and Lemma 9.25 of~\cite{Kosorok08}, we have that 
\begin{equation}\label{eq:basic_joint}
\log N_{[\, ]}(\nu, \mathcal{F}_n^\eta\times \bar{\h}, \|\cdot \|_{2, P_X})\le \frac{ A (4d M_h+ M_1+ B_n )}{\nu}.
\end{equation}
Define $A_n:=\sup_{f\in \mathcal{F}_n^\eta\times \bar{\h}}\left\| f \right\|_{2, P_X}$. As $\bar{\h}$ is uniformly bounded by $2M_h$, by definition of $\mathcal{F}_n^\eta$, we have that
\begin{align}\label{eq:bound}
\begin{split}
A_n={}&\sup_{f\in \mathcal{F}_n^\eta, h\in \bar\h} \left\| fh \right\|_{2, P_X}\le{}2\sqrt{d}M_h\sup_{f\in \mathcal{F}_n^\eta} \left\| f \right\|_{ P_X}\le{}2  \sqrt{d} M_h \zeta_n.
\end{split}
\end{align}
Recall that by~\eqref{eq:f_eta}, Theorem~\ref{thm:unif_eta}, and definition of $\bar\h$, we have that
\begin{equation}\label{eq:in_prob_f}
\P\left(\eta_{\theta_0}\big(|{\theta_0}-X|^2\big)-\tilde\eta_{\hat\theta}\big(|{\hat\theta}-X|^2\big) \in \mathcal{F}_n^\eta\right) = o(1)  \text{ and } \P\left(\big[h_{\hat\theta}\big(|{\hat\theta}-X|\big)- \bar{h}_{n, \hat\theta}\big(|{\hat\theta}-X|\big)\big] \in \bar\h\right) =1.\end{equation}
Thus
\begin{align}
 &\P\left(\left|\G_n\Big[ \big[\eta_{\theta_0}\big(|{\theta_0}-X|^2\big)-\eta_{\hat\theta}\big(|{\hat\theta}-X|^2\big)\big] \big[h_{\hat\theta}\big(|{\hat\theta}-X|\big)- \bar{h}_{n, \hat\theta}\big(|{\hat\theta}-X|\big)\big]\Big]\right| >\delta\right)\nonumber\\
\le{}& \P\left(\sup_{f \in \mathcal{F}_n^\eta\times \bar{\h}} \left|\G_n f\right| >\delta\right) +\P \left(\big[\eta_{\theta_0}\big(|{\theta_0}-X|^2\big)-\eta_{\hat\theta}\big(|{\hat\theta}-X|^2\big)\big] \big[h_{\hat\theta}\big(|{\hat\theta}-X|\big)- \bar{h}_{n, \hat\theta}\big(|{\hat\theta}-X|\big)\big] \notin \F_n^\eta \times \bar\h\right)\nonumber\\
={}&\P\left(\sup_{f\in \mathcal{F}_n^\eta\times \bar{\h}} |\G_n f| >\delta \text{ and } \right) + o(1).\label{eq:12_1}
\end{align}
 The proof is now complete, as by Lemma 3.4.2 of~\cite{VdVW96} (for uniformly bounded function classes) and~\eqref{eq:12_1} imply that
\begin{align}\label{eq:13_final}
\begin{split}
 &\P\left(\left|\G_n\Big[ \big[\eta_{\theta_0}\big(|{\theta_0}-X|^2\big)-\eta_{\hat\theta}\big(|{\hat\theta}-X|^2\big)\big] \big[h_{\hat\theta}\big(|{\hat\theta}-X|\big)- \bar{h}_{n, \hat\theta}\big(|{\hat\theta}-X|\big)\big]\Big]\right| >\delta\right)\\
\le{}&\delta^{-1} \sqrt{d} \sum_{i=1}^n \E\left(\sup_{f \in \mathcal{F}_n^\eta\times \bar{\h}} \left|\G_n\Big[ f_i \Big]\right|\right) +o(1)\\
={}& \left( (4d M_h+ M_1+ B_n ) A_n\right)^{1/2}  \left( 1+ \frac{\left( (4d M_h+ M_1+ B_n ) A_n\right)^{1/2}}{\sqrt{n} A_n^2} 2 M_h B_n  \right) +o(1)\\
\lesssim{}& (B_n A_n)^{1/2} + \frac{B_n^2}{\sqrt{n} A_n} + o(1)\\
={}& O\left( n^{-1/6} n^{1/q} \log n + (\log n)^2 n^{1/q} n^{-1/6}\right) +o(1) =o(1). \qedhere
 \end{split}
\end{align}
\end{proof}

\begin{lemma}\label{lem:Second_split}
Suppose  the assumptions of Theorem~\ref{thm:Asymp_norml} hold. If $q\ge 6$, we have that 
\begin{equation}\label{eq:Second_split_approx_1}
\int \left(\eta_{\hat\theta}\big(|{\hat\theta}-X|^2\big) - \tilde\eta_{\hat\theta}\big(|{\hat\theta}-X|^2\big)\right)\big[X-{\hat\theta}- h_{\hat\theta}\big(|{\hat\theta}-X|\big)\big] d\P_n =  o_p(n^{-1/2} +|\hat\theta-\theta_0|). 
\end{equation}
\end{lemma}
\begin{proof}
First note that 
\begin{align}\label{eq:SS_1}
\begin{split}
&\int \left(\eta_{\hat\theta}\big(|{\hat\theta}-X|^2\big) - \tilde\eta_{\hat\theta}\big(|{\hat\theta}-X|^2\big)\right)\big[X-{\hat\theta}- h_{\hat\theta}\big(|{\hat\theta}-X|\big)\big] dP_0 \\
={}&\E\left[\big(\eta_{\hat\theta}\big(|{\hat\theta}-X|^2\big) - \tilde\eta_{\hat\theta}\big(|{\hat\theta}-X|^2\big)\big)\big[X-{\hat\theta}- h_{\hat\theta}\big(|{\hat\theta}-X|\big)\big]\right]\\
={}&\E\bigg[\left(\eta_{\hat\theta}\big(|{\hat\theta}-X|^2\big) - \tilde\eta_{\hat\theta}\big(|{\hat\theta}-X|^2\big)\right)\E\Big[\big[X-{\hat\theta}- h_{\hat\theta}\big(|{\hat\theta}-X|\big)\big]\Big||{\hat\theta}-X|^2 \Big]\bigg]
=0. 
\end{split}
\end{align}
Thus, we have 
\begin{align}\label{eq:SS_2}
\begin{split}
&\int \left(\eta_{\hat\theta}\big(|{\hat\theta}-X|^2\big) - \tilde\eta_{\hat\theta}\big(|{\hat\theta}-X|^2\big)\right)\big[X-{\hat\theta}- h_{\hat\theta}\big(|{\hat\theta}-X|\big)\big] d\P_n \\
={}&\int \left(\eta_{\hat\theta}\big(|{\hat\theta}-X|^2\big) - \tilde\eta_{\hat\theta}\big(|{\hat\theta}-X|^2\big)\right)\big[X-{\hat\theta}- h_{\hat\theta}\big(|{\hat\theta}-X|\big)\big] d(\P_n -P_0)\\
={}&\G_n \Big[\left(\eta_{\hat\theta}\big(|{\hat\theta}-X|^2\big) - \tilde\eta_{\hat\theta}\big(|{\hat\theta}-X|^2\big)\right)\big(X-{\hat\theta}- h_{\hat\theta}\big(|{\hat\theta}-X|\big)\big)\Big]
\end{split}
\end{align}
Let us define
\begin{align}\label{eq:G^*}
\begin{split}
\g^*_n:=  \Big\{\left(\eta_{\theta}\big(|{\theta}-X|^2\big) - f(|\theta-X|)\right)&\big(X-{\theta}- h_{\theta}\big(|{\theta}-X|\big)\big): f\in \M^{B_n}, \\
&\theta \in B(\theta_0, r), \text{ and } \|\eta_{\theta}\big(|{\theta}-X|^2\big) - f(|\theta-X|)\|_{P_X} \le \zeta_n\Big\},
\end{split}
\end{align}
where $B_n$ and $\zeta_n $ are defined as in~\eqref{eq:B_n_def}. By Theorem~\ref{thm:unif_eta}, we have that
\begin{equation}\label{eq:non_inclusion_prob}
\P\left(\big(\eta_{\hat\theta}\big(|{\hat\theta}-X|^2\big) - \tilde\eta_{\hat\theta}\big(|{\hat\theta}-X|^2\big)\big)\big(X-{\hat\theta}- h_{\hat\theta}\big(|{\hat\theta}-X|\big)\big) \notin \g^*_n\right) =o(1).
\end{equation}
Then by~\eqref{eq:joint_ent}, ~\eqref{eq:eta_theta_ent}, ~\eqref{eq:log_ent_H}, and Lemma 9.25 of~\cite{Kosorok08}, we have that 
\begin{equation}\label{eq:log_ent_g}
\log N_{[]}(\nu, \g^*_n, \|\cdot\|_{2, P_X}) \le  \frac{4 AM_h}{\nu} +\frac{A(M_1 +B_n)}{\nu} \lesssim \frac{B_n}{\nu}.
\end{equation}
Combining~\eqref{eq:SS_2}, ~\eqref{eq:non_inclusion_prob}, ~\eqref{eq:log_ent_g}, and Theorem~3.4.2 of~\cite{VdVW96}, we get 
\begin{align}
&\P\left(\G_n \Big[\left(\eta_{\hat\theta}\big(|{\hat\theta}-X|^2\big) - \tilde\eta_{\hat\theta}\big(|{\hat\theta}-X|^2\big)\right)\big(X-{\hat\theta}- h_{\hat\theta}\big(|{\hat\theta}-X|\big)\big)\Big] \ge \delta\right)\nonumber \\
\le{}&\delta^{-1}\sqrt{d} \sum_{i=1}^{d} \E\left(\sup_{g \in \g^*_n} \big|\G_n g_i\big| \right)+ o(1)\nonumber \\
 \lesssim{}&\delta^{-1}\sqrt{d} B_n^{1/2} \zeta_n^{1/2} \left( 1+  \frac{ B_n^{1/2} \zeta_n^{1/2}}{\sqrt{n} \zeta_n^2}  (B_n+ M_1) R\right) + o(1). \label{eq:1_4_final}
\end{align}
The proof is now complete, as it is very easy to see that the first term is $o(1)$ if $q \ge 6$ as
 \[ B_n \zeta_n = n^{-1/3} n^{2/q}(\log n)^2\text{   and   } \frac{B_n^2}{\sqrt{n} \zeta_n} = \log n n^{-1/6} n^{1/q}. \qedhere \]
\end{proof}

\begin{lemma}\label{lem:proof_third_1}
Suppose  the assumptions of Theorem~\ref{thm:Asymp_norml} hold, then 
\begin{equation}\label{eq:proof_third_1}
\left|\G_n \left(Y - \eta_{\theta_0}\big(|{\theta_0}-X|^2\big)\right)\big[\theta_0+ h_{\theta_0}\big(|\theta_0-X|\big)-{\hat\theta}- h_{\hat\theta}\big(|{\hat\theta}-X|\big)\big]\right| = o_p(1).
\end{equation}
\end{lemma}
\begin{proof}
For the proof of this lemma, let us define
\begin{equation}\label{eq:g_def}
g_\theta(\cdot) := \E\left(X\big| |X-\theta|=\cdot\right).
\end{equation}
By definition of~\eqref{eq:h_def}, we have that 
\[\theta_0+ h_{\theta_0}\big(|\theta_0-x|\big)-\big[{\theta}+ h_{\theta}\big(|{\hat\theta}-x|\big)\big] = g_{\theta_0}(|{\theta_0}-x|) -g_{\hat\theta}(|{\hat\theta}-x|).\]
By Lemma~\ref{lem:H_bar_diff} and the fact that $\rchi$ and $\Theta$ is bounded, we have that  \begin{equation}\label{eq:g_boun}
\sup_{\theta \in B(\theta_0, r)}\|g_{\theta}\|_{2, \infty} \le T , \quad\quad \sup_{\theta \in B(\theta_0, r)} \sum_{i=1}^n TV(g_{\theta, i}) \le M_h,
\end{equation}
and 
\begin{equation}\label{eq:h_diff_1}
\|g_{\theta} -g_{\theta_1}\|_{2, \infty} \le M^*_h |\theta-\theta_1| \qquad \forall \theta, \theta_1 \in B(\theta_0, r).
\end{equation}
By arguments similar to the proof of~\eqref{eq:log_ent_H} or proof of~\eqref{eq:eta_theta_ent} in Lemma~\ref{lem:basic_joint} and the fact that a bounded function with finite total variation can be written as a difference of two bounded and monotone functions (also see Lemma~F.5 of~\cite{balabdaoui2019score}),  we can show that 
\begin{equation}\label{eq:h_theta_ent}
\log N_{[]}(\nu, \{g_\theta( |X-\theta|) -g_{\theta_0}(|X-\theta_0|) : \theta\in B (\theta_0, r)\}, \|\cdot\|_{2, P_X}) \lesssim \frac{(M_h+T)}{\eta}.
\end{equation}
Let $\{a_n\}$ be a sequence such that  $a_n\rightarrow 0$, $ a_n \log n\rightarrow \infty$, and $|\hat\theta-\theta_0| =o_p(a_n)$\footnote{Note that Theorem~\ref{thm:Consis} guarantees the existence of such a sequence.}. Thus for every $\zeta>0$, there exist $N_\zeta$ and  $C_\zeta$ such that   $\P(|\hat\theta-\theta_0|  \ge C_\zeta a_n) \le \zeta/2$  for all $n >N_\zeta.$ Now  for any $n > N_\zeta$, let us define
\begin{align}\label{eq:B_delta}
\begin{split}
\mathcal{B}_n^\zeta:= \Big\{g_\theta(|{\theta}-x| ) - g_{\theta_0}|{\theta_0}-x| ) : \theta\in B (\theta_0, C_\zeta a_n)\Big\}
\end{split}
\end{align}
For all $n> N_\zeta$, by Chebyshev's inequality, we have 
\begin{align}
&\P\left(\left|\G_n \epsilon \big(g_{\theta_0}(|{\theta_0}-X|) -g_{\hat\theta}(|{\hat\theta}-X|\big)\right| \ge \delta \right)\nonumber\\
\le{}&\P\left(\left|\G_n \epsilon \big(g_{\theta_0}(|{\theta_0}-X|) -g_{\hat\theta}(|{\hat\theta}-X|\big)\right| \ge \delta \text{ and } \hat\theta\in B(\theta_0, C_\zeta a_n)\right) + \P\left(\hat\theta\notin B(\theta_0, C_\zeta a_n)\right)\nonumber\\
\le{}&\P\bigg(\sup_{g \in \mathcal{B}_n^\zeta }\left|\G_n \epsilon g\right| \ge \delta\bigg) +\zeta/2\label{eq:P_E_dim_g}\\
\le{}& 2 \delta^{-1} \sqrt{d} \sum_{i=1}^{d}\E\bigg(\sup_{g\in B_n}\left|\G_n\epsilon g_i\right|\bigg) + \zeta/2. \nonumber
\end{align}
Observe that $\sup_{g \in \mathcal{B}_n^\zeta} \|g\|_{2, \infty} \le 4T$ as $\sup_{x\in \rchi} |x| \le T$.  Thus by arguments similar to those in the proof of Theorem~\ref{thm:unif_eta} and Lemma F.4 of~\cite{2017arXiv170800145K}, we have that for every $i \in \{1, \ldots, d\}$
\begin{align}\label{eq:e_bar_split_g}
\begin{split}
\E\bigg(\sup_{g\in \mathcal{B}_n^\zeta}\left|\G_n\epsilon g_i\right|\bigg)\le{}&\E\bigg(\sup_{g\in \mathcal{B}_n^\zeta}\left|\G_n\bar \epsilon g_i\right|\bigg) + 2 \frac{ 4R C_\epsilon}{\sqrt{n}},
\end{split}
\end{align}
where $\bar\epsilon_i := \epsilon_i\mathbbm{1}_{\{|\epsilon_i|\le C_{\epsilon}\}}$ and  $C_\epsilon \lesssim n^{1/q}$. We will now use Lemma~F.7 of~\cite{2017arXiv170800145K} to bound the first term on the right of~\eqref{eq:e_bar_split_g}. Observe that for every $i \in \{1, \ldots, d\}$
\begin{equation}\label{eq:inf_delt_g}
\sup_{g\in \mathcal{B}_n^\zeta} \| g_i\|_{2, P_X} \le  M^*_h \sup_{\theta \in B (\theta_0, C_\zeta a_n)}|\theta-\theta_0| = M^*_h C_\zeta a_n.
\end{equation}
Thus by Lemma~F.7 of~\cite{2017arXiv170800145K}, we have that 
\begin{align}\label{eq:final_b}
\begin{split}
\E\bigg(\sup_{g\in \mathcal{B}_n^\zeta}\left|\G_n\bar \epsilon g_i\right|\bigg) &\le \sigma  \left((M_h+T)  M^*_h C_\zeta a_n\right)^{1/2} \left(1+ \frac{\sigma \left((M_h+T)  M^*_h C_\zeta a_n\right)^{1/2}}{\sqrt{n} \left(M^*_h C_\zeta a_n\right)^2} C_{\epsilon}  4R\right)\\
&\lesssim a_n^{1/2} + \frac{C_\epsilon}{\sqrt{n} a_n}.
\end{split}
\end{align}
The proof is complete by combining~\eqref{eq:P_E_dim_g}, \eqref{eq:e_bar_split_g}, ~\eqref{eq:final_b}, and the facts that $C_\eta \lesssim n^{1/q}$, $a_n \rightarrow 0$, and $\log n a_n\rightarrow \infty.$
\end{proof}
\begin{lemma}\label{lem:proof_third_2}
Suppose  the assumptions of Theorem~\ref{thm:Asymp_norml} hold, then 
\begin{equation}\label{eq:proof_third_2}
\left|\G_n \big[\eta_{\theta_0}\big(|{\theta_0}-X|^2\big) - \eta_{\hat\theta}\big(|{\hat\theta}-X|^2\big)\big]\big[X-{\hat\theta}- h_{\hat\theta}\big(|{\hat\theta}-X|\big)\big] \right| = o_p(1).
\end{equation}
\end{lemma}
\begin{proof}
Observe that, as $\|X-{\hat\theta}- h_{\hat\theta}\big(|{\hat\theta}-X|\big)\|_{2, \infty} \le 4 \sqrt{d} T$. Thus
\begin{equation}\label{eq:inf_bound_16}
\left\|\big[\eta_{\theta_0}\big(|{\theta_0}-X|^2\big) - \eta_{\hat\theta}\big(|{\hat\theta}-X|^2\big)\big]\big[X-{\hat\theta}- h_{\hat\theta}\big(|{\hat\theta}-X|\big)\big]\right\|_{2, \infty} \le \sqrt{d} T M_1.
\end{equation}
Using arguments similar to Lemma~F.3 of~\cite{balabdaoui2019score} and Lemma~\ref{lem:deriv_M}, we have that there exists a constant $C$ such that
\begin{equation}\label{eq:eta_theta_smooth}
\sup_{x\in \rchi}|\eta_{\theta_0} (|\theta_0-x|) -\eta_{\theta} (|\theta-x|) | \le C |\theta_0-\theta| \qquad \forall \theta \in B(\theta_0, r).
\end{equation}
Thus, we have that 
\begin{align}\label{eq:L2_boun_1.6}
 \begin{split}
&\left\|\big[\eta_{\theta_0}\big(|{\theta_0}-X|^2\big) - \eta_{\hat\theta}\big(|{\hat\theta}-X|^2\big)\big]\big[X-{\hat\theta}- h_{\hat\theta}\big(|{\hat\theta}-X|\big)\big]\right\|_{2, P_X}\\
 \le{}& \sqrt{d} 2 R \|\eta_{\theta_0}\big(|{\theta_0}-X|^2\big) - \eta_{\hat\theta}\big(|{\hat\theta}-X|^2\big)\|_{P_X}\\
 \le{}& 2R C |\hat\theta-\theta_0|.
 \end{split}
 \end{align}
Let us define 
\begin{align}\label{eq:k_def}
\begin{split}
\mathcal{K}:= \Big\{ f(|\theta-x|) (x- \theta- &g_1(|\theta-x|)+ g_2(|\theta-x|)) : f\in \M^{M_1}, g_1, g_2: \mathcal{D} \to \R^d, \\& g_{i, j} \in \M^{2M_h} \text{ for } i\in \{1, 2\} \text{ and } j\in\{1, \ldots, d\}, \text{ and }   \theta \in B(\theta_0, R).\Big\}
\end{split}
\end{align}
Using arguments similar to the proof of~\eqref{eq:log_ent_H}, we can show that \[
  \big[\eta_{\theta_0}\big(|{\theta_0}-X|^2\big) - \eta_{\hat\theta}\big(|{\hat\theta}-X|^2\big)\big]\big[X-{\hat\theta}- h_{\hat\theta}\big(|{\hat\theta}-X|\big)\big] \in \mathcal{K}.
\]
Furthermore, by stability of Donsker classes, we have that 
\begin{equation}\label{eq:k_ent}
\log N_{[]}(\nu, \mathcal{K}, \|\cdot\|_{2, P_X}) \lesssim \frac{M_1+ d M_h }{\nu}.
\end{equation}
Let us now define $\mathcal{K}_n:= \{ f\in \mathcal{K}: \|f\|_{2, P_X} \le 2RC a_n\}, $ where $\{a_n\}$ be a sequence such that  $a_n\rightarrow 0$, $ a_n \log n\rightarrow \infty$, and $|\hat\theta-\theta_0| =o_p(a_n)$.  By arguments similar to~\eqref{eq:1_4_final}, we have that
\begin{align}
&\P\left(\left|\G_n \big[\eta_{\theta_0}\big(|{\theta_0}-X|^2\big) - \eta_{\hat\theta}\big(|{\hat\theta}-X|^2\big)\big]\big[X-{\hat\theta}- h_{\hat\theta}\big(|{\hat\theta}-X|\big)\big] \right| \ge \delta\right)\label{eq:1.5_final}\\
\le{}&\delta^{-1}\sqrt{d} \sum_{i=1}^{d} \E\left(\sup_{g \in \mathcal{K}_n} \big|\G_n g_i\big| \right)+\P\left(  \big[\eta_{\theta_0}\big(|{\theta_0}-X|^2\big) - \eta_{\hat\theta}\big(|{\hat\theta}-X|^2\big)\big]\big[X-{\hat\theta}- h_{\hat\theta}\big(|{\hat\theta}-X|\big)\big] \notin \mathcal{K}_n\right)\nonumber
\end{align}
However, by~\eqref{eq:L2_boun_1.6} and definition of~$a_n$, we have that 
\begin{equation}\label{eq:not_in_prob}
\P\left(  \big[\eta_{\theta_0}\big(|{\theta_0}-X|^2\big) - \eta_{\hat\theta}\big(|{\hat\theta}-X|^2\big)\big]\big[X-{\hat\theta}- h_{\hat\theta}\big(|{\hat\theta}-X|\big)\big] \notin \mathcal{K}_n\right) =o(1).
\end{equation}
We will now bound the expectation at the right of~\eqref{eq:1.5_final}.   By Lemma 3.4.2 of~\cite{VdVW96}, ~\eqref{eq:inf_bound_16}, and~\eqref{eq:L2_boun_1.6}, we have that 
\begin{align}\label{eq:16_final}
\begin{split}
\E\left(\sup_{g \in \mathcal{K}_n} \big|\G_n g_i\big| \right) \le \left((M_1+ d M_h) R a_n\right)^{1/2} \left(1+\frac{\left((M_1+ d M_h) R a_n\right)^{1/2}}{\sqrt{n}(R a_n)^2 } \sqrt{d} R M_1\right) =o(1),
\end{split}
\end{align}
where the last inequality follows from the definition of $a_n.$
 \end{proof}

\section{Auxiliary lemmas for Appendices~\ref{sec:thm:SSE_consis} and~\ref{sub:asymptotic_normality_of_}} 
\label{sec:auxiliary_lemmas_for_section_ref}

\begin{lemma}\label{lem:Approx_M}
Suppose assumptions~\ref{a1}--\ref{assum:err_mom} hold, then 
\begin{equation}\label{eq:Approx_M}
\sup_{\theta \in B(\theta_0, r)} \big|\MM_n(\theta) - M(\theta)\big| = o_p(1).
\end{equation}
\end{lemma}
\begin{proof}
For any $\theta\in B(\theta_0, r)$, we have
\begin{align}\label{eq:Approx_1}
\begin{split}
\MM_n(\theta) &=  n^{-1}\sum_{i=1}^{n}\left[Y_i - \eta_\theta\big(|\theta-X_i|^2\big)\right]({X_i-\theta}) + n^{-1}\sum_{i=1}^{n}\left[ \eta_\theta\big(|\theta-X_i|^2\big) - \tilde\eta_\theta\big(|\theta-X_i|^2\big)\right]({X_i-\theta})\\
&= M(\theta)+  \int \left[y - \eta_\theta\circ\theta(x)\right](x-\theta) d(\P_n- P_0)(x, y)\\
&\qquad + \int \left[ \eta_\theta\circ\theta(x)-\tilde \eta_\theta\circ\theta(x)\right](x-\theta) d(\P_n- P_0)(x, y)\\
&\qquad + \int \left[ \eta_\theta\circ\theta(x)-\tilde \eta_\theta\circ\theta(x)\right](x-\theta)  dP_0(x, y)\\
\end{split}
\end{align}
By~\eqref{eq:unif_eta} of Theorem~\ref{thm:unif_eta}, we have 
\[
\sup_{\theta\in B(\theta_0, \delta_0)} \int \left[ \eta_\theta\circ\theta(x)-\tilde \eta_\theta\circ\theta(x)\right](x-\theta)  dP_X = O_p(n^{-1/3} n^{1/q})
\]
In the following we will now prove that 
\begin{align}
 \sup_{\theta\in B(\theta_0, r)} \left|\int \left[y - \eta_\theta\circ\theta(x)\right](x-\theta) d(\P_n- P_0)(x, y)\right| &= O_p(n^{-1/2})\label{eq:extraterms_approx1}\\
\sup_{\theta\in B(\theta_0, r)} \left| \int \left[ \eta_\theta\circ\theta(x)-\tilde \eta_\theta\circ\theta(x)\right](x-\theta) d(\P_n- P_0)(x, y)\right| &=O_p( n^{-4/6}n^{1/q})\label{eq:extraterms_approx2}
\end{align}
First, observe that 
\begin{align}\label{eq:extra_1}
\begin{split}
&\sup_{\theta\in B(\theta_0, r)} \left|\int \left[y - \eta_\theta\circ\theta(x)\right](x-\theta) d(\P_n- P_0)(x, y)\right|\\
\le{}&  \sup_{\theta\in B(\theta_0, r)} \left|\int \left[\eta_0\circ\theta_0(x) - \eta_\theta\circ\theta(x)\right](x-\theta) d(\P_n- P_0)(x, y)\right| + \sup_{\theta\in B(\theta_0, r)} \left|\int \epsilon(x-\theta) d(\P_n- P_0)(x, y)\right| \\
 \le{}&  \sup_{\theta\in B(\theta_0, r)} \left|\int \left[\eta_0\circ\theta_0(x) - \eta_\theta\circ\theta(x)\right](x-\theta) d(\P_n- P_0)(x, y)\right| +  \left| \P_n\epsilon x \right| + \left| \P_n\epsilon  \right| \sup_{\theta\in B(\theta_0, r)} |\theta| \\ 
  \le{}&  \sup_{\theta\in B(\theta_0, r)} \left|\int \left[\eta_0\circ\theta_0(x) - \eta_\theta\circ\theta(x)\right](x-\theta) d(\P_n- P_0)(x, y)\right| + O_p(n^{-1}) \\
  ={}&\frac{1}{\sqrt{n}}\sup_{\theta\in B(\theta_0, r)} \left| \G_n \big[\big(\eta_0\circ\theta_0(x) - \eta_\theta\circ\theta(x)\big)(x-\theta)\big] \right| + O_p(n^{-1}).
\end{split}
\end{align}
The proof of~\eqref{eq:extraterms_approx1}, will be complete if we can show that $\sup_{\theta\in B(\theta_0, r)} \left| \G_n \big[\big(\eta_0\circ\theta_0(x) - \eta_\theta\circ\theta(x)\big)(x-\theta)\big] \right| =O_p(1)$. We show this next. Note that 
\[ \sup_{\theta\in B(\theta_0, r)} \big\|\eta_0\circ\theta_0 - \eta_\theta\circ\theta\big\|\le \sup_{\theta\in B(\theta_0, r)} \big\|\eta_0\circ\theta_0 - \eta_\theta\circ\theta\big\|_{\infty} \le M_1,
\] and by Lemma~\ref{lem:basic_joint}, we have that 
\[\log N_{[\, ]}(\nu, \{\eta_0\circ\theta_0 -\eta_\theta\circ\theta : \theta \in B(\theta_0, r)\}, \|\cdot \|)\le  \frac{A_0M_1}{\nu}, \]
where $A_0$ is a constant depending only on $\rchi$. Now by Theorem~3.4.2 of~\cite{VdVW96}, we have 
\begin{align}\label{eq:extra_2}
\begin{split}
&\E\left(\sup_{\theta\in B(\theta_0, r)}\left| \G_n \big[\big(\eta_0\circ\theta_0(x) - \eta_\theta\circ\theta(x)\big)(x-\theta)\big] \right| \right)\\
\le{}& \sum_{j=1}^{d} \E\left(\sup_{\theta\in B(\theta_0, r)}\left| \G_n \big[\big(\eta_0\circ\theta_0(x) - \eta_\theta\circ\theta(x)\big)(x_j-\theta_j)\big] \right| \right)\\
\le{}& d \sqrt{A_0} M_1^{3/2} \bigg(1+ \frac{\sqrt{A_0} M_1^{3/2} M_1}{M_1^2\sqrt{n}}\bigg) \le d \sqrt{A_0} M_1^2
\end{split}
\end{align}
We will now establish~\eqref{eq:extraterms_approx2}.  Fix any $\delta>0$. Let $C_1$ and $C_2$ be constants such that \[
\P\left( \sup_{\theta \in B(\theta_0, \delta_0)} \|\tilde{\eta}_\theta\|_\infty \ge Cn^{1/q}\right) \le  \delta/4.
\]
and 
\[\P\bigg( \sup_{\theta\in B(\theta_0, \delta_0)} \int \{\tilde{\eta}_\theta(|\theta-x|^2) -\eta_\theta(|\theta-x|^2)\}^2 dP_X(x) \ge C_2 n^{-2/3} n^ {2/q}\bigg) \le \delta/4.\]
Note that such constants exist by Theorem~\ref{thm:unif_eta}. Now consider the following class of functions
\begin{align}\label{eq:A_approx}
\begin{split}
\B_1(\gamma):=\bigg\{\eta\circ\theta &-\eta_\theta\circ\theta: \eta\in \M^{C_1 n^{1/q}}, \; \theta \in  B(\theta_0, r), \\
&\text{ and }\int \{\tilde{\eta}_\theta(|\theta-x|^2) -\eta_\theta(|\theta-x|^2)\}^2 dG(x)\le \gamma \bigg\}.
\end{split}
\end{align}
Then by Lemma~\ref{lem:basic_joint}, we have that 
\[
  \log N_{[\, ]}(\nu, \B_1(\nu), \|\cdot\|) \le \frac{2A(M_1+ C_1 n^{1/q})}{\nu}.
\]
Observe that 
\begin{align}\label{eq:part_2_approx}
\begin{split}
&\P\bigg(n^{1/2}\sup_{\theta\in B(\theta_0, r)} \left| \int \left[ \eta_\theta\circ\theta(x)-\tilde \eta_\theta\circ\theta(x)\right](x-\theta) d(\P_n- P_0)(x, y)\right| \ge C_3 n^{1/q} n^{-1/6}\bigg) \\
\le{}&\P\bigg(\sup_{\theta\in B(\theta_0, r)} \left| \G_n  \left[ \big(\eta_\theta\circ\theta(x)-\tilde \eta_\theta\circ\theta(x)\big)(x-\theta)\right] \right| \ge C_3n^{1/q} n^{-1/6}, \sup_{\theta \in B(\theta_0, \delta_0)} \|\tilde{\eta}_\theta\|_\infty \le Cn^{1/q}, \\
 &\qquad \qquad\qquad \text{ and } \sup_{\theta\in B(\theta_0, \delta_0)} \int \{\tilde{\eta}_\theta(|\theta-x|^2) -\eta_\theta(|\theta-x|^2)\}^2 dG(x) \le C_2 n^{-1/3} n^ {1/q}\bigg) +\delta/2\\
 \le{}&\P\bigg(\sup_{f_1\in \B_1(C_2 n^{-2/3} n^ {2/q})} \left| \G_n  \left[ (x-\theta) f\right] \right| \ge C_3n^{1/q} n^{-1/6}\bigg) +\delta/2\\
 \le{}&\frac{1}{C_3n^{1/q} n^{-1/6}}\E\bigg(\sup_{f_1\in \B_1(C_2 n^{-1/3} n^ {1/q})} \left| \G_n  \left[ (x-\theta) f\right] \right|\bigg) +\delta/2\\
 \le{}& \frac{1}{C_3n^{1/q} n^{-1/6}}\sum_{j=1}^{d}\E\bigg(\sup_{f_1\in \B_1(C_2 n^{-1/3} n^ {1/q})} \left| \G_n  \left[ (x_j-\theta_j) \right] \right|\bigg) +\delta/2
\end{split}
\end{align}
Moreover, by Theorem~3.4.2 of~\cite{VdVW96}, we have that
\begin{align}\label{eq:112789}
\begin{split}
&\E\bigg(\sup_{f_1\in \B_1(C_2 n^{-1/3} n^ {1/q})} \left| \G_n  \left[ (x_j-\theta_j) \right] \right|\bigg) \\\le{}& J_{[\, ]}(C_2 n^{-1/3} n^ {1/q}, \B_1, \|\cdot\|) \bigg(1+ \frac{J_{[\, ]}(C_2 n^{-1/3} n^ {1/q}, \B_1, \|\cdot\|)}{\sqrt{n} [C_2 n^{-1/3} n^ {1/q}]^2} C_1 n^{1/q}\bigg)\\
\lesssim{}& \sqrt{n^{2/q}n^{-1/3} }\bigg(1+ \frac{\sqrt{n^{2/q}n^{-1/3} } n^{1/q}}{\sqrt{n} n^{2/q} n^{-2/3}}\bigg)\\
\lesssim{}& \sqrt{n^{2/q}n^{-1/3} }  + \frac{n^{3/q}n^{-1/3}  }{\sqrt{n} n^{2/q} n^{-2/3}} \\
 \lesssim{}& n^{ 1/q} n^{-1/6}.
\end{split}
\end{align}
We have now proved~\eqref{eq:extraterms_approx2}, as 
\begin{align}\label{eq:tt1}
\begin{split}
&\P\bigg(\sup_{\theta\in B(\theta_0, r)} \left| \int \left[ \eta_\theta\circ\theta(x)-\tilde \eta_\theta\circ\theta(x)\right](x-\theta) d(\P_n- P_0)(x, y)\right| \ge C_3 n^{1/q} n^{-4/6}\bigg) \\
\lesssim{}& \frac{1}{C_3n^{1/q} n^{-1/6}} n^{1/q} n^{-1/6} \le \frac{1}{C_3}.
\end{split}
\end{align}
\end{proof}

\subsection{Property of $M(\theta)$} 
\label{sec:PropertyofM}
 The following two lemmas establish some properties of $\theta \mapsto M(\theta).$
Recall that 
\begin{equation}
M(\theta):=\int_{\D}\left[Y - \eta_\theta\big(|\theta-X|^2\big)\right]({X-\theta}) dP_X,
\end{equation}
is the population version of$\;\;\mathbb{M}_n(\cdot)$ and 
\[{\eta}_\theta(u):= \E\big(\eta_0(|\theta_0-X|^2) \big| |\theta-X|^2 = u\big).\]
\begin{lemma}\label{lem:zeroCrossing} It is easy to see that $|M(\theta_0)|=0$.
Suppose assumptions~\ref{a1}--\ref{assum:err_mom}, and \ref{assum:NonzeroEverywhere} hold, then
\begin{enumerate}
\item For all $\theta\in B(\theta_0, \delta)$, we have that  $(\theta-\theta_0)^\top M(\theta) \ge 0.$

\item There does not exist $\theta_1 \neq \theta_0$ such  that $(\theta-\theta_1)^\top M(\theta) \ge 0$ for all $\theta\in B(\theta_0, \delta)$.

\end{enumerate}

\end{lemma}

\begin{proof}
\textbf{Proof of (1):} The proof here is similar to Proof of Lemma F.2~\cite{balabdaoui2019score}. By definition of $\eta_\theta(\cdot)$ (see~\eqref{eq:eta_profile_pop}), we have 
 \begin{align}\label{eq:m_cov}
\begin{split}
M(\theta)={}&\int_{\D}(x-\theta) \left[\eta_0\big(|\theta_0-x|^2\big) - \eta_\theta\big(|\theta-x|^2\big)\right] dP_X(x) \\
={}&\int_{\D}(x-\theta) \left[\eta_0\big(|\theta_0-x|^2\big) - \E\big(\eta_0(|\theta_0-X|^2) \big| |\theta-X|^2 =|\theta-x|^2\big)\right] dP_X(x) \\
={}&\E\left[\text{Cov}\Big({X-\theta}, \eta_0\big(|\theta_0-X|^2\big)\big| |\theta-X|^2\Big)\right].
\end{split}
\end{align}
Thus
\begin{align}\label{eq:proof_1}
\begin{split}
(\theta- \theta_0)^\top M(\theta) &= \E\left[\text{Cov}\Big((\theta- \theta_0)^\top ({X-\theta}), \eta_0\big(|\theta_0-X|^2\big)\big| |\theta-X|^2\Big)\right] \\
&= \E\left[\text{Cov}\Big((\theta- \theta_0)^\top ({X-\theta}), \eta_0\big(|{X-\theta}|^2 + |\theta-\theta_0|^2 + 2(\theta-\theta_0)^\top ({X-\theta}) \big)\big| |\theta-X|^2\Big)\right].
\end{split}
\end{align}
We will next show that $\text{Cov}\Big((\theta- \theta_0)^\top ({X-\theta}), \eta_0\big(|{X-\theta}|^2 + |\theta-\theta_0|^2 + 2(\theta-\theta_0)^\top ({X-\theta}) \big)\big| |\theta-X|^2\Big)$ is positive because $\eta_0$ is increasing.  Define $Z_1= (\theta- \theta_0)^\top ({X-\theta})$ and $Z_2= \eta_0(u+ Z_1)$. Let $\tilde{z_2} := \eta_0^{{-1}}(z_2)- u,$ then by monotonicity of $\eta_0$, we have that 
\begin{align}\label{eq:Prob_ineq}
\begin{split}
\P(Z_1 \ge z_1, Z_2\ge z_2)= \P\big(Z_1 \ge\max(z_1, \tilde{z_2})\big) &\ge  \P\big(Z_1 \ge \max(z_1, \tilde{z_2})\big) \P\big(Z_1 \ge\min(z_1, \tilde{z_2})\big) \\
& = \P(Z_1\ge z_1) \P(Z_2 \ge z_2).
\end{split}
\end{align}
 Thus for  all $\theta \in B(\theta_0, \delta), $ we have
\begin{align}\label{eq:16908}
\begin{split}
&\text{Cov}\Big((\theta- \theta_0)^\top ({X-\theta}), \eta_0\big(|{X-\theta}|^2 + |\theta-\theta_0|^2 + 2(\theta-\theta_0)^\top ({X-\theta}) \big)\big| |\theta-X|^2 = u\Big)\\
={}& \int \Big[\P(Z_1 \ge z_1, Z_2\ge z_2)- \P(Z_1\ge z_1) \P(Z_2 \ge z_2)\Big] dz_1 dz_2\\
\ge{}&0 
\end{split}
\end{align}
\textbf{Proof of (2):}  Suppose there exists $\theta_1$ such that $(\theta- \theta_1)^\top M(\theta) \ge 0$ for all $\theta \in B(\theta_0, \delta).$ Take $\theta'= (\theta_1 +\theta_0)/2$, then $(\theta'-\theta_0)^\top M(\theta') = -(\theta' -\theta_1) M(\theta')$. A contradiction since by Assumption~\ref{assum:NonzeroEverywhere}, we have that 
\[\text{Cov}\Big((\theta- \theta_0)^\top ({X-\theta}), \eta_0\big(|\theta_0-X|^2\big)\big| |\theta-X|^2\Big)\neq 0 \quad \text{almost everywhere}. \qedhere\] 
\end{proof}
\todo[inline]{\cite{balabdaoui2019score} assume ~\ref{assum:NonzeroEverywhere}, it is their A7. }

\begin{lemma}\label{lem:deriv_M}
Under assumptions~\ref{a1},~\ref{eta_bound},~\ref{assum:err_mom}, and \ref{assum:NonzeroEverywhere}, we have 
\begin{equation}\label{eq:eta_deriv}
\left. \frac{\partial \eta_\theta(|\theta-x|^2)}{\partial \theta}\right\vert_{\theta=\theta_0} = \Big( \E\big(X| |\theta_0-X|^2= |\theta_0-x|^2\big) -X\Big) \eta_0'(|\theta_0-x|^2)
\end{equation}
and 
\begin{equation}\label{eq:M_prime}
M'(\theta_0)= \E\Big(\eta'_0(|\theta_0-X|^2) \text{Cov}\big(X\big||\theta_0-X|^2 \big)\Big)
\end{equation}
\end{lemma}
\begin{proof}
Let $h^\theta_+(\cdot|u)$ be the conditional density of $( X_2, \ldots, X_d)$ when $|{X-\theta}| =u$ and $\sign(X_1-\theta_1) =1$ and $h^\theta_-(\cdot|u)$be the conditional density of $( X_2, \ldots, X_d)$ when $|{X-\theta}| =u$ and $\sign(X_1-\theta_1) =-1$. Further, let $p^{\theta}_+ =\P(\sign(X_1-\theta_1) =1)$ and $p^{\theta}_-= 1-p^{\theta}_+$]. In this proof, we use the following notation $\theta= (\theta_1, \ldots, \theta_d)$ and $\theta_0= (\theta_{01}, \ldots, \theta_{0d})$. Then 
\begin{align}\label{eq:deriv_11}
\begin{split}
\eta_\theta(|\theta-x|^2) &=\E\big(\eta_0(|\theta_0-X|^2)\big| |\theta-X|^2 =|\theta-x|^2\big)\\
&= p^\theta_+ \E\big(\eta_0(|\theta_0-X|^2)\big| |\theta-X|^2 =|\theta-x|^2, \sign(X_1-\theta_1) =1\big) \\
&\qquad \quad+ p^\theta_-\E\big(\eta_0(|\theta_0-X|^2)\big| |\theta-X|^2 =|\theta-x|^2, \sign(X_1-\theta_1) = -1\big)\\
\end{split}
\end{align}
and
\begin{align}\label{eq:deriv_1}
\begin{split}
\frac{\partial}{\partial \theta_j} \eta_\theta(|\theta-x|^2)\bigg\vert_{\theta=\theta_0}
&= \frac{\partial}{\partial \theta_j} p^\theta_+ \bigg\vert_{\theta=\theta_0} \E\big(\eta_0(|\theta_0-X|^2)\big| |\theta_0-X|^2 =|\theta_0-x|^2, \sign(X_1-\theta_{01}) =1\big) \\
&\qquad \quad+ \frac{\partial}{\partial \theta_j}p^\theta_- \bigg\vert_{\theta=\theta_0}\E\big(\eta_0(|\theta_0-X|^2)\big| |\theta_0-X|^2 =|\theta_0-x|^2, \sign(X_1-\theta_{01}) = -1\big)\\
&\qquad\quad+  p^\theta_+ \frac{\partial}{\partial \theta_j}\E\big(\eta_0(|\theta_0-X|^2)\big| |\theta-X|^2 =|\theta-x|^2, \sign(X_1-\theta_1) =1\big)\bigg\vert_{\theta=\theta_0} \\
&\qquad \quad+ p^\theta_- \frac{\partial}{\partial \theta_j}\E\big(\eta_0(|\theta_0-X|^2)\big| |\theta-X|^2 =|\theta-x|^2, \sign(X_1-\theta_1) = -1\big)\bigg\vert_{\theta=\theta_0}\\
&= \frac{\partial}{\partial \theta_j} p^\theta_+ \bigg\vert_{\theta=\theta_0} \eta_0(|\theta_0-x|^2) +\frac{\partial}{\partial \theta_j}p^\theta_- \bigg\vert_{\theta=\theta_0}\eta_0(|\theta_0-x|^2)\\
&\qquad\quad+  p^\theta_+ \frac{\partial}{\partial \theta_j}\E\big(\eta_0(|\theta_0-X|^2)\big| |\theta-X|^2 =|\theta-x|^2, \sign(X_1-\theta_1) =1\big)\bigg\vert_{\theta=\theta_0} \\
&\qquad \quad+ p^\theta_- \frac{\partial}{\partial \theta_j}\E\big(\eta_0(|\theta_0-X|^2)\big| |\theta-X|^2 =|\theta-x|^2, \sign(X_1-\theta_1) = -1\big)\bigg\vert_{\theta=\theta_0}\\
&=   p^\theta_+ \frac{\partial}{\partial \theta_j}\E\big(\eta_0(|\theta_0-X|^2)\big| |\theta-X|^2 =|\theta-x|^2, \sign(X_1-\theta_1) =1\big)\bigg\vert_{\theta=\theta_0} \\
&\qquad \quad+ p^\theta_- \frac{\partial}{\partial \theta_j}\E\big(\eta_0(|\theta_0-X|^2)\big| |\theta-X|^2 =|\theta-x|^2, \sign(X_1-\theta_1) = -1\big)\bigg\vert_{\theta=\theta_0},
\end{split}
\end{align}
for all $j =1, \ldots, d.$ Note that 
\begin{align}\label{eq:theta_splt}
\begin{split}
|\theta_0-X|^2 &= (X_1-\theta_{01})^2 +\sum_{k=2}^{d} (X_k-\theta_{0k})^2\\
&= |X_1-\theta_1|^2 + (\theta_1-\theta_{01})^2 + 2(X_1-\theta_1) (\theta_1-\theta_{01})+\sum_{k=2}^{d} (X_k-\theta_{0k})^2\\
&= |{X-\theta}|^2 -\sum_{k=2}^{d} (X_k-\theta_{k})^2+ 2\sign(X_1-\theta_1)(\theta_1-\theta_{01})\sqrt{|{X-\theta}|^2 -\sum_{k=2}^{d} (X_k-\theta_{k})^2} \\
&\qquad\quad  +(\theta_1-\theta_{01})^2 + \sum_{k=2}^{d} (X_k-\theta_{0k})^2.
\end{split}
\end{align}
Thus 
\begin{align}
&\E\big(\eta_0(|\theta_0-X|^2)\big| |\theta-X|^2 =|\theta-x|^2, \sign(X_1-\theta_1) =1\big)\\
={}& \int_{\sign(\tilde{x}_1-\theta_1) =1} \eta_0(|\theta_0-X|^2) h^\theta_+(\tilde{x}_2, \ldots, \tilde{x}_d | |\theta-x| ) d\tilde{x}_2\ldots d\tilde{x}_d\label{eq:+_part}\\
={}& \int_{\sign(\tilde{x}_1-\theta_1) =1} \eta_0\Bigg(|x-\theta|^2 -\sum_{k=2}^{d} (\tilde{x}_k-\theta_{k})^2+ 2(\theta_1-\theta_{01})\sqrt{|x-\theta|^2 -\sum_{k=2}^{d} (\tilde{x}_k-\theta_{k})^2} \\
&\qquad\qquad\qquad\qquad\qquad\qquad +(\theta_1-\theta_{01})^2 + \sum_{k=2}^{d} (\tilde{x}_k-\theta_{0k})^2\Bigg) h^\theta_+(\tilde{x}_2, \ldots, \tilde{x}_d | |\theta-x|) d\tilde{x}_2\ldots d\tilde{x}_d.
\end{align}
Define, 
\begin{equation}\label{eq:hloc_def}
h\big(x, \theta, \{\tilde{x}\}_2^d\big):= |x-\theta|^2 -\sum_{k=2}^{d} (\tilde{x}_k-\theta_{k})^2+ 2(\theta_1-\theta_{01})\sqrt{|x-\theta|^2 -\sum_{k=2}^{d} (\tilde{x}_k-\theta_{k})^2}+ (\theta_1-\theta_{01})^2 + \sum_{k=2}^{d} (\tilde{x}_k-\theta_{0k})^2.
\end{equation}
Note that for $j\ge 2$, we have 
\begin{equation}\label{eq:h_loc_deriv}
\frac{\partial}{\partial \theta_j} h\big(x, \theta, \{\tilde{x}\}_2^d\big)= 2(\tilde{x}_j -x_j)\Big\{ 1+ \frac{ \theta_1-\theta_{01} }{\sqrt{|x-\theta|^2 -\sum_{k=2}^{d} (\tilde{x}_k-\theta_{k})^2}}\Big\}.
\end{equation}
For $j=2, \ldots, d$, we have that 
\begin{align}\label{eq:+part_deriv}
\begin{split}
&\frac{\partial\E\big(\eta_0(|\theta_0-X|^2)\big| |\theta-X|^2 =|\theta-x|^2, \sign(X_1-\theta_1) =1\big)}{\partial \theta_j}\\
={}&\int_{\sign(\tilde{x}_1-\theta_1) =1} \eta_0\Bigg(h\big(x, \theta, \{\tilde{x}\}_2^d\big)\Bigg) \frac{\partial}{\partial \theta_j} h^\theta_+(\tilde{x}_2, \ldots, \tilde{x}_d | |\theta-x|) d\tilde{x}_2\ldots d\tilde{x}_d\\
 &+\int_{\sign(\tilde{x}_1-\theta_1) =1}  \eta_0'\Bigg(h\big(x, \theta, \{\tilde{x}\}_2^d\big)\Bigg)  2(\tilde{x}_j -x_j)\Big\{ 1+ \frac{ \theta_1-\theta_{01} }{\sqrt{|x-\theta|^2 -\sum_{k=2}^{d} (\tilde{x}_k-\theta_{k})^2}}\Big\}\\
 &\hspace{3.5in}h^\theta_+(\tilde{x}_2, \ldots, \tilde{x}_d | |\theta-x|) d\tilde{x}_2\ldots d\tilde{x}_d\\
\end{split}
\end{align}
Thus
\begin{align}\label{eq:+part_derivattheta0}
\begin{split}
&\left.\frac{\partial\E\big(\eta_0(|\theta_0-X|^2)\big| |\theta-X|^2 =|\theta-x|^2, \sign(X_1-\theta_1) =1\big)}{\partial \theta_j}\right\vert_{\theta=\theta_0}\\
={}&\int_{\sign(\tilde{x}_1-\theta_{01}) =1}  \eta_0'\big(|x-\theta_0|^2 \big)  2(\tilde{x}_j -x_j)h^{\theta_0}_+(\tilde{x}_2, \ldots, \tilde{x}_d | |\theta_0-x|) d\tilde{x}_2\ldots d\tilde{x}_d\\
={}& \eta_0'\big(|x-\theta_0|^2 \big)     2\Big(\E\big({X}_j\big| \sign(X_1-\theta_0) =1,  |\theta_0-X| =|\theta-x|\big)-x_j\Big)
\end{split}
\end{align}
Similarly, we can show that 
\begin{align}\label{eq:-part_derivattheta0}
\begin{split}
&\left.\frac{\partial\E\big(\eta_0(|\theta_0-X|^2)\big| |\theta-X|^2 =|\theta-x|^2, \sign(X_1-\theta_1) =-1\big)}{\partial \theta_j}\right\vert_{\theta=\theta_0}\\
={}& \eta_0'\big(|x-\theta_0|^2 \big)     2\Big(\E\big({X}_j\big| \sign(X_1-\theta_0) =-1,  |\theta_0-X| =|\theta-x|\big)-x_j\Big)
\end{split}
\end{align}
Combining~\eqref{eq:deriv_1} with ~\eqref{eq:+part_derivattheta0} and~\eqref{eq:-part_derivattheta0}, we get that
\begin{equation}\label{eq:j>2}
\frac{\partial}{\partial \theta_j} \eta_\theta(|\theta-x|^2)\bigg\vert_{\theta=\theta_0}
=  \eta_0'\big(|x-\theta_0|^2 \big)     2\Big(\E\big({X}_j\big|  |\theta_0-X| =|\theta-x|\big)-x_j\Big),
\end{equation}
for all $j\ge 2$. We will now compute~\eqref{eq:deriv_1} for $j=1.$ By~\eqref{eq:+_part}, we have that 

\begin{align}\label{eq:+part_deriv1}
\begin{split}
&\left.\frac{\partial\E\big(\eta_0(|\theta_0-X|^2)\big| |\theta-X|^2 =|\theta-x|^2, \sign(X_1-\theta_1) =1\big)}{\partial \theta_1}\right\vert_{\theta=\theta_0}\\
={}&  \int  \eta_0'\Bigg(h\big(x, \theta, \{\tilde{x}\}_2^d\big)\Bigg) \Big\{2\sqrt{|x-\theta_0|^2 -\sum_{k=2}^{d} (\tilde{x}_k-\theta_{0k})^2} -2x_1 +2\theta_{01}\Big\}h^\theta_+(\tilde{x}_2, \ldots, \tilde{x}_d | |\theta_0-x|) d\tilde{x}_2\ldots d\tilde{x}_d\\
&\qquad +  \int \eta_0\Bigg(h\big(x, \theta_0, \{\tilde{x}\}_2^d\big)\Bigg) \frac{\partial}{\partial \theta_1}\left. h^\theta_+(\tilde{x}_2, \ldots, \tilde{x}_d | |\theta-x|)\right\vert_{\theta=\theta_0} d\tilde{x}_2\ldots d\tilde{x}_d\\
=& \eta_0'\big(|x-\theta_0|^2 \big)     2\Big(\E\big({X}_1\big| \sign(X_1-\theta_0) =1,  |\theta_0-X| =|\theta-x|\big)-x_1\Big).
\end{split}
\end{align}
The other case can be solved similarly. Combining the above results, we have~\eqref{eq:eta_deriv} and 
\begin{align}\label{eq:mprime_proof}
M'(\theta_0)&= \E\Bigg( -\left. \frac{\partial \eta_\theta(|\theta-x|^2)}{\partial \theta}\right\vert_{\theta=\theta_0} ({X-\theta_0})\Bigg)=\E\Big(\eta'_0(|\theta_0-X|^2) \text{Cov}\big(X\big||\theta_0-X|^2 \big)\Big). \qedhere
\end{align}
\end{proof}

\bibliographystyle{apalike}
\bibliography{SigNoise}

\end{document}

%% file: header.tex
\usepackage{amsfonts}
\usepackage{epsfig,amsmath,amssymb,amsthm,amsbsy, verbatim,enumitem}
\usepackage{xcolor}
\usepackage[colorlinks,citecolor=blue,urlcolor=blue,linkcolor=blue,naturalnames=false]{hyperref}
\usepackage{graphicx,graphics, rotating,array,booktabs,bm,bbm,bigints}

\newtheorem{lemma}{Lemma}[section]

\newtheorem{theorem}{Theorem}[section]
\newtheorem{remark}{Remark}[section]
\newtheorem{defn}{Definition}[section]
\usepackage{mathtools}
\mathtoolsset{showonlyrefs}

\newcommand{\E}{{\mathbb E}}

\newcommand{\D}{{\mathcal{D}}}
\newcommand{\Q}{{\mathbb Q}}
\newcommand{\M}{{\mathcal{M}}}
\newcommand{\MM}{{\mathbb{M}}}

\newcommand{\R}{{\mathbb R}}

\renewcommand{\P}{{\mathbb P}}

\newcommand{\A}{{\mathcal{A}}}
\newcommand{\B}{{\mathcal{B}}}

\newcommand{\G}{{\mathbb{G}}}
\newcommand{\g}{{\mathcal{G}}}

\newcommand{\F}{{\mathcal F}}
\newcommand{\h}{{\mathcal H}}

\DeclareMathOperator{\sign}{sign}

\newcommand{\argmax}{\operatornamewithlimits{\textrm{argmax}}}
\newcommand{\argmin}{\operatornamewithlimits{\textrm{argmin}}}

\definecolor{lgray}{gray}{0.70}

\DeclarePairedDelimiter{\floor}{\lfloor}{\rfloor}

\DeclareRobustCommand{\rchi}{{\mathpalette\irchi\relax}}
\newcommand{\irchi}[2]{\raisebox{\depth}{$#1\chi$}}